\documentclass[sigplan,screen]{acmart}

\usepackage{listproc}
\usepackage{mathpartir}
\usepackage{ottalt}

\newenvironment{dmGLdefnblock}[3][]{ \framebox{\mbox{#2}} \quad #3 \\[0pt]}{}

\newcommand{\dmGLnt}[1]{\mathit{#1}}
\newcommand{\dmGLmv}[1]{\mathit{#1}}

\newcommand{\dmGLsym}[1]{#1}

\newcommand{\dmGLdrulename}[1]{\textsc{#1}}

        \newcommand{\sfoperator}[1]{\operatorname{\mathsf{#1} } }

\newcommand{\dmGLdruleGXXsubusageName}[0]{\dmGLdrulename{G\_subusage}}

\newcommand{\dmGLdruleGXXweakName}[0]{\dmGLdrulename{G\_weak}}

\newcommand{\dmGLdruleGXXconvertName}[0]{\dmGLdrulename{G\_convert}}

\newcommand{\dmGLdruleGXXunitElimName}[0]{\dmGLdrulename{G\_unitElim}}

\newcommand{\dmGLdruleGXXradjIntroName}[0]{\dmGLdrulename{G\_radjIntro}}

\newcommand{\dmGLdruleMXXidName}[0]{\dmGLdrulename{M\_id}}

\newcommand{\dmGLdruleMXXsubusageName}[0]{\dmGLdrulename{M\_subusage}}

\newcommand{\dmGLdruleMXXweakName}[0]{\dmGLdrulename{M\_weak}}

\newcommand{\dmGLdruleMXXexchangeName}[0]{\dmGLdrulename{M\_exchange}}

\newcommand{\dmGLdruleMXXladjIntroName}[0]{\dmGLdrulename{M\_ladjIntro}}

\newcommand{\dmGLdruleMXXladjElimName}[0]{\dmGLdrulename{M\_ladjElim}}

\newcommand{\dmGLdruleMXXradjElimName}[0]{\dmGLdrulename{M\_radjElim}}

  \renewottcommands[dmGL]

\newenvironment{daldefnblock}[3][]{ \framebox{\mbox{#2}} \quad #3 \\[0pt]}{}

\newcommand{\dalnt}[1]{\mathit{#1}}
\newcommand{\dalmv}[1]{\mathit{#1}}

\newcommand{\dalsym}[1]{#1}

		\newcommand*{\doubleplus}{\null+\kern-.8em+\kern0.3em\null}

  \renewottcommands[dal]

\usepackage{amsmath}
\usepackage{amsthm}
\usepackage{thm-restate}
\usepackage[nameinlink]{cleveref}
\usepackage{mathtools}
\usepackage{graphicx}

\usepackage{enumitem}
\usepackage{mdframed}

\theoremstyle{plain}
\newtheorem{theorem}{Theorem}[section]
\newtheorem{lem}[theorem]{Lemma}

\newtheorem{cor}[theorem]{Corollary}

\theoremstyle{definition}
\newtheorem{defn}[theorem]{Definition}

\newtheorem{exmp}[theorem]{Example}

\newtheorem{rem}[theorem]{Remark}

\newtheorem{conv}[theorem]{Convention}

\newcommand*{\mb}{\mathbf}
\newcommand*{\mc}{\mathcal}
\newcommand*{\mf}{\mathfrak}
\newcommand*{\ms}{\mathsf}
\newcommand*{\mbb}{\mathbb}
\newcommand*{\N}{\mbb N}

\newcommand*{\Q}{\mbb Q}

\newcommand*{\nat}{\mb {Nat}}

\usepackage{tikz}
\usetikzlibrary{cd}     % tikzcd -- package for commutative diagrams
\usetikzlibrary{arrows.meta,decorations.pathmorphing,backgrounds,positioning,fit,petri}
\tikzcdset
   { arrow style = tikz
   , diagrams = {> = {Straight Barb[scale = 0.8]}}
   } % change the style of arrows to look good with most fonts.

\usepackage{todonotes}
\newcommand*{\dmGL}{{\normalfont\textsc{dmGL}}}
\newcommand*{\proves}{\vdash}
\newcommand*{\from}{\colon}
\newcommand*{\at}{\odot}
\DeclareMathOperator{\Weak}{\ms{Weak}}

\newcommand*{\Modes}{\ms{Modes}}

\newcommand*{\True}{\ms{True}}
\newcommand*{\False}{\ms{False}}

\newcommand*{\covar}{{\uparrow\uparrow}}
\newcommand*{\invar}{{\sim\sim}}
\newcommand*{\contvar}{{\downarrow\downarrow}}
\newcommand*{\unvar}{{??}}

\newcommand{\LNLD}{\ensuremath{\text{LNL}_{\text{D}}}}
\newcommand{\GraD}{\textsc{Grad}}
\newcommand{\Grad}{\GraD}
\newcommand{\GlaD}{\textsc{Glad}}
\newcommand{\glad}{\GlaD}
\newcommand*{\Cont}{\ms{Cont}}

\usepackage{pgffor}
\makeatletter

\newcommand*{\namedrule}[2][\relax]{%
  \def\pr@fix{#1XX}%
  \def\n@me{#2}%
  \edef\ou@t{\expandafter\noexpand\csname dmGLdrule\pr@fix\n@me\endcsname}%
  \ou@t{}
}

\newcommand*{\drulename}[2][\relax]{%
  \def\pr@fix{#1XX}%
  \def\n@me{#2Name}%
  \edef\ou@t{\expandafter\noexpand\csname dmGLdrule\pr@fix\n@me\endcsname}%
  \ou@t{}\hspace{-0.25em}%
}

\newcommand*{\drulenames}[2][\relax]{%
  \foreach \@iter in {#2}{%
    \drulename[#1]{\@iter},
  }}

\newcommand*{\pprules}[2][\relax]{%
  \small%
    \foreach \@iter in {#2}{%
      \namedrule[#1]{\@iter} \and%
    }%
}

\newcommand*{\dalrulename}[2][\relax]{%
  \def\pr@fix{#1XX}%
  \def\n@me{#2Name}%
  \edef\ou@t{\expandafter\noexpand\csname daldrule\pr@fix\n@me\endcsname}%
  \ou@t{}\hspace{-0.25em}%
}

\newcommand*{\nameddalrule}[2][\relax]{%
  \def\pr@fix{#1XX}%
  \def\n@me{#2}%
  \edef\ou@t{\expandafter\noexpand\csname daldrule\pr@fix\n@me\endcsname}%
  \ou@t{}
}

\newcommand*{\dalrules}[2][\relax]{%
  \small%
  \begin{mathpar}%
    \foreach \@iter in {#2}{%
      \nameddalrule[#1]{\@iter} \and%
    }%
  \end{mathpar}%
}
\makeatother

\newcommand*{\proofitem}[1]{%
  \item{\itshape #1}%
}

\usepackage{underscore}         % Only needed if you use pdflatex.
\usepackage[T1]{fontenc}        % Recommended with pdflatex

\setcopyright{rightsretained}
\acmPrice{}
\acmDOI{10.1145/3609027.3609408}
\acmYear{2023}
\copyrightyear{2023}
\acmSubmissionID{icfpws23tydemain-p11-p}
\acmISBN{979-8-4007-0299-0/23/09}
\acmConference[TyDe '23]{Proceedings of the 8th ACM SIGPLAN International Workshop on Type-Driven Development}{September 4, 2023}{Seattle, WA, USA}
\acmBooktitle{Proceedings of the 8th ACM SIGPLAN International Workshop on Type-Driven Development (TyDe '23), September 4, 2023, Seattle, WA, USA}
\received{2023-06-01}
\received[accepted]{2023-06-29}

\begin{document}

\title{Combining dependency, grades, and adjoint logic}

\author{Peter Hanukaev}
\email{phanukaev@augusta.edu}
\orcid{0009-0006-0056-3091}
\author{Harley Eades III}
\email{harley.eades@gmail.com}
\orcid{0000-0001-8474-5971}
\affiliation{%
  \institution{School of Computer and Cyber Sciences}
  \streetaddress{1120 15th Street}
  \city{Augusta}
  \state{Georgia}
  \country{USA}
  \postcode{30912}
}

\begin{abstract}
  We propose two new dependent type systems. The first, is a dependent
graded/linear type system where a graded dependent type system is
connected via modal operators to a linear type system in the style of
Linear/Non-linear logic. We then generalize this system to support
many graded systems connected by many modal operators through the
introduction of modes from Adjoint Logic.  Finally, we prove several
meta-theoretic properties of these two systems including graded
substitution.

\end{abstract}

\begin{CCSXML}
<ccs2012>
   <concept>
       <concept_id>10003752.10003790.10003801</concept_id>
       <concept_desc>Theory of computation~Linear logic</concept_desc>
       <concept_significance>500</concept_significance>
       </concept>
   <concept>
       <concept_id>10003752.10003790.10011740</concept_id>
       <concept_desc>Theory of computation~Type theory</concept_desc>
       <concept_significance>500</concept_significance>
       </concept>
   <concept>
       <concept_id>10003752.10003790.10003792</concept_id>
       <concept_desc>Theory of computation~Proof theory</concept_desc>
       <concept_significance>300</concept_significance>
       </concept>
   <concept>
       <concept_id>10003752.10010124.10010131.10010137</concept_id>
       <concept_desc>Theory of computation~Categorical semantics</concept_desc>
       <concept_significance>300</concept_significance>
       </concept>
 </ccs2012>
\end{CCSXML}

\ccsdesc[500]{Theory of computation~Linear logic}
\ccsdesc[500]{Theory of computation~Type theory}
\ccsdesc[300]{Theory of computation~Proof theory}
\ccsdesc[300]{Theory of computation~Categorical semantics}

\keywords{linear logic, graded types, adjoint logic, dependent types, semantics, resource tracking}

\maketitle

\section{Introduction}
\label{sec:introduction}
In this paper we consider the combination of dependent types, graded
types, and adjoint logic.  We propose two new dependent type
systems. The first, is a graded version of Krishnaswami
et. al's~\cite{Krishnaswami2015} Dependent LNL system, where the
dependently typed side is replaced with a graded dependent type system
called $\GraD$~\cite{Choudhury2021}. As is the case for LNL style systems
there are two systems, a linear one and a non-linear one, connected by
a set of modal operators for moving between the two. We then consider
how to merge these two systems and support many such modal operators
through the introduction of modes from Adjoint
Logic~\cite{Pruiksma2018}.

Linear logic~\cite{Girard1987} and bunched
implications~\cite{OHearn:1999} are two examples of substructural
logics that have been instrumental in the study of type-based
resource tracking.  In the former, the structural rules for weakening
and contraction have been restricted to only be usable when types are
annotated with a modal operator called of-course and denoted by $!A$. This
restriction enforces that variables can only be used exactly once
unless their type is marked by the of-course modality. Thus, if we
always annotate our types, then we are essentially programming in a
standard functional programming language.  Linear types can be given a
resource-based semantics which allows the type system to enforce that
certain resources must be used or a type error is produced.  For
example, we can use linear types for the types of file handles which
if they are used after they are closed, then a type error is produced.
Linear types are also the logical foundation of session type systems
which can be used to reason about message passing
systems~\cite{Caires:2010}.

Bunched implications is the logical foundation of separation
logic~\cite{OHearn:1999,OHearn:2001} which is used to reason about
imperative programs.  Linear logic can be seen a combination of linear
types and non-linear types through the of-course modality.  Now
Bunched implications is a similar combination, but rather than through
a modality, linear types and non-linear types coexist in the same
system without any modality.  Thus, one can program in either one
without any annotations.  Then both systems interact through the two
notions of implication from the linear and non-linear sides.

Linear logic and Bunched implications can be seen as two ends of a
spectrum based on how much we need to annotate types.  One difficulty
with Bunched implications is that it can be difficult to implement,
and it is unclear how to generalize it with more than two systems.
Benton's Linear/Non-linear logic (LNL)~\cite{Benton1995} is a
compromise between being fully annotated and not needing annotations.
LNL combines linear and non-linear types in such a way that one can
program using both without the need to annotate non-linear types with
the of-course modality, but when one wishes to mix the two fragments
one can use two modal operators that allows transporting a non-linear
type to the linear fragment, and a linear type to the non-linear
fragment.  If one composes these two modal operators, then one obtains
the of-course modality.  Furthermore, we know how to generalize LNL to
support more than two systems.

Adjoint Logic~\cite{Pruiksma2018} generalizes LNL by combining a
family of logics with varying degrees of structural rules as opposed
to just two fragments of non-linear (allowing weakening and
contraction) and linear (no structural rules).  This is accomplished
by taking a linear base system where all types are annotated with a
mode. Each mode is then assigned which structural rules will be
allowed in the fragment the mode represents.  Finally, one can
transport types from one mode to another through modalities similar to
the ones found in LNL. This generalization greatly increases the
expressiveness of the logic.  For example, LNL is easily an instance
by taking two modes one allowing both weakening and contraction and
one that does not.  We could also add a third mode allowing only
contraction resulting in combining non-linear logic, relevance logic,
and linear logic.

Dependent types allow one to specify and prove properties of programs
within the same language they are written~\cite{Nordstrom:1990}.
Linear logic has the benefit of affording the ability to specify and
prove properties of imperative programs.  Krishnaswami et
al.~\cite{Krishnaswami2015} show how to integrate dependent types with
linear types by generalizing the non-linear fragment of LNL to a
dependent type system $\LNLD$ where the modality from the non-linear
fragment now transports a dependent type to the linear fragment.  Then
using this new mixture of dependent and linear types they show how to
specify and prove properties of imperative programs in the style of
Bunched implications.  The modes found in adjoint logic have also been
used to design dependent type system similar to $\LNLD$ with more than
two fragments, but with an eye towards combining dependent types and a
family of modal logics~\cite{10.1145/3373718.3394736} rather than just
controlling the existence of structural rules.

Graded types are a rather recent addition to linear types where types
are annotated with a resource annotation describing how variables of
those types can be used; essentially controlling their dataflow.  The
type system is parameterized by an ordered semiring whose elements are
the grades.  The grades on types offer more fine grained control
over resource usage.  For example, if the ordered semiring is taken to be
the natural numbers, then the grade describes exactly the number of
times the variable is allowed to be used.  Furthermore, graded types
also have been shown to be a means of combining linear types with
dependent types~\cite{Moon2021,Orchard2019,Choudhury2021,Atkey2018}.

\textbf{Contributions.} We combine dependent types, graded types, and
the modes of adjoint logic to define a new system capable of combining
lots of substructural logics.  First, we generalize $\LNLD$ into a new
system called $\dmGL$, and then we generalize $\dmGL$ by adding modes
producing our final system called $\GlaD$. All of our contributions
are as follows:
\begin{itemize}
\item We replace the dependent type system of $\LNLD$ with $\GraD$ a
  graded dependent type system. This system called $\dmGL$ gives more
  control over resource usage~\cite{Choudhury2021} producing a graded
  dependent linear/non-linear system.  Then we prove:
  \begin{itemize}
  \item Substitution for the entire system ensuring that typed graded
    composition is preserved.
  \item Context and type well formedness.
  \item Graded contraction and weakening are admissible in the mixed
    linear/non-linear fragment.
  \item Subject reduction for the entire system.
  \end{itemize}
\item The previous system has two explicit fragments, but with grading
  on one side. Now we generalize this system one step further by
  introducing the modes from adjoint logic.  This system is
  parameterized by a family of modes and preordered semirings where
  each mode is paired with a potentially different preordered
  semiring.  Then each type is annotated with both a mode and a grade
  (element of the preordered semiring).  Then we prove:
  \begin{itemize}
  \item Substitution for the entire system ensuring that typed graded
    composition is preserved.
  \item Contraction is admissible in the mixed fragment.
  \end{itemize}
\end{itemize}

%%% Local Variables: ***
%%% mode:latex ***
%%% tex-main-file: "paper.tex"  ***
%%% End: ***

\section{A Dependent Mixed Graded and Linear Type System}
\label{sec:dmGL}
In this section we present our first type system,
which combines the graded dependently typed system \GraD{} with linear logic.
We call this system \dmGL{} (dependent mixed graded linear).
As with previously proposed dependently typed graded systems
\cite{Atkey2018, Choudhury2021, Moon2021},
variables in \dmGL{} are annotated by grades drawn from a semiring
which captures a computational notion of resource usage.

\begin{defn}[Grades]
  \dmGL{} is parametrized by a preordered semiring
  $ (R, 0, 1, +, \cdot, \leq ) $,
  that is $ R $ is equipped with a preorder $ \leq $
  and a semiring structure
  $ ( R, 0, 1, +, \cdot ) $
  such that the operations $ ( + ) $ and $ ( \cdot ) $ are monotonic in both arguments.
  Elements of $ R $ are called \emph{grades}
  and denoted $ \dmGLnt{r}, \dmGLnt{p}, \dmGLnt{q} $.
\end{defn}

\begin{exmp}[Variable Re-use]
  \label{exmp:var-reuse}
  We take $ R = \N $ the semiring of natural numbers,
  with the usual addition and multiplication.
  In the judgment
  \[
	x :^2 \nat \proves
    \text{\sffamily if Even$ (x) $ then $ x / 2 $ else $ 3 x + 1 $} : \nat
  \]
  the annotation of $ 2 \in \N $,
  indicates that the variable $ x $ is used two times in the computation of the consequent term.
  This grading was originally introduced by Girard for Bounded Linear Logic~\cite{Girard:1992}
  and used to characterize polynomial time computation,
  but has also been used, for example, for automated garbage collection~\cite{Choudhury2021}.
  The preorder we choose on $ \N $ is also relevant:
  It will be used to control the discarding of resources.
  If we choose the ordinary preorder $ \leq $ on $ \N $,
  then we could replace the annotation $ 2 $ above by some other integer $ k \ge 2 $.
  In this case, the annotation $ k $ would mean that the variable $ x $
  is used up to $ k $ times in the computation.
  On the other hand, choosing the preorder to be the trivial one with $ m \leq n \iff m = n $,
  would guarantee a usage of two times exactly.
\end{exmp}

\begin{exmp}[Quantitative Semirings]
  \label{exmp:quant-semiring}
  Continuing from the previous example,
  call a semiring $ R $ \emph{quantitative}%
  \footnote{Here, we follow terminology by Moon et al. \cite{Moon2021}}
  if it satisfies
  \begin{enumerate} [label = \roman*)]
  \item
    $ 0 \neq 1 $

  \item
    $ r + p = 0 \implies r = p = 0 $

  \item
    $ r \cdot p = 0 \implies r = 0 \lor p = 0 $
    
  \end{enumerate}
  Choudhury et al. proved of \GraD{} that a variable graded with $ 0 $ in such a semiring
  is guaranteed to be computationally irrelevant
  and since our system is based on \Grad{}, a similar result is expected to hold for \dmGL{}.
  Therefore, such semirings allow the tracking of computationally relevant vs. irrelevant data.
  This is particularly relevant in dependently typed programs,
  as it allows a distinction between variables which are only used in type checking
  and those which are used in the execution of a program.
  Examples of such semirings are the natural numbers,
  and the following two which we elaborate upon in more detail.

  The boolean semiring is $ R = \{ 0, 1 \} $ with $ 1 + 1 = 1 $.
  This semiring's tracking of variable usage is coarse grained with
  $ 0 $ meaning computational irrelevance and and $ 1 $ representing some usage.
  We have not yet discussed whether $ 0 \leq 1 $ should hold in this semiring.
  \dmGL{} features a subusage rule which assert a variable graded with $ r \in R $
  may also be graded with $ q $, so long as $ r \leq q $.
  If we choose $ 0 \leq 1 $ to be true, the subusage rule will allow
  us to discard variables graded $ 1 $.
  On the other hand if we choose $ 0 \leq 1 $ to not be true,
  variables graded $ 1 $ are guaranteed to be computationally relevant.
  
  The \emph{none-one-tons} semiring is $ R = \{ 0, 1, \omega \} $ in which we have
  $ 1 + 1 = 1 + \omega = \omega $.
  This semiring offers slightly more fine grained tracking,
  with $ 1 $ now representing linear use, and $ \omega $ representing unrestricted use.
  We take $ 0 \leq \omega $ to allow the discarding of unrestricted variables, and
  $1 \leq \omega$ to allow promotion of linear variables to the unrestricted case.
  If we make $ 1 $ incomparable by $ \leq $ with the other elements,
  we can guarantee that variables graded $ 1 $ are in fact used linearly.
\end{exmp}

For the remainder of this section, we fix a preordered semiring 
$ (R, 0, 1, +, \cdot, \leq ) $.

The syntax of terms and types in \dmGL{} is given in \Cref{fig:syntax}
and will be explained throughout the remainder of this section as it becomes relevant.

\begin{defn}
  \emph{Grade vectors} are finite lists of grades, and denoted by $ \delta $.
  They have the syntax
  \[
    \delta := \emptyset \mid \delta ,  \dmGLnt{r}
  \]
  with $ \emptyset $ denoting the empty grade vector.
  We write $ \delta ,  \delta' $ for the concatenation of grade vectors $ \delta $ and $ \delta' $
  and use $ \vec{0} $ to denote any grade vector consisting of only $ 0 $s.
  We extend the operations $ + $ and $ \leq $ to grade vectors of equal length pointwise
  and define scalar multiplication $ \dmGLnt{r}   \cdot   \delta $ in the obvious way.
  \begin{mathpar}
    \emptyset  +  \emptyset = \emptyset
    \and
    \dmGLsym{(}   \delta ,  \dmGLnt{r}   \dmGLsym{)}  +  \dmGLsym{(}   \delta' ,  \dmGLnt{r'}   \dmGLsym{)} = \delta  +  \delta' ,  \dmGLnt{r}  +  \dmGLnt{r'}
    \\
    \emptyset  \leq  \emptyset \iff \True
    \and
    \dmGLsym{(}   \delta ,  \dmGLnt{r}   \dmGLsym{)}  \leq  \dmGLsym{(}   \delta' ,  \dmGLnt{r'}   \dmGLsym{)} \iff \delta  \leq  \delta' \land \dmGLnt{r}  \leq  \dmGLnt{r'}
    \\
    \dmGLnt{r}   \cdot   \emptyset = \emptyset
    \and
    \dmGLnt{r}   \cdot   \dmGLsym{(}   \delta ,  \dmGLnt{q}   \dmGLsym{)} = \dmGLnt{r}   \cdot   \delta  ,  \dmGLnt{r}  \cdot  \dmGLnt{q}
  \end{mathpar}
\end{defn}

The basic structure of \dmGL{} is similar to other mixed
linear/non-linear type systems
\cite{Benton1995, mGL, Krishnaswami2015}.
\dmGL{} consist of two fragments:
A purely graded fragment and a mixed fragment.
Terms and types are divided into graded and linear as well.
Typing judgments in the graded fragment may only have graded hypotheses,
and produce \emph{graded terms} which belong to \emph{graded types}.
In the mixed fragment,
typing judgments produce \emph{linear terms} belonging to \emph{linear types},
but assumptions may consist of both linear and graded formulas.
The graded fragment is the dependently typed system \GraD{},
a type system by Choudhury et al.~\cite{Choudhury2021}.
That work also provides a more detailed discussion of the rules presented here.
In the mixed fragment, linear types may dependent on graded variables but not on linear ones.
Because of this, we treat both graded and linear types as graded terms,
belonging to type universes $ \mathsf{Type} $ and $ \mathsf{Linear} $ respectively.
Graded types are denoted $ \dmGLnt{X}, \dmGLnt{Y}, \dmGLnt{Z}, \dmGLnt{W} $
and linear types are denoted $ \dmGLnt{A}, \dmGLnt{B}, \dmGLnt{C} $.
Their syntax is given as part of the complete syntax of \dmGL{}
in \Cref{fig:syntax}.

\begin{defn}[Contexts]
  Contexts are lists of typing assignments to variables.
  We use the letters $ \dmGLmv{x}, \dmGLmv{y}, \dmGLmv{z} $ for variables.
  While we make no syntactic distinction between variables assigned to graded or linear types,
  we do distinguish between \emph{graded} and \emph{linear contexts},
  denoted by $ \Delta $ and $ \Gamma $ respectively
  and assigning variables to only graded or linear types respectively.
  The graded fragment is dependently typed,
  and in the mixed fragment we allow types in the linear context to depend on variables
  appearing in the graded context.
  Because of this, our system requires judgments asserting that contexts are well formed.
  These judgment forms are $ \delta  \odot  \Delta  \vdash_\mathsf{G} \, \mathsf{ctx} $ and
  $ \delta  \odot  \Delta  \dmGLsym{;}  \Gamma  \vdash_\mathsf{M} \, \mathsf{ctx} $ respectively and their rules are given in \Cref{fig:ctxWellFormed}.
  Here, the formulas $ \dmGLmv{x} \, \notin \, \sfoperator{dom} \, \Delta $ and $ \dmGLmv{x} \, \notin \, \sfoperator{dom} \, \Gamma $ indicate
  that that the variable $ \dmGLmv{x} $ is not bound in context $ \Delta $ or $ \Gamma $ respectively.

  The former of the above judgment forms means that $ \Delta $ is well-formed context
  and also ensures that the attached grade vector has the same length as $ \Delta $.
  The latter ensures that all types appearing in linear context $ \Gamma $ are
  well formed over the graded context $ \Delta $.
\end{defn}

\begin{rem}
  Note that in the context extension rule for graded contexts,
  the grade vector $ \delta $ is extended by an arbitrary grade $ \dmGLnt{r} $.
  Because of this it is actually provable that if $ \delta  \odot  \Delta  \vdash_\mathsf{G} \, \mathsf{ctx} $,
  and $ \delta' $ is any grade vector with the same length as $ \delta $,
  then $ \delta'  \odot  \Delta  \vdash_\mathsf{G} \, \mathsf{ctx} $ and similarly for mixed contexts.
\end{rem}

\begin{figure}
  \begin{mdframed}
    \begin{align*}
        \dmGLnt{X} , \dmGLnt{Y}, \dmGLnt{Z}, \dmGLnt{W} := \ &
            \mathbf{J}
            \mid  ( \dmGLmv{x}  :^{ \dmGLnt{r} }  \dmGLnt{X} )  \boxtimes   \dmGLnt{Y} 
            \mid  \dmGLnt{X}  \boxplus  \dmGLnt{Y} 
            \mid  (  \dmGLmv{x}  :^{ \dmGLnt{r} }  \dmGLnt{X}  )  \to   \dmGLnt{Y}
      \\ &
            \mid  \mathcal{G} \, \dmGLnt{A}
            \mid  \mathsf{Type}
            \mid  \mathsf{Linear}
            \\
        \dmGLnt{t}, \dmGLnt{s} := \ &
        \dmGLmv{x}
            \mid \mathbf{j}
            \mid \sfoperator{let} \, \mathbf{j} \, \dmGLsym{=}  \dmGLnt{t_{{\mathrm{1}}}} \, \sfoperator{in} \, \dmGLnt{t_{{\mathrm{2}}}}
            \mid \dmGLsym{(}  \dmGLnt{t_{{\mathrm{1}}}}  \dmGLsym{,}  \dmGLnt{t_{{\mathrm{2}}}}  \dmGLsym{)}
      \\ &
            \mid \sfoperator{let} \, \dmGLsym{(}  \dmGLmv{x}  \dmGLsym{,}  \dmGLmv{y}  \dmGLsym{)}  \dmGLsym{=}  \dmGLnt{t_{{\mathrm{1}}}} \, \sfoperator{in} \, \dmGLnt{t_{{\mathrm{2}}}}
            \mid \sfoperator{inl} \, \dmGLnt{t}
            \mid \sfoperator{inr} \, \dmGLnt{t}
      \\ &
            \mid \sfoperator{case} _{ \dmGLnt{q} }  \dmGLnt{t}   \sfoperator{of}   \dmGLnt{t_{{\mathrm{1}}}}  ;  \dmGLnt{t_{{\mathrm{2}}}}
            \mid \lambda  \dmGLmv{x}  \dmGLsym{.}  \dmGLnt{t}
            \mid \dmGLnt{t_{{\mathrm{1}}}} \, \dmGLnt{t_{{\mathrm{2}}}}
            \mid \mathcal{G} \, \dmGLnt{l}
            \mid X
            \mid A
            \\
        \dmGLnt{l} := \ &
            \dmGLmv{x}
            \mid \mathbf{i}
            \mid \sfoperator{let} \, \mathbf{i} \, \dmGLsym{=}  \dmGLnt{l_{{\mathrm{1}}}} \, \sfoperator{in} \, \dmGLnt{l_{{\mathrm{2}}}}
            \mid \lambda  \dmGLmv{x}  \dmGLsym{.}  \dmGLnt{l}
            \mid \dmGLnt{l_{{\mathrm{1}}}} \, \dmGLnt{l_{{\mathrm{2}}}}
            \mid \dmGLsym{(}  \dmGLnt{l_{{\mathrm{1}}}}  \dmGLsym{,}  \dmGLnt{l_{{\mathrm{2}}}}  \dmGLsym{)}
            \\ &
            \mid \sfoperator{let} \, \dmGLsym{(}  \dmGLmv{x}  \dmGLsym{,}  \dmGLmv{y}  \dmGLsym{)}  \dmGLsym{=}  \dmGLnt{l_{{\mathrm{1}}}} \, \sfoperator{in} \, \dmGLnt{l_{{\mathrm{2}}}}
            \mid \mathcal{F} \, \dmGLsym{(}  \dmGLnt{t}  \dmGLsym{,}  \dmGLnt{l}  \dmGLsym{)}
            \\ &
            \mid \sfoperator{let} \, \mathcal{F} \, \dmGLsym{(}  \dmGLmv{x}  \dmGLsym{,}  \dmGLmv{y}  \dmGLsym{)}  \dmGLsym{=}  \dmGLnt{l_{{\mathrm{1}}}} \, \sfoperator{in} \, \dmGLnt{l_{{\mathrm{2}}}}
            \mid \mathcal G^{-1} \, \dmGLnt{t}\\
        \dmGLnt{A} , \dmGLnt{B}, \dmGLnt{C} := \ &
            \mathbf{I}
            \mid  \dmGLnt{A}  \multimap  \dmGLnt{B}
            \mid  \dmGLnt{A}  \otimes  \dmGLnt{B}
            \mid  \mathcal{F}  ( \dmGLmv{x}  :^{ \dmGLnt{r} }  \dmGLnt{X} ). \dmGLnt{A}
      \\
      \Delta := \ &
                     \emptyset
                     \mid \Delta  \dmGLsym{,}  \dmGLmv{x}  \dmGLsym{:}  \dmGLnt{X}
      \\
      \Gamma := \ &
                     \emptyset
                     \mid \Gamma  \dmGLsym{,}  \dmGLmv{x}  \dmGLsym{:}  \dmGLnt{A}
    \end{align*}
  \end{mdframed}
  \caption{Syntax of \dmGL{}}
  \label{fig:syntax}
\end{figure}

\begin{figure}
  \begin{mdframed}
    \begin{mathpar}
    \pprules[gradedCtx]{empty, extend}
    \pprules[mixedCtx]{empty, extend}
    \end{mathpar}
  \end{mdframed}
  \caption{Rules for well formed contexts}
  \label{fig:ctxWellFormed}
\end{figure}

Aside from the two type universes,
our system only contains two basic types,
namely the graded and linear unit types $ \mathbf{J} $ and $ \mathbf{I} $ respectively.
To construct more complex graded types,
our system includes coproduct types $ \dmGLnt{X_{{\mathrm{1}}}}  \boxplus  \dmGLnt{X_{{\mathrm{2}}}} $,
dependent function types $ (  \dmGLmv{x}  :^{ \dmGLnt{r} }  \dmGLnt{X}  )  \to   \dmGLnt{Y} $ and a dependent pair type $ ( \dmGLmv{x}  :^{ \dmGLnt{r} }  \dmGLnt{X} )  \boxtimes   \dmGLnt{Y} $.
We explain the roles of the grade annotations in the latter two below.
To form more complex linear types,
we have the linear function type $ \dmGLnt{A}  \multimap  \dmGLnt{B} $ and tensor product type $ \dmGLnt{A}  \otimes  \dmGLnt{B} $
at our disposal.
Lastly, we have the modal operators $ \mathcal{F} $ and $ \mathcal{G} $,
which mediate between the two fragments,
transforming linear types into graded ones and vice versa.
We give the complete type formation rules in \Cref{fig:typeFormation}.

\begin{figure}
  \begin{mdframed}
    \begin{mathpar}
      \pprules[G]{type,linear,unit,function,gradedPair,coproduct,linearFunction,tensor,ladj,radj}
    \end{mathpar}
  \end{mdframed}
  \caption{Rules for type formation}
  \label{fig:typeFormation}
\end{figure}

We now explain the typing rules of our type system in more detail.
The typing judgments for the graded and mixed fragment have the forms
\[
  \delta  \odot  \Delta  \vdash_\mathsf{G}  \dmGLnt{t}  \dmGLsym{:}  \dmGLnt{X}
  \quad \text{and} \quad
  \delta  \odot  \Delta  \dmGLsym{;}  \Gamma  \vdash_\mathsf{M}  \dmGLnt{l}  \dmGLsym{:}  \dmGLnt{A}
\]
respectively.
The annotations $ \ms G $ and $ \ms M $ on the turnstiles
indicate whether the judgment is in the graded or mixed fragment.
The rules enforce that the length of $ \delta $ and $ \Delta $
are equal in any provable judgment.
If $ \delta = \dmGLnt{r_{{\mathrm{1}}}} , \mathellipsis , \dmGLnt{r_{\dmGLmv{n}}} $
and $ \Delta = x_1 : X_1 , \mathellipsis , x_n : X_n $,
the above judgment forms indicate that variable $ x_i $ is used with grade
$ r_i $ in the construction of the term $ t $ (resp. $ l $).
Notice that both graded and linear types are themselves graded terms,
but the above judgment forms contain no information about the grades used
in the construction of the type $ X $ (resp. $ A $).

The rule \drulename[G]{weak} allows weakening, provided the newly added variable
is used with grade $ 0 $.
Similarly, for the variable rule \drulename[G]{var} we require that the variable in the conclusion of the rule
is used exactly with grade $ 1 $ and all other variables are used with grade $ 0 $.
Finally, we include a sub-usage rule \drulename[G]{subusage}
which asserts that we can make typing judgments with higher grades than necessary.
The graded unit type $ \mathbf{J} $ has one closed constructor $ \mathbf{j} $
and a term of unit type can be eliminated by pattern matching
The graded dependent pair type $ ( \dmGLmv{x}  :^{ \dmGLnt{r} }  \dmGLnt{X} )  \boxtimes   \dmGLnt{Y} $ comes with a grade annotation $ \dmGLnt{r} $.
This annotation means that to eliminate a term of the form $ \dmGLsym{(}  \dmGLnt{t_{{\mathrm{1}}}}  \dmGLsym{,}  \dmGLnt{t_{{\mathrm{2}}}}  \dmGLsym{)} $ of this type,
the grade at which the first component $ \dmGLnt{t_{{\mathrm{1}}}} $ of the pair is used
must be $ \dmGLnt{r} $ times the grade at which $ \dmGLnt{t_{{\mathrm{2}}}} $ is used.

The coproduct type $ \dmGLnt{X_{{\mathrm{1}}}}  \boxplus  \dmGLnt{X_{{\mathrm{2}}}} $ has the expected left and right injections
$ \sfoperator{inl} $ and $ \sfoperator{inr} $ as constructors
and its elimination form $ \sfoperator{case} _{ \dmGLnt{q} }  \dmGLnt{t}   \sfoperator{of}   \dmGLnt{s_{{\mathrm{1}}}}  ;  \dmGLnt{s_{{\mathrm{2}}}} $ works by case distinction.
The one caveat is that the functions $ \dmGLnt{s_{{\mathrm{1}}}} $ and $ \dmGLnt{s_{{\mathrm{2}}}} $,
which describe the two cases, must use their input with the same grade $ \dmGLnt{q} $,
which we include as annotation on the elimination form.

The dependent function type $ (  \dmGLmv{x}  :^{ \dmGLnt{r} }  \dmGLnt{X}  )  \to   \dmGLnt{Y} $ has a grade annotation
which indicates at which grade the variable $ \dmGLmv{x} $ of type $ \dmGLnt{X} $ must be used:
Introduction is done via lambda abstraction, with the constraint that the variable
that we are abstracting over must be used at grade $ \dmGLnt{r} $.
Similarly, if $ \dmGLnt{t} $ is of type $ (  \dmGLmv{x}  :^{ \dmGLnt{r} }  \dmGLnt{X}  )  \to   \dmGLnt{Y} $ and $ \dmGLnt{t'} $ of type $ \dmGLnt{X} $,
then the grades used to construct $ \dmGLnt{t'} $ are multiplied by $ \dmGLnt{r} $ when constructing
the application $ \dmGLnt{t} \, \dmGLnt{t'} $.

Since term judgments contain no information about the grades used in the type,
the grade annotations in the dependent function type $ (  \dmGLmv{x}  :^{ \dmGLnt{r} }  \dmGLnt{X}  )  \to   \dmGLnt{Y} $
and $ ( \dmGLmv{x}  :^{ \dmGLnt{r} }  \dmGLnt{X} )  \boxtimes   \dmGLnt{Y} $ do not need to be the same as the grade
with which $ x : X $ is used in the construction of $ Y $.
A similar remark holds for the left adjoint $ \mathcal{F}  ( \dmGLmv{x}  :^{ \dmGLnt{r} }  \dmGLnt{X} ). \dmGLnt{A} $ below.

\begin{figure}
  \begin{mdframed}
    \begin{mathpar}
      \pprules[G]{subusage,weak,var,convert,unitIntro,unitElim,gradedPairIntro,gradedPairElim,
        coproductInl,coproductInr,coproductElim,lambda,app,radjIntro}
    \end{mathpar}
  \end{mdframed}
  \caption{Graded system type assignment}
  \label{fig:GTyping}
\end{figure}

The mixed fragment behaves like a linear logic,
with an additional context of graded variables available.
When linear contexts are concatenated,
the grade vectors of the shared graded context are added.
The mixed fragment features a linear unit type $ \mathbf{I} $ with inhabitant $ \mathbf{i} $,
a linear function type $ \dmGLnt{A}  \multimap  \dmGLnt{B} $ and a pair type $ \dmGLnt{A}  \otimes  \dmGLnt{B} $.
The rules the mixed fragment are specified in \Cref{fig:MTyping} and \Cref{fig:MTyping2}.
Some rules in the mixed fragment feature concatenation of linear contexts.
In those cases we assume that variables are renamed to avoid name clashes.

\begin{figure}
  \begin{mdframed}
    \begin{mathpar}
      \pprules[M]{id,subusage,weak,exchange,convert,unitIntro,unitElim,lambda,app}
    \end{mathpar}
  \end{mdframed}
  \caption{Mixed system type assignment}
  \label{fig:MTyping}
\end{figure}

\begin{figure}
  \begin{mdframed}
    \begin{mathpar}
      \pprules[M]{tensorIntro,tensorElim,ladjIntro,ladjElim,radjElim}
    \end{mathpar}
  \end{mdframed}
  \caption{Mixed system type assignment continued}
  \label{fig:MTyping2}
\end{figure}

Finally, we discuss the modal operators $ \mathcal{F} $ and $ \mathcal{G} $.
The operator $ \mathcal{G} $ takes a linear type $ \dmGLnt{A} $ and produces
a graded type $ \mathcal{G} \, \dmGLnt{A} $.
It's function is analogous to the operators of Benton \cite{Benton1995}
and Krishnaswami et al. \cite{Krishnaswami2015}.
It transforms linear terms $ \dmGLnt{l} $ of type $ \dmGLnt{A} $
with no free linear variables into graded terms of type $ \mathcal{G} \, \dmGLnt{A} $.
Elimination of terms of type $ \mathcal{G} \, \dmGLnt{A} $ is handled by the operator $ \mathcal G^{-1} \, \dmGLnt{t} $ which
produces terms of type $ \dmGLnt{A} $ from terms of type $ \mathcal{G} \, \dmGLnt{A} $.
Here, we can see that linearity and grading line up, if $ \dmGLnt{l} $
is a linear term with free variable $ \dmGLmv{x} $ of linear type $ \dmGLnt{A} $,
there is a corresponding term $ \dmGLsym{[}  \mathcal G^{-1} \, \dmGLmv{y}  \dmGLsym{/}  \dmGLmv{x}  \dmGLsym{]}  \dmGLnt{l} $ with free variable $ \dmGLmv{y} $
of type $ \mathcal{G} \, \dmGLnt{A} $ and variable $ \dmGLmv{y} $ is graded $ 1 $,
see \Cref{prop:RadjLeftRule} below.
The operator $ \mathcal{F} $ is again similar to that of Kirshnaswami \cite{Krishnaswami2015}.
The type $ \mathcal{F}  ( \dmGLmv{x}  :^{ \dmGLnt{r} }  \dmGLnt{X} ). \dmGLnt{A} $ behaves like a dependent pair type
where the first component belongs to graded type $ \dmGLnt{X} $
and the second component belongs to linear type $ \dmGLnt{A} $.
Elimination from this type is done by pattern matching
$ \sfoperator{let} \, \mathcal{F} \, \dmGLsym{(}  \dmGLmv{x}  \dmGLsym{,}  \dmGLmv{y}  \dmGLsym{)}  \dmGLsym{=}  \dmGLnt{l_{{\mathrm{1}}}} \, \sfoperator{in} \, \dmGLnt{l_{{\mathrm{2}}}} $ and the grade annotation $ \dmGLnt{r} $
in the type forces the first component to be used with grade $ \dmGLnt{r} $
for the eliminator to be invoked.

\Cref{fig:cbn} contains the reduction rules (full $\beta$-reduction) for \dmGL{}.
Since reduction rules only see the syntactic form of terms, the grades are not involved at all,
and hence the reduction rules are as one would expect.
The application of a lambda abstraction to some term reduces by substitution,
and pattern matching expressions reduce when the term being matched on
is exactly of the form of the pattern, in which case the reduction is also by substitution.
Finally, the case expression for the coproduct eliminator
reduces if the scrutinee was constructed by one of the injections,
in which case the reduction works by applying the respective function.
We write $ \equiv $ for the congruence closure of $ \leadsto $,
on graded terms.
In other words, $ \equiv $ is the smallest equivalence relation
on graded terms that contains $ \leadsto $ and such that
$ \dmGLnt{t_{{\mathrm{1}}}} \equiv \dmGLnt{t_{{\mathrm{2}}}} $ implies $ \dmGLsym{[}  \dmGLnt{t_{{\mathrm{1}}}}  \dmGLsym{/}  \dmGLmv{x}  \dmGLsym{]}  \dmGLnt{t} \equiv \dmGLsym{[}  \dmGLnt{t_{{\mathrm{2}}}}  \dmGLsym{/}  \dmGLmv{x}  \dmGLsym{]}  \dmGLnt{t} $
for all terms $ t, t_1, t_2 $ and likewise with linear terms.
Similarly, we write $ \equiv $ for the smallest equivalence relation on
linear terms which contains $ \leadsto $ and satisfies the implications
$ t_1 \equiv t_2 \implies \dmGLsym{[}  \dmGLnt{t_{{\mathrm{1}}}}  \dmGLsym{/}  \dmGLmv{x}  \dmGLsym{]}  \dmGLnt{l} \equiv \dmGLsym{[}  \dmGLnt{t_{{\mathrm{2}}}}  \dmGLsym{/}  \dmGLmv{x}  \dmGLsym{]}  \dmGLnt{l} $
and
$ l_1 \equiv l_2 \implies \dmGLsym{[}  \dmGLnt{l_{{\mathrm{1}}}}  \dmGLsym{/}  \dmGLmv{x}  \dmGLsym{]}  \dmGLnt{l} \equiv \dmGLsym{[}  \dmGLnt{l_{{\mathrm{2}}}}  \dmGLsym{/}  \dmGLmv{x}  \dmGLsym{]}  \dmGLnt{l} $
for all $ t_1, t_2, l_1, l_2 $ and $ l $.
We will only use the relation $ \equiv $ in the type conversion rules, which are standard.

\begin{figure}
  \begin{mdframed}
    \begin{mathpar}
      \pprules[gRED]{unitBeta,pairBeta,coproductBetaLeft,coproductBetaRight,lambda,appL}
      \pprules[mRED]{unitBeta,tensorBeta,ladjBeta,radjBeta,lambda,appL}
    \end{mathpar}
  \end{mdframed}
  \caption{Reduction Rules}
  \label{fig:cbn}
\end{figure}

\subsection*{Metatheory}
\label{sec:metatheory}

We now turn our attention to the metatheory of \dmGL{}.  First we
state some well-formedness conditions.  Then we show how grading
interacts with substitution and linearity.  Finally, we show that the
full $\beta$-reduction rules preserve typing and grading.
Proofs of theorems stated in this section can be found in \Cref{app-sec:dmgl}.

In provable typing judgments,
the context in which the typing occurs is well-formed.
Furthermore, terms always have well-formed types.

\begin{restatable}{prop}{ContextWellFormedness}
  \label{prop:ctxWellFormed}
  The following hold by mutual induction:
  \begin{enumerate}[label=\roman*)]
    \item 
      If $ \delta  \odot  \Delta  \vdash_\mathsf{G}  \dmGLnt{t}  \dmGLsym{:}  \dmGLnt{X} $, then $ \delta  \odot  \Delta  \vdash_\mathsf{G} \, \mathsf{ctx} $.

    \item
      If $ \delta  \odot  \Delta  \dmGLsym{;}  \Gamma  \vdash_\mathsf{M}  \dmGLnt{l}  \dmGLsym{:}  \dmGLnt{A} $, then $ \delta  \odot  \Delta  \vdash_\mathsf{G} \, \mathsf{ctx} $
      and $ \delta  \odot  \Delta  \dmGLsym{;}  \Gamma  \vdash_\mathsf{M} \, \mathsf{ctx} $.
  \end{enumerate}
\end{restatable}

\begin{restatable}{prop}{TypesWellFormed}
  \label{cor:TypesWellFormed}
  The following hold by mutual induction:
  \begin{enumerate}[label=\roman*)]
    \item 
      If $ \delta  \odot  \Delta  \vdash_\mathsf{G}  \dmGLnt{t}  \dmGLsym{:}  \dmGLnt{X} $, then $ \delta'  \odot  \Delta  \vdash_\mathsf{G}  \dmGLnt{X}  \dmGLsym{:}   \mathsf{Type} $
      for some grade vector $ \delta' $.

    \item 
      If $ \delta  \odot  \Delta  \dmGLsym{;}  \Gamma  \vdash_\mathsf{M}  \dmGLnt{l}  \dmGLsym{:}  \dmGLnt{A} $,
      then $ \delta'  \odot  \Delta  \vdash_\mathsf{G}  \dmGLnt{A}  \dmGLsym{:}   \mathsf{Linear} $ for some grade vector $ \delta' $.
  \end{enumerate}
\end{restatable}

Next, we consider substitution.
Since we have a graded and a linear fragment
we need to state substitution for both fragments,
additionally, in the linear fragment substitution is split further into
cases where a variable is replaced by a graded or a linear term.
Since our system has dependent typing,
when a graded term is substituted for a variable we also need to 
substitute it in part of the context.
We therefore make the following definition:

\begin{defn}
  Let $ \Delta $ be a graded context, $ \Gamma $ a linear context,
  $ \dmGLmv{x} $ a term variable
  and suppose $ \dmGLmv{x} \, \notin \, \sfoperator{dom} \, \Delta $ and $ \dmGLmv{x} \, \notin \, \sfoperator{dom} \, \Gamma $.
  We define $ \dmGLsym{[}  \dmGLnt{t}  \dmGLsym{/}  \dmGLmv{x}  \dmGLsym{]}  \Delta $ and $ \dmGLsym{[}  \dmGLnt{t}  \dmGLsym{/}  \dmGLmv{x}  \dmGLsym{]}  \Gamma $ as follows:
  \begin{mathpar}
    \dmGLsym{[}  \dmGLnt{t}  \dmGLsym{/}  \dmGLmv{x}  \dmGLsym{]}  \emptyset = \emptyset
    \and
    \dmGLsym{[}  \dmGLnt{t}  \dmGLsym{/}  \dmGLmv{x}  \dmGLsym{]}  \dmGLsym{(}  \Delta  \dmGLsym{,}  \dmGLmv{y}  \dmGLsym{:}  \dmGLnt{Y}  \dmGLsym{)} = \dmGLsym{[}  \dmGLnt{t}  \dmGLsym{/}  \dmGLmv{x}  \dmGLsym{]}  \Delta  \dmGLsym{,}  \dmGLmv{y}  \dmGLsym{:}  \dmGLsym{[}  \dmGLnt{t}  \dmGLsym{/}  \dmGLmv{x}  \dmGLsym{]}  \dmGLnt{Y}
    \and
    \dmGLsym{[}  \dmGLnt{t}  \dmGLsym{/}  \dmGLmv{x}  \dmGLsym{]}  \dmGLsym{(}  \Gamma  \dmGLsym{,}  \dmGLmv{y}  \dmGLsym{:}  \dmGLnt{A}  \dmGLsym{)} = \dmGLsym{[}  \dmGLnt{t}  \dmGLsym{/}  \dmGLmv{x}  \dmGLsym{]}  \Gamma  \dmGLsym{,}  \dmGLmv{y}  \dmGLsym{:}  \dmGLsym{[}  \dmGLnt{t}  \dmGLsym{/}  \dmGLmv{x}  \dmGLsym{]}  \dmGLnt{A}
  \end{mathpar}
\end{defn}

We use an additional notational convention:
\begin{conv}
  \label{conv:list-match}
  We assume that lengths of grade vectors and corresponding contexts match in judgments.
  For example, when we write
  $ \delta ,  \dmGLnt{r}  ,  \delta'   \odot  \Delta  \dmGLsym{,}  \dmGLmv{x}  \dmGLsym{:}  \dmGLnt{X}  \dmGLsym{,}  \Delta'  \vdash_\mathsf{G}  \dmGLnt{t}  \dmGLsym{:}  \dmGLnt{Y} $,
  we assume that $ \delta $ has the same length as $ \Delta $
  and similarly for $ \delta' $ and $ \Delta' $.
\end{conv}

We can now state the substitution theorem.
Parallel composition is modeled by addition in the semiring,
while sequential composition is modeled by multiplication.

\begin{restatable}[Substitution]{theorem}{substitutionTheorem}
  \label{thm:substitution}
  The following hold by mutual induction:
  \begin{enumerate}[label=\roman*)]
    \item (Graded Contexts)
      If $ \delta_{{\mathrm{0}}}  \odot  \Delta  \vdash_\mathsf{G}  \dmGLnt{t_{{\mathrm{0}}}}  \dmGLsym{:}  \dmGLnt{X} $,
      and $ \delta ,  \dmGLnt{r}  ,  \delta'   \odot  \Delta  \dmGLsym{,}  \dmGLmv{x}  \dmGLsym{:}  \dmGLnt{X}  \dmGLsym{,}  \Delta'  \vdash_\mathsf{G} \, \mathsf{ctx} $
      then
      \[
        \delta  +   \dmGLnt{r}   \cdot   \delta_{{\mathrm{0}}}  ,  \delta'   \odot  \Delta  \dmGLsym{,}  \dmGLsym{[}  \dmGLnt{t_{{\mathrm{0}}}}  \dmGLsym{/}  \dmGLmv{x}  \dmGLsym{]}  \Delta'  \vdash_\mathsf{G} \, \mathsf{ctx}
      \]

    \item (Mixed Contexts)
      If $ \delta_{{\mathrm{0}}}  \odot  \Delta  \vdash_\mathsf{G}  \dmGLnt{t_{{\mathrm{0}}}}  \dmGLsym{:}  \dmGLnt{X} $,
      and $ \delta ,  \dmGLnt{r}  ,  \delta'   \odot  \Delta  \dmGLsym{,}  \dmGLmv{x}  \dmGLsym{:}  \dmGLnt{X}  \dmGLsym{,}  \Delta'  \dmGLsym{;}  \Gamma  \vdash_\mathsf{M} \, \mathsf{ctx} $
      then
      \[
        \delta  +   \dmGLnt{r}   \cdot   \delta_{{\mathrm{0}}}  ,  \delta'   \odot  \Delta  \dmGLsym{,}  \dmGLsym{[}  \dmGLnt{t_{{\mathrm{0}}}}  \dmGLsym{/}  \dmGLmv{x}  \dmGLsym{]}  \Delta'  \dmGLsym{;}  \dmGLsym{[}  \dmGLnt{t_{{\mathrm{0}}}}  \dmGLsym{/}  \dmGLmv{x}  \dmGLsym{]}  \Gamma  \vdash_\mathsf{M} \, \mathsf{ctx}
      \]

    \item (Graded)
      If $ \delta_{{\mathrm{0}}}  \odot  \Delta  \vdash_\mathsf{G}  \dmGLnt{t_{{\mathrm{0}}}}  \dmGLsym{:}  \dmGLnt{X} $
      and $ \delta ,  \dmGLnt{r}  ,  \delta'   \odot  \Delta  \dmGLsym{,}  \dmGLmv{x}  \dmGLsym{:}  \dmGLnt{X}  \dmGLsym{,}  \Delta'  \vdash_\mathsf{G}  \dmGLnt{t}  \dmGLsym{:}  \dmGLnt{Y} $,
      then
      \[
        \delta  +  \dmGLnt{r}  \cdot  \delta_{{\mathrm{0}}} ,  \delta'   \odot  \Delta  \dmGLsym{,}  \dmGLsym{[}  \dmGLnt{t_{{\mathrm{0}}}}  \dmGLsym{/}  \dmGLmv{x}  \dmGLsym{]}  \Delta'  \vdash_\mathsf{G}  \dmGLsym{[}  \dmGLnt{t_{{\mathrm{0}}}}  \dmGLsym{/}  \dmGLmv{x}  \dmGLsym{]}  \dmGLnt{t}  \dmGLsym{:}  \dmGLsym{[}  \dmGLnt{t_{{\mathrm{0}}}}  \dmGLsym{/}  \dmGLmv{x}  \dmGLsym{]}  \dmGLnt{Y}
      \]

    \item (Mixed Graded)
      If $ \delta_{{\mathrm{0}}}  \odot  \Delta  \vdash_\mathsf{G}  \dmGLnt{t_{{\mathrm{0}}}}  \dmGLsym{:}  \dmGLnt{X} $
      and $ \delta ,  \dmGLnt{r}  ,  \delta'   \odot  \Delta  \dmGLsym{,}  \dmGLmv{x}  \dmGLsym{:}  \dmGLnt{X}  \dmGLsym{,}  \Delta'  \dmGLsym{;}  \Gamma  \vdash_\mathsf{M}  \dmGLnt{l}  \dmGLsym{:}  \dmGLnt{A} $,
      then
      \[
        \delta  +  \dmGLnt{r}  \cdot  \delta_{{\mathrm{0}}}  \odot  \Delta  \dmGLsym{,}  \dmGLsym{[}  \dmGLnt{t_{{\mathrm{0}}}}  \dmGLsym{/}  \dmGLmv{x}  \dmGLsym{]}  \Delta'  \dmGLsym{;}  \dmGLsym{[}  \dmGLnt{t_{{\mathrm{0}}}}  \dmGLsym{/}  \dmGLmv{x}  \dmGLsym{]}  \Gamma  \vdash_\mathsf{M}  \dmGLsym{[}  \dmGLnt{t_{{\mathrm{0}}}}  \dmGLsym{/}  \dmGLmv{x}  \dmGLsym{]}  \dmGLnt{l}  \dmGLsym{:}  \dmGLsym{[}  \dmGLnt{t_{{\mathrm{0}}}}  \dmGLsym{/}  \dmGLmv{x}  \dmGLsym{]}  \dmGLnt{A}
      \]

    \item
      (Mixed Linear)
      If $ \delta_{{\mathrm{0}}}  \odot  \Delta  \dmGLsym{;}  \Gamma_{{\mathrm{0}}}  \vdash_\mathsf{M}  \dmGLnt{l_{{\mathrm{0}}}}  \dmGLsym{:}  \dmGLnt{A} $
      and $ \delta  \odot  \Delta  \dmGLsym{;}   \Gamma  \dmGLsym{,}  \dmGLmv{x}  \dmGLsym{:}  \dmGLnt{A}  ,  \Gamma'   \vdash_\mathsf{M}  \dmGLnt{l}  \dmGLsym{:}  \dmGLnt{B} $,
      then
      \[
        \delta  +  \delta_{{\mathrm{0}}}  \odot  \Delta  \dmGLsym{;}    \Gamma  ,  \Gamma_{{\mathrm{0}}}   ,  \Gamma'   \vdash_\mathsf{M}  \dmGLsym{[}  \dmGLnt{l_{{\mathrm{0}}}}  \dmGLsym{/}  \dmGLmv{x}  \dmGLsym{]}  \dmGLnt{l}  \dmGLsym{:}  \dmGLnt{B}
      \]
  \end{enumerate}
\end{restatable}

\begin{cor}[Contraction]
  \dmGL{} admits graded contraction rules:
  \begin{enumerate}[label=\roman*)]
  \item
    If $ \delta ,  \dmGLnt{p}  ,  \dmGLnt{q}  ,  \delta'   \odot  \Delta  \dmGLsym{,}  \dmGLmv{x}  \dmGLsym{:}  \dmGLnt{X}  \dmGLsym{,}  \dmGLmv{y}  \dmGLsym{:}  \dmGLnt{X}  \dmGLsym{,}  \Delta'  \vdash_\mathsf{G}  \dmGLnt{t}  \dmGLsym{:}  \dmGLnt{Y} $,
    then
    $ \delta ,  \dmGLnt{p}  +  \dmGLnt{q}  ,  \delta'   \odot  \Delta  \dmGLsym{,}  \dmGLmv{x}  \dmGLsym{:}  \dmGLnt{X}  \dmGLsym{,}  \dmGLsym{[}  \dmGLmv{x}  \dmGLsym{/}  \dmGLmv{y}  \dmGLsym{]}  \Delta'  \vdash_\mathsf{G}  \dmGLsym{[}  \dmGLmv{x}  \dmGLsym{/}  \dmGLmv{y}  \dmGLsym{]}  \dmGLnt{t}  \dmGLsym{:}  \dmGLsym{[}  \dmGLmv{x}  \dmGLsym{/}  \dmGLmv{y}  \dmGLsym{]}  \dmGLnt{Y} $.

  \item
    If $ \delta ,  \dmGLnt{p}  ,  \dmGLnt{q}  ,  \delta'   \odot  \Delta  \dmGLsym{,}  \dmGLmv{x}  \dmGLsym{:}  \dmGLnt{X}  \dmGLsym{,}  \dmGLmv{y}  \dmGLsym{:}  \dmGLnt{X}  \dmGLsym{,}  \Delta'  \dmGLsym{;}  \Gamma  \vdash_\mathsf{M}  \dmGLnt{l}  \dmGLsym{:}  \dmGLnt{A} $,
    then
    $ \delta ,  \dmGLnt{p}  +  \dmGLnt{q}  ,  \delta'   \odot  \Delta  \dmGLsym{,}  \dmGLmv{x}  \dmGLsym{:}  \dmGLnt{X}  \dmGLsym{,}  \dmGLsym{[}  \dmGLmv{x}  \dmGLsym{/}  \dmGLmv{y}  \dmGLsym{]}  \Delta'  \dmGLsym{;}  \Gamma  \vdash_\mathsf{M}  \dmGLsym{[}  \dmGLmv{x}  \dmGLsym{/}  \dmGLmv{y}  \dmGLsym{]}  \dmGLnt{l}  \dmGLsym{:}  \dmGLsym{[}  \dmGLmv{x}  \dmGLsym{/}  \dmGLmv{y}  \dmGLsym{]}  \dmGLnt{A} $.
  \end{enumerate}
\end{cor}

\begin{proof}
  Apply \Cref{thm:substitution} to the judgment
  $ \vec{0} ,   1    \odot  \Delta  \dmGLsym{,}  \dmGLmv{x}  \dmGLsym{:}  \dmGLnt{X}  \vdash_\mathsf{G}  \dmGLmv{x}  \dmGLsym{:}  \dmGLnt{X} $
  obtained from the variable rule.
\end{proof}

We now show how the modal operator $ \mathcal{G} $ interacts
with linearity and grading.
The following proposition essentially says that
we can embed the linear fragment into the graded one,
by making the grades on all variables from the linear fragment equal to $ 1 $.

\begin{restatable}{prop}{RadjLeftRule}
  \label{prop:RadjLeftRule}
  If $ \delta  \odot  \Delta  \dmGLsym{;}  \Gamma  \dmGLsym{,}  \dmGLmv{x}  \dmGLsym{:}  \dmGLnt{A}  \vdash_\mathsf{M}  \dmGLnt{l}  \dmGLsym{:}  \dmGLnt{B} $,
  then $ \delta ,   1    \odot  \Delta  \dmGLsym{,}  \dmGLmv{y}  \dmGLsym{:}  \mathcal{G} \, \dmGLnt{A}  \dmGLsym{;}  \Gamma  \vdash_\mathsf{M}  \dmGLsym{[}  \mathcal G^{-1} \, \dmGLmv{y}  \dmGLsym{/}  \dmGLmv{x}  \dmGLsym{]}  \dmGLnt{l}  \dmGLsym{:}  \dmGLnt{B} $.
\end{restatable}

Finally, we state the subject reduction theorem.
Reduction does not only preserve typing,
it also preserves the grades used in typing judgments.
Since most reductions are by substitution,
the bulk of the work has already been done in \Cref{thm:substitution}.

\begin{restatable}[Subject Reduction]{theorem}{SubjectReduction}
  \label{prop:cbnSubjectReduction}
  \begin{enumerate}[label=\roman*)]
    \item 
      If $ \delta  \odot  \Delta  \vdash_\mathsf{G}  \dmGLnt{t}  \dmGLsym{:}  \dmGLnt{X} $ and $ \dmGLnt{t}  \leadsto  \dmGLnt{t'} $,
      then $ \delta  \odot  \Delta  \vdash_\mathsf{G}  \dmGLnt{t'}  \dmGLsym{:}  \dmGLnt{X} $.

    \item 
      If $ \delta  \odot  \Delta  \dmGLsym{;}  \Gamma  \vdash_\mathsf{M}  \dmGLnt{l}  \dmGLsym{:}  \dmGLnt{A} $ and $ \dmGLnt{l}  \leadsto  \dmGLnt{l'} $,
      then $ \delta  \odot  \Delta  \dmGLsym{;}  \Gamma  \vdash_\mathsf{M}  \dmGLnt{l'}  \dmGLsym{:}  \dmGLnt{A} $.
  \end{enumerate}
\end{restatable}

\section{A graded type system in the style of adjoint logic}
\label{sec:dal}
We now present a further generalization of the previous type system,
which is inspired by Adjoint Logic \cite{Pruiksma2018},
and which we cal \glad{}.
Adjoint Logic offers a smooth way of combining an arbitrary number of
substructural logics which are identified by \emph{modes}.
Each mode $ \mathfrak{m} $ is assigned a set $ \sigma( \mathfrak{m} ) \subseteq \{\ms W, \ms C\} $
of the structural rules weakening (\textsf{W}) and contraction (\textsf{C})
satisfied by its logic,
and the modes are arranged in a preorder such that the map $ \sigma $ is monotone,
i.e. if $ \mathfrak{m}_{{\mathrm{1}}} \leq \mathfrak{m}_{{\mathrm{2}}} $, then $ \sigma( \mathfrak{m}_{{\mathrm{1}}}) \subseteq \sigma( \mathfrak{m}_{{\mathrm{2}}}) $.
One of the key insights of Adjoint Logic is that judgments
\[
  A_{ \mathfrak{m}_{{\mathrm{1}}}}^1 , \mathellipsis ,A_{ \mathfrak{m}_{\dalmv{n}}}^n \proves B_{\mathfrak{m}}
\]
must satisfy $ \mathfrak{m}_{\dalmv{i}} \ge \mathfrak{m} $ for each $ i $,
where the subscripts indicate the mode a proposition belongs to.

In the system of this section, modes come equipped with a preordered semiring
controlling the resource structure.
For any one of these semirings, this system is a dependently typed system
with the same rules as \dmGL{}.
We can control weakening in each of these fragments,
but in the presence of dependent types, controlling for contraction becomes difficult.
We will discuss the issues of this later.

\begin{defn}
  \label{defn:semiring-hom}
  Let $ R, S $ be preordered semirings.
  A morphism of preordered semirings $ R \to S $ is a map
  $ f \from R \to S $ such that $ f(0) = 0 $, $ f(1) = 1 $
  and for all $ a, b \in R $, we have
  $ f(a + b) = f(a) + f(b) $ and $ f(a \cdot b) = f(a) \cdot f(b) $
  and $ a \leq b \implies f(a) \leq f(b) $.
  If $ f \from R \to S $ is a morphism of preordered semirings
  and $ r \in R $ and $ s \in S $,
  then we write $ r \cdot s := f(r) \cdot s \in S $.
\end{defn}

\begin{defn}
  A \emph{mode} $ \mathfrak{m} $ is a pair $ (R_{\mathfrak{m}} , \mathsf{Weak} (  \mathfrak{m}  ) ) $,
  where $ R_{\mathfrak{m}} $ is a preordered semiring and
  $ \mathsf{Weak} (  \mathfrak{m}  ) $ is either true or false.
  We will write $ \dalnt{r}  \colon  \mathfrak{m} $ to mean $ r \in R_{ \mathfrak{m}} $.
  Modes are denoted by the lowercase fraktur letters $ \mf{m, n, l} $.

  Let $ \mathfrak{m} $ and $ \mf n $ be modes
  and assume that the proposition $ \mathsf{Weak} (  \mathfrak{m}  ) \implies \mathsf{Weak} (  \mathfrak{n}  ) $ is true.
  A morphism of modes $ \mathfrak{m} \to \mf n $ is a morphism $ R_\mathfrak{m} \to R_\dalmv{n} $
  of the underlying preordered semirings.
  There is a category of modes, denoted by $ \Modes $.
\end{defn}

For the rest of this section, fix a preordered set $ I $ and a functor $ I \to \Modes $.
That is, fix:
For each $ i \in I $, a mode $ \mathfrak{m}_{\dalmv{i}} $,
and for each $ i \leq j $ in $ I $, a morphism of modes 
$ f_{ij} \from \mf m_i \to \mf m_j $ such that
for $ i \leq j \leq k $, $ f_{jk} \circ f_{ij} = f_{ik} $ holds.
We will write $ \mathfrak{m}_{\dalmv{i}} \leq \mathfrak{m}_{\dalmv{j}} $ when $ i \leq j $.
\glad{} is parametrized by this data, and we give examples below.

\begin{figure}
  \begin{mdframed}
    \begin{align*}
      \text{Types: } A, B, C ::= & \ \mathsf{Type}
                           \mid \mathbf{I}_{ \mathfrak{m} }
                           \mid (  \dalmv{x}  :^{ \dalnt{r}  :  \mathfrak{m} }  \dalnt{A}  )  \multimap   \dalnt{B}
      \\ &
                           \mid (  \dalmv{x}  :^{ \dalnt{r}  :  \mathfrak{m} }  \dalnt{A}  )  \otimes   \dalnt{B}
                           \mid \dalnt{A}  \oplus  \dalnt{B}
                           \mid \uparrow_{ \mathfrak{m}_{{\mathrm{1}}} }^{ \mathfrak{m}_{{\mathrm{2}}} }\!\!  \dalnt{A}
      \\
      \text{Terms: } a, b, c ::= & \
                                   \dalmv{x}
                                   \mid \star_{ \mathfrak{m} }
                                   \mid \operatorname{\mathsf{let} }  \star_{ \mathfrak{m} } =  \dalnt{a}   \operatorname{\mathsf{in} }   \dalnt{b}
                                   \mid \lambda  \dalmv{x}  .  \dalnt{a}
                                   \mid \dalnt{a} \, \dalnt{b}
      \\ &
                                   \mid \operatorname{\mathsf{let} } \, \dalsym{(}  \dalmv{x}  \dalsym{,}  \dalmv{y}  \dalsym{)}  \dalsym{=}  \dalnt{a} \, \operatorname{\mathsf{in} } \, \dalnt{b}
                                   \mid \operatorname{\mathsf{inl} } \, \dalnt{a}
                                   \mid \operatorname{\mathsf{inr} } \, \dalnt{a}
      \\ &
                                   \mid \operatorname{\mathsf{case} } _{ \dalnt{q} }  \dalnt{a}   \operatorname{\mathsf{of} }   \dalnt{b_{{\mathrm{1}}}}  ;  \dalnt{b_{{\mathrm{2}}}}
                                   \mid \uparrow_{ \mathfrak{m}_{{\mathrm{1}}} }^{ \mathfrak{m}_{{\mathrm{2}}} }\!\!  \dalnt{a}
                                   \mid \downarrow_{ \mathfrak{m}_{{\mathrm{1}}} }^{ \mathfrak{m}_{{\mathrm{2}}} }\!\!  \dalnt{a}
    \end{align*}
  \end{mdframed}
  \caption{Syntax of types and terms in \glad{}}
  \label{fig:dal-syntax}
\end{figure}

The syntax of types and terms in \glad{} is given in \Cref{fig:dal-syntax}.
It strongly resembles that of \dmGL{},
with a unit type $ \mathbf{I}_{ \mathfrak{m} } $,
dependent function and pair types $ (  \dalmv{x}  :^{ \dalnt{r}  :  \mathfrak{m} }  \dalnt{A}  )  \multimap   \dalnt{B} $ and $ (  \dalmv{x}  :^{ \dalnt{r}  :  \mathfrak{m} }  \dalnt{A}  )  \otimes   \dalnt{B} $
and a coproduct type $ \dalnt{A}  \oplus  \dalnt{B} $.
The modal operators $ \uparrow_{ \mathfrak{m}_{{\mathrm{1}}} }^{ \mathfrak{m}_{{\mathrm{2}}} }\!\!  \dalnt{a} $ and $ \downarrow_{ \mathfrak{m}_{{\mathrm{1}}} }^{ \mathfrak{m}_{{\mathrm{2}}} }\!\!  \dalnt{a} $
take the role of $ \mc G $ and $ \mc G^{-1} $ from before
while the role of the modal operator $ \mc F $ is now subsumed by the
dependent pair type $ (  \dalmv{x}  :^{ \dalnt{r}  :  \mathfrak{m} }  \dalnt{A}  )  \otimes   \dalnt{B} $,
which allows the first and second component of the pair to belong to different modes.
The mode of any type appearing in a valid judgment will always be determined uniquely,
and therefore annotations on types to specify their mode are not necessary.
We choose to annotate the unit type with its mode,
as this will make things easier in the future.

\begin{figure}
  \begin{mdframed}
    \dalrules[ctx]{empty, extend}
  \end{mdframed}
  \caption{\glad{} well-formed context judgment}
  \label{fig:dal-ctx}
\end{figure}

The judgment for well-formed contexts in \glad{} has the form
\[
  \delta  \mid  \mathcal
{M}   \odot   \Gamma   \vdash   \mathsf{ctx}.
\]
The rules for this judgment are given in \Cref{fig:dal-ctx}.
In this judgment form, $ \delta $ is a list of grades,
$ \mathcal
{M} $ is a list of modes, and $ \Gamma $ is a context.
If $ \delta = (\dalnt{r_{{\mathrm{1}}}}, \mathellipsis, \dalnt{r_{\dalmv{n}}}) $,
$ \mathcal
{M} = (\mathfrak{m}_{{\mathrm{1}}}, \mathellipsis, \mathfrak{m}_{\dalmv{n}}) $
and $ \Gamma = \dalmv{x_{{\mathrm{1}}}} : \dalnt{A_{{\mathrm{1}}}} , \mathellipsis , \dalmv{x_{\dalmv{n}}}  : \dalnt{A_{\dalmv{n}}} $,
then the judgment above indicates that $ \dalnt{r_{\dalmv{i}}}  \colon  \mathfrak{m}_{\dalmv{i}} $,
that the variable $ \dalmv{x_{\dalmv{i}}} $ is graded with grade $ \dalnt{r_{\dalmv{i}}} $
and that the type $ \dalnt{A_{\dalmv{i}}} $ belongs to mode $ \mathfrak{m}_{\dalmv{i}} $.
The same is true for typing judgments:
\[
  \delta  \mid  \mathcal
{M}   \odot   \Gamma   \vdash _{ \mathfrak{m} }  \dalnt{a}  \colon  \dalnt{A}.
\]
The annotation $ \mathfrak{m} $ on the turnstile indicates that the judgment is made in mode $ \mathfrak{m} $,
and that $ \dalnt{a} $ and $ \dalnt{A} $ belong to mode $ \mathfrak{m} $,
therefore supporting the structural rules allowed by $ \mathfrak{m} $.
As in adjoint logic, we demand that in such a judgment, we have $ \mathfrak{m}  \leq  \mathcal
{M} $,
that is $ \mathfrak{m}  \leq  \mathfrak{n} $ for each mode $ \mf n $ appearing in $ \mathcal
{M} $.
This property is enforced by the typing rules:
If we assume that all premise judgments of a rule satisfy this property,
then the conclusion necessarily does, too.

\begin{figure}
  \begin{mdframed}
    \dalrules[glad]{type, unit, function, tensor, coproduct, shiftUp}
  \end{mdframed}
  \caption{\glad{} type formation rules}
  \label{fig:dal-types}
\end{figure}

\begin{figure}
  \begin{mdframed}
    \dalrules[glad]{var, weak, subusage, unitIntro, unitElim, lambda, app}
  \end{mdframed}
  \caption{\glad{} typing rules}
  \label{fig:dal-terms}
\end{figure}

\begin{figure}
  \begin{mdframed}
    \dalrules[glad]{tensorIntro, tensorElim, inl, inr, coproductElim, raise, unraise}
  \end{mdframed}
  \caption{\glad{} typing rules continued}
  \label{fig:dal-terms-ctd}
\end{figure}

We give the full rules for type formation in \Cref{fig:dal-types} and
those for typing in \Cref{fig:dal-terms} and \Cref{fig:dal-terms-ctd}.
The rules strongly resemble the ones of \dmGL{}, so we will omit most explanations,
focusing primarily on the differences.
\glad{} has more general control over weakening:
We may add unused variables of mode $ \mathfrak{m} $ to the context,
so long as they are graded with the grade $ 0 $ and the mode $ \mathfrak{m} $ allows weakening.
If $ \mathfrak{m} $ is a mode for which $ \mathsf{Weak} (  \mathfrak{m}  ) $ is true,
the functionality of the preorder on $ R_{ \mathfrak{m} } $ is overloaded in the following way:
It captures both the subusaging relation and gives control over grades which may be computationally irrelevant.
For nonzero elements of $ R_{ \mathfrak{m}} $, the preorder captures the subusaging behavior,
allowing to use a higher grade of resources than necessary to construct a term.
However, for elements $ q $ with $ 0 \leq q $,
the preorder captures that variables graded with $ q $ may be discarded.

Some of the rules feature a scalar multiplication $ \dalnt{r}   \cdot   \delta $.
There are two things to note here:
First, the vector $ \delta $ is a list grades coming from potentially different semirings,
but this doesn't affect the definition of scalar multiplication
as entrywise multiplication with $ \dalnt{r} $.
Second, let $ \mathfrak{m}_{\dalmv{i}} $ be the mode of the $ i $-th entry of $ \delta $
and $ \mathfrak{m} $ the mode of $ \dalnt{r} $.
Observe that the rules where scalar multiplication occurs
guarantee that $ \mathfrak{m}  \leq  \mathfrak{m}_{\dalmv{i}} $ for each~$ i $ and therefore the scalar multiplication
is indeed well-defined according to \Cref{defn:semiring-hom}.
Similarly, some rules feature addition of grade vectors $ \delta  \dalsym{+}  \delta' $.
Notice that whenever this is the case, the annotation by mode vectors $ \mathcal
{M} $ forces the
$ i $-th entries of $ \delta $ and $ \delta' $ to belong to the same mode,
ensuring that the sums of grade vectors are well-defined.
Finally, the subusage rule features the preorder relation $ \delta_{{\mathrm{1}}}  \leq  \delta_{{\mathrm{2}}} $ on grade vectors.
This is defined in the natural way as componentwise $ \leq $,
with the additional condition that for each $ i $, the $ i $-th component of $ \delta_{{\mathrm{1}}} $ and $ \delta_{{\mathrm{2}}} $
must be from the same mode.

The dependent pair type $ (  \dalmv{x}  :^{ \dalnt{r}  :  \mathfrak{m} }  \dalnt{A}  )  \otimes   \dalnt{B} $
now carries an additional mode annotation,
indicating that type $ A $ belongs to mode $ \mathfrak{m} $
and that $ \dalnt{r}  \colon  \mathfrak{m} $.
Due to the construction of the dependent pair type, if
$ \delta  \mid  \mathcal
{M}   \odot   \Gamma   \vdash _{ \mathfrak{n} }   (  \dalmv{x}  :^{ \dalnt{r}  :  \mathfrak{m} }  \dalnt{A}  )  \otimes   \dalnt{B}   \colon  \mathsf{Type} $,
then $ \dalnt{B} $ and $ (  \dalmv{x}  :^{ \dalnt{r}  :  \mathfrak{m} }  \dalnt{A}  )  \otimes   \dalnt{B} $ belong to mode $ \mf m $.
If $ \mf n $ and $ \mathfrak{m} $ are the same mode,
then we recover ordinary versions of the dependent pair type
belonging to mode $ \mathfrak{m} $.
On the other hand, if the modes are distinct,
the first component of the pair is ``moved'' from the higher mode $ \mathfrak{m} $
to the lower mode $ \mf n $.
This is the way the left adjoint $ \mc F(x :^r X).A $ functions in \dmGL{}.
Since the dependent pair and left adjoint $ \mc F $ have the same introduction
and elimination rules (mutandis mutatis),
we treat them as special instances of the same type in \glad{},
subsuming both functionalities.

The dependent function type $ (  \dalmv{x}  :^{ \dalnt{r}  :  \mathfrak{m} }  \dalnt{A}  )  \multimap   \dalnt{B} $ also carries a mode annotation now.
This is because we allow the modes of $ A $ and $ B $ to be distinct,
and the annotation indicates the mode of $ A $.
Like with the dependent pair type, if $ B $ has mode $ \mf n $, then so does
$ (  \dalmv{x}  :^{ \dalnt{r}  :  \mathfrak{m} }  \dalnt{A}  )  \multimap   \dalnt{B} $ and we necessarily have $ \mf m \ge \mf n $.
A similar construction exists in \LNLD{} \cite{Krishnaswami2015},
where there are two dependent function types,
one between intuitionistic types and one
whose functions take intuitionistic arguments and produce linear terms.

\iffalse
The dependent function type $ (  \dalmv{x}  :^{ \dalnt{r}  :  \mathfrak{m} }  \dalnt{A}  )  \multimap   \dalnt{B} $ and
dependent pair type $ (  \dalmv{x}  :^{ \dalnt{r}  :  \mathfrak{m} }  \dalnt{A}  )  \otimes   \dalnt{B} $
now carry an additional mode annotation $ \mathfrak{m} $,
indicating that the type $ \dalnt{A} $ belongs to mode $ \mathfrak{m} $ and that $ \dalnt{r} \in R_\mathfrak{m} $.
Due to the construction of these types, if
$ \delta  \mid  \mathcal
{M}   \odot   \Gamma   \vdash _{ \mathfrak{n} }   (  \dalmv{x}  :^{ \dalnt{r}  :  \mathfrak{m} }  \dalnt{A}  )  \otimes   \dalnt{B}   \colon  \mathsf{Type} $,
then $ \dalnt{B} $ and $ (  \dalmv{x}  :^{ \dalnt{r}  :  \mathfrak{m} }  \dalnt{A}  )  \otimes   \dalnt{B} $ belong to mode $ \mf n $
and we necessarily have $ \mathfrak{n}  \leq  \mathfrak{m} $.
The same holds for the dependent function type.
If $ \mf n $ and $ \mathfrak{m} $ are the same mode,
then we recover ordinary versions of the dependent pair type
belonging to mode $ \mathfrak{m} $.
On the other hand, if the modes are distinct,
the first component of the pair is ``moved'' from the higher mode $ \mathfrak{m} $
to the lower mode $ \mf n $.
This is the way the left adjoint $ \mc F(x :^r X).A $ functions in \dmGL{}.
Since the dependent pair and left adjoint $ \mc F $ have the same introduction
and elimination rules (mutandis mutatis),
we treat them as special instances of the same type in \glad{},
subsuming both functionalities.
\fi

\subsection*{Metatheory}

We discuss some metatheory of \glad{},
and return to the point of controlling contraction in \glad{}.
Substitution holds in \glad{}.
We use \Cref{conv:list-match} again,
but extend it to also imply that lists of modes $ \mathcal
{M} $ have
the same length as the corresponding grade vectors and contexts.

\begin{restatable}[\glad{} substitution]{theorem}{dalsubst}
  \label{thm:dal-subst}
  The following hold by mutual induction:
  \begin{enumerate}[label=\roman*)]
  \item (Contexts)
    If
    $ \delta  \dalsym{,}  \dalnt{r}  \dalsym{,}  \delta'  \mid  \mathcal
{M}  \dalsym{,}  \mathfrak{m}_{{\mathrm{0}}}  \dalsym{,}  \mathcal
{M}'   \odot   \Gamma  \dalsym{,}  \dalmv{x}  \colon  \dalnt{A}  \dalsym{,}  \Gamma'   \vdash   \mathsf{ctx} $
    and
    $ \delta_{{\mathrm{0}}}  \mid  \mathcal
{M}   \odot   \Gamma   \vdash _{ \mathfrak{m}_{{\mathrm{0}}} }  \dalnt{a_{{\mathrm{0}}}}  \colon  \dalnt{A} $
    then
    \[
      \delta  \dalsym{+}   \dalnt{r}   \cdot   \delta_{{\mathrm{0}}}   \dalsym{,}  \delta'  \mid  \mathcal
{M}  \dalsym{,}  \mathcal
{M}'   \odot   \Gamma  \dalsym{,}  \dalsym{[}  \dalnt{a_{{\mathrm{0}}}}  \dalsym{/}  \dalmv{x}  \dalsym{]}  \Gamma'   \vdash   \mathsf{ctx}
    \]

  \item (Terms)
    If
    $ \delta  \dalsym{,}  \dalnt{r}  \dalsym{,}  \delta'  \mid  \mathcal
{M}  \dalsym{,}  \mathfrak{m}_{{\mathrm{0}}}  \dalsym{,}  \mathcal
{M}'   \odot   \Gamma  \dalsym{,}  \dalmv{x}  \colon  \dalnt{A}  \dalsym{,}  \Gamma'   \vdash _{ \mathfrak{m} }  \dalnt{b}  \colon  \dalnt{B} $
    and
    $ \delta_{{\mathrm{0}}}  \mid  \mathcal
{M}   \odot   \Gamma   \vdash _{ \mathfrak{m}_{{\mathrm{0}}} }  \dalnt{a_{{\mathrm{0}}}}  \colon  \dalnt{A} $,
    then
    \[
      \delta  \dalsym{+}   \dalnt{r}   \cdot   \delta_{{\mathrm{0}}}   \dalsym{,}  \delta'  \mid  \mathcal
{M}  \dalsym{,}  \mathcal
{M}'   \odot   \Gamma  \dalsym{,}  \dalsym{[}  \dalnt{a_{{\mathrm{0}}}}  \dalsym{/}  \dalmv{x}  \dalsym{]}  \Gamma'   \vdash _{ \mathfrak{m} }  \dalsym{[}  \dalnt{a_{{\mathrm{0}}}}  \dalsym{/}  \dalmv{x}  \dalsym{]}  \dalnt{b}  \colon  \dalsym{[}  \dalnt{a_{{\mathrm{0}}}}  \dalsym{/}  \dalmv{x}  \dalsym{]}  \dalnt{B}
    \]
  \end{enumerate}
\end{restatable}

The proof is by induction and given in \Cref{app:dal-meta}.
As for \dmGL{} we obtain a graded contraction rule as a corollary.

\begin{cor}[\glad{} contraction]
  If 
  \[
    \delta  \dalsym{,}  \dalnt{r_{{\mathrm{1}}}}  \dalsym{,}  \dalnt{r_{{\mathrm{2}}}}  \dalsym{,}  \delta'  \mid  \mathcal
{M}  \dalsym{,}  \mathfrak{m}  \dalsym{,}  \mathfrak{m}  \dalsym{,}  \mathcal
{M}'   \odot   \Gamma  \dalsym{,}  \dalmv{x}  \colon  \dalnt{A}  \dalsym{,}  \dalmv{y}  \colon  \dalnt{A}  \dalsym{,}  \Gamma'   \vdash _{ \mathfrak{n} }  \dalnt{b}  \colon  \dalnt{B},
  \]
  then
  \[
    \delta  \dalsym{,}  \dalnt{r_{{\mathrm{1}}}}  \dalsym{+}  \dalnt{r_{{\mathrm{2}}}}  \dalsym{,}  \delta'  \mid  \mathcal
{M}  \dalsym{,}  \mathfrak{m}  \dalsym{,}  \mathcal
{M}'   \odot   \Gamma  \dalsym{,}  \dalmv{x}  \colon  \dalnt{A}  \dalsym{,}  \dalsym{[}  \dalmv{x}  \dalsym{/}  \dalmv{y}  \dalsym{]}  \Gamma'   \vdash _{ \mathfrak{n} }  \dalsym{[}  \dalmv{x}  \dalsym{/}  \dalmv{y}  \dalsym{]}  \dalnt{b}  \colon  \dalsym{[}  \dalmv{x}  \dalsym{/}  \dalmv{y}  \dalsym{]}  \dalnt{B}.
  \]
\end{cor}

\begin{proof}
  Apply substitution with the judgment
  \[
    \vec{0}  \dalsym{,}   1   \mid  \mathcal
{M}  \dalsym{,}  \mathfrak{m}   \odot   \Gamma  \dalsym{,}  \dalmv{x}  \colon  \dalnt{A}   \vdash _{ \mathfrak{m} }  \dalmv{x}  \colon  \dalnt{A}
  \]
  obtained from the \dalrulename[glad]{var} rule.
\end{proof}

Upon closer inspection, the thing that makes contraction work in our system,
is the form of the substitution theorem.
In its statement the grades $ \delta_{{\mathrm{0}}} $ for constructing $ a : A $
are added to the ones used in the construction of $ b : B $.
In other words the substitution theorem contains a contraction implicitly.
In fact, contraction is implicitly included in the rules of \glad{}:
For example in the rule \dalrulename[glad]{unitElim},
the grade vectors $ \delta $ and $ \delta' $ are added.
This is implicitly a contraction, as we are adding grades instead of concatenating contexts.

But concatenating contexts is not well-suited to the dependently typed setting.
When we concatenate contexts in a simply-typed system which has control over contraction,
there is a renaming of the variables in the context to avoid name clashes.
While such a renaming is possible in a dependently typed system,
it immediately becomes very difficult to type any terms
as the variables occurring in a term may also be part of its type.
For example consider the rule \dalrulename[glad]{app}:
If we opted to rename variables,
in order to know that the application $ \dalnt{c} \, \dalnt{a} $ is well typed,
we would need to remember the fact
that the domain type of $ c $ and the type of $ a $ were equal prior to the renaming.
It is not clear to us at the moment how to handle this renaming
and to how to incorporate it with grades.
We leave an investigation of this for future work.

\subsection*{Examples}
\label{subsec:examples}
In this section, we give some example instantiations of \glad{}
and show how we may recover existing graded and mixed systems.

\begin{exmp}[Recovering \dmGL{}]
  \label{exmp:dmgl}
  We explore the relationship between \dmGL{} and \glad{}.
  Let $ R $ be a semiring.
  We ask the following question:
  How closely can we approximate \dmGL{} graded by $ R $ using modes of \glad{}?
  A natural approach to solving this, is to take \glad{} with two modes,
  one with the semiring $ R $ and one with a semiring that captures linearity.
  But since \glad{} admits contraction,
  while the linear fragment of \dmGL{} does not, the latter semiring cannot exist.
  Because of this, we can only hope to recover a version of \dmGL{}
  where the linear fragment is graded as in \Cref{exmp:var-reuse}.
  This turns out to work:
  We take as modes $ \ms L $ (linear) and $ \ms G $ (graded) with $ \ms L \leq \ms G $
  and set $ R_{\ms G} = R $ and $ R_\ms L = \N $ with the trivial preorder $ n \leq m \iff n = m $.
  We also take $ \ms{Weak}(\ms G) = \True $ and $ \ms{Weak}(\ms L) = \False $.
  There exists a unique morphism of modes $ \phi \from \ms L \to \ms G $.
  \glad{} instantiated with this data produces a system with two fragments,
  one graded by $ R $ and which has the same rules as the graded fragment of \dmGL{}.
  The other fragment corresponds to the mixed fragment of \dmGL{}
  and has assumptions graded by $ R $ as well as by $ \N $,
  with the assumptions graded by $ \N $ behaving in a way that's similar to Bounded Linear Logic.
\end{exmp}

\begin{exmp}[Recovering \LNLD{} \cite{Krishnaswami2015}]
  \label{emxp:lnld}
  Let $ \ms U $ be the unrestricted mode where $ R_{\ms U} = 0 $,
  the trivial semiring with $ 0 = 1 $ and $ \Weak (\ms U) = \True $.
  Furthermore, let $ \ms L $ be the mode
  with $ R_{\ms L} $ the none-one-tons semiring of \Cref{exmp:quant-semiring}
  with the reflexive preorder relation and $ \Weak(\ms L) = \False $.
  Variables in mode $ \ms U $ have no grading information attached to them,
  and therefore can be used intuitionistically.
  On the other hand, variables in mode $ \ms L $ are used linearly by default
  and may not be weakened or duplicated.
  We have a unique morphism of modes $ \ms L \to \ms U $.
  Instantiating \glad{} with this data allows us to recover a system similar to \LNLD{}.
\end{exmp}

\begin{exmp}
  \label{exmp:relevance}
  In this example we consider two modes $ \ms W, \ms R $
  with the semiring for both modes being the none-one-tons semiring of \Cref{exmp:quant-semiring}.
  We set $ \Weak (\ms W) = \True $ and $ \Weak(\ms R) = \False $.
  We chose the preorder generated by $ 0 \leq \omega $ for $ \ms W $
  and the reflexive preorder for $ \ms R $.
  There is now a morphism of modes $ \ms R \to \ms W $ which
  is the identity morphism on the underlying semirings.
  The mode $ \ms W $ admits weakening, for variables used with grades $ 0 $ or $ \omega $,
  while the mode $ \ms R $ does not.
  In other words, $ \ms R $ is a relevance logic.

  Similar to the \emph{of course} modality $ ! $ of linear logic,
  and its decomposition into two adjoints $ F $ and $ G $ in LNL,
  we can use the mode shifting operators of \glad{}
  to introduce irrelevantly used variables to the mode $ \ms R $ in a controlled way.
\end{exmp}

\begin{exmp}
  \label{exmp:variance}
  Consider the semiring
  $ \ms{Var} = \{ \invar, \covar, \contvar, \unvar \} $,
  with the preorder generated by
  $ \invar \leq \covar, \contvar \leq \unvar $,
  with $ 0 = \unvar $, $ 1 = \covar $,
  addition defined by $ a + b = \inf (a , b) $ the greatest lower bound on $ \{ a, b\} $
  and multiplication determined by the equations
  \begin{gather*}
	\contvar \cdot \contvar = \covar, \\
    \invar \cdot \contvar = \invar \cdot \invar = \invar
  \end{gather*}
  and the requirement that multiplication is commutative.
  This is the \emph{variance} semiring introduced by Wood and Atkey \cite{Wood:2022}
  and it allows to track whether a term depends on a variable
  covariantly ($ \covar $), contravariantly ($ \contvar $),
  invariantly ($ \invar $), or if there are no guarantees ($ \unvar $).
  We define the mode $ \ms V $ to have $ R_{\ms V} = \ms{Var} $
  and $ \Weak {\ms V} = \True $.

  We add two more modes:
  $ \ms L $ with $ R_{\ms L} = \N $ and
  $ \ms M $ with $ R_{\ms M} $ the none-one-tons semiring.
  We take the preorders on these semirings to be the trivial reflexive preorders.
  Furthermore, we set $ \Weak \ms L = \False $ and $ \Weak \ms M = \True $.
  There are unique morphisms of modes
  $ \ms L \to \ms M \to \ms V $.
\end{exmp}
% subsection examples (end)

\section{Discussion, Future Work, Related Work}
\label{sec:discussion}
\subsection*{Related Work}

Closely related to $\GlaD$ is the framework of Licata et
al.~\cite{Licata:2019}.  Their system is a simply typed linear sequent
calculus equipped with a mode theory where every formula and judgment
is annotated with a mode that constrains the structural rules allowed
within the context.  The modes found in $\GlaD$ are not as elaborate
as the mode theory found in their system.  Their modes are generic and
have morphisms between them, but our modes are specifically semirings.
In addition, $\GlaD$ is dependently typed.

$\GlaD$ is based on $\GraD$ of Choudhury et al.~\cite{Choudhury2021}.
$\GraD$ is essentially identical to the graded side of $\dmGL$ with
the addition of the modal operators.  However, $\GlaD$ differs
substantially from $\GraD$ in that the former now supports multiple
semirings and a theory of modes.

The Graded Modal Dependent Type Theory (GMDTT) of Moon et
al.~\cite{Moon2021} is very similar to the graded side of $\dmGL$ and
our second system $\GlaD$, but GMDTT strives to track resource usage
in types as well as terms where $\GlaD$ and $\dmGL$ only tracks
resource usage in terms. In addition, GMDTT is not based on the theory
of adjoints and modal operators in line with LNL and Adjoint Logic.

Gratzer et al.~\cite{10.1145/3373718.3394736} propose modal dependent
type theory which uses modes to support the embedding of a family of
modal logics.  Their mode theory is similar to the one found in
$\GlaD$, but our system's goal is to relate graded type systems and
theirs is to relate modal logics.

\subsection*{Simply Typed Version with Control over Contraction}

We presented \glad{} as a dependently typed system with grading and
modes, similar to adjoint logic.
The fact that \glad{} is dependently typed makes it difficult to
control for contraction in a manner that's similar to adjoint logic,
and we have given an argument for why this is the case.
It appears that the difficulties with controlling graded contraction
disappear if one considers a simply typed system instead.
As a next step, we will investigate a simply typed version of \glad{}
in which we can control for contraction.
We will take the approach of equipping a mode $ \mf m $ with a subset 
$ \Cont(\mf m) \subseteq R_{\mf m} $ which is closed under addition,
but may also need to satisfy other algebraic properties.
We can then introduce an explicit graded contraction rule such as
\[
  \inferrule
  {
    r, q \in \Cont(\mf m)\\
    \delta , r , q , \delta' \mid
    \mc M , \mf m , \mf m , \mc M'
    \at \Gamma, x : A, y : A, \Gamma'
    \proves_{\mf n}
    b : B
  }
  {
    \delta , r + q , \delta' \mid
    \mc M , \mf m , \mc M'
    \at \Gamma , x : A , \Gamma'
    \proves_{\mf n}
    [x / y] b : B
  }
\]

\subsection*{Categorical Semantics}

Categorical semantics for dependent graded type systems are not well explored at the time of writing.
The only approach know to us is presented for Atkey's QTT \cite{Atkey2018}.
However, Katsumata \cite{Katsumata2018} has developed a general approach to the categorical semantics
using graded linear exponential comonads and formulates the coherence conditions on
such comonads in a compact way using double categories.
In our preliminary considerations on the categorical semantics of \glad{},
we have recovered Katsumata's approach exactly.

A common approach to categorical semantics of dependent type theory is
through categories with families \cite{Dybjer1996} and this is also the approach taken
by Atkey via quantitative categories with families (QCwF's).
A similar approach is to use comprehension categories \cite{Jacobs1991}.
We believe the latter to be slightly nicer,
as the category of comprehension categories embeds into
the $ 2 $-category of cartesian fibrations over a base category $ \mc B $
and this $ 2 $-category enjoys nice properties.
Furthermore,
a comprehension category can be very compactly described as morphism of cartesian fibrations
\[
  \begin{tikzcd}[column sep = small]
    \mc E
    \arrow[rd, "p"']
    \arrow[rr, "P"]
    &&
    \mc B^{\rightarrow}
    \arrow[ld, "\ms{cod}"]
    \\
    &
    \mc B
  \end{tikzcd}
\]
with $ \mc B $ cartesian closed,
$ \mc B^{\rightarrow} $ the category of arrows in $ \mc B $
and $ \ms{cod} $ the codomain fibration.
In this regard, one minor criticism we have of Atkey's QCwF's,
is that they do not (or at least are not know to) arise as an instantiation
of a more general categorical concept, like CwF's do with fibrations.

In category theory, it is often helpful to formulate specific concepts
as instances of more general ones.
Our goal for the future is to combine the general categorical pictures provided by
Katsumata's graded linear exponential comonads and comprehension categories
to develop categorical semantics for graded dependent type theory.

\subsection*{Conclusion}

In the present work we have presented two graded dependent type systems.
The first was a obtained by replacing the dependent fragment of \LNLD{}
with the graded dependent type system \GraD{}.
The second type system is a further generalization of the first,
allowing different assumptions to be graded by grades coming from different semirings.
This system resembles adjoint logic in its structure, and employs a similar construct of \emph{modes}.
We proved meta-theoretic properties of these systems:
For the former we proved substitution and presented a reduction relation
which we showed to preserve grading and types.
The latter system was proven to admit substitution and full graded contraction.

\begin{acks}   
   This work is supported by the
   \grantsponsor{1}{National Science Foundation}{https://www.nsf.gov} under Grant No.:~\grantnum{1}{2104535}.
\end{acks}

\bibliographystyle{plain}
\bibliography{refs.bib}

\appendix
\section{Metatheory of \dmGL{}}
\label{app-sec:dmgl}
% TEX root = paper.tex

\ContextWellFormedness*

\begin{lem}
  \label{lem:ctx-change-vector}
  Let $ \delta, \delta' $ be grade vectors of the same length.
  \begin{enumerate}[label=\roman*)]
    \item 
      If $ \delta  \odot  \Delta  \vdash_\mathsf{G} \, \mathsf{ctx} $, then $ \delta'  \odot  \Delta  \vdash_\mathsf{G} \, \mathsf{ctx} $.

    \item
      If $ \delta  \odot  \Delta  \dmGLsym{;}  \Gamma  \vdash_\mathsf{M} \, \mathsf{ctx} $, then $ \delta'  \odot  \Delta  \dmGLsym{;}  \Gamma  \vdash_\mathsf{M} \, \mathsf{ctx} $.
  \end{enumerate}
\end{lem}

\begin{proof}
  \begin{enumerate}[label=\roman*)]
    \item 
      By induction on the derivation of $ \delta  \odot  \Delta  \vdash_\mathsf{G} \, \mathsf{ctx} $.
      Notice that in the rule for graded context extension
      \[
        \inferrule
        {
          \dmGLmv{x} \, \notin \, \sfoperator{dom} \, \Delta \\
          \delta  \odot  \Delta  \vdash_\mathsf{G} \, \mathsf{ctx} \\
          \delta'_{{\mathrm{0}}}  \odot  \Delta  \vdash_\mathsf{G}  \dmGLnt{X}  \dmGLsym{:}   \mathsf{Type}
        }
        {
          \delta ,  \dmGLnt{r}   \odot  \Delta  \dmGLsym{,}  \dmGLmv{x}  \dmGLsym{:}  \dmGLnt{X}  \vdash_\mathsf{G} \, \mathsf{ctx}
        }
      \]
      the grade $ \dmGLnt{r} $ by which we the grade vector is extended is arbitrary.
      So, in the derivation of $ \delta  \odot  \Delta  \vdash_\mathsf{G} \, \mathsf{ctx} $
      ``build up'' $ \delta' $ instead of $ \delta $.

    \item
      Similarly to i), the grade vector is not relevant to extending the linear
      section of a mixed context.
      \qedhere
  \end{enumerate}
\end{proof}

\begin{lem}
  \label{lem:mixed-ctx-formation}
  If $ \delta  \odot  \Delta  \vdash_\mathsf{G} \, \mathsf{ctx} $ then 
  $ \delta  \odot  \Delta  \dmGLsym{;}  \Gamma  \vdash_\mathsf{M} \, \mathsf{ctx} $ is provable
  if and only if for each type $ \dmGLnt{A} $ occurring in $ \Gamma $,
  the judgment $ \delta'  \odot  \Delta  \vdash_\mathsf{G}  \dmGLnt{A}  \dmGLsym{:}   \mathsf{Linear} $ is provable for some $ \delta' $,
  and all variables in $ \Gamma $ are distinct and do not occur in $ \Delta $.
\end{lem}

\begin{proof}
  This is easily seen by induction on the length of $ \Gamma $
  using the rule \drulename[mixedCtx]{extend}.
\end{proof}

\begin{proof}[Proof of \Cref{prop:ctxWellFormed}]

  Straightforward, by mutual induction on the derivations of
  $ \delta  \odot  \Delta  \vdash_\mathsf{G}  \dmGLnt{t}  \dmGLsym{:}  \dmGLnt{X} $ and $ \delta  \odot  \Delta  \dmGLsym{;}  \Gamma  \vdash_\mathsf{M}  \dmGLnt{l}  \dmGLsym{:}  \dmGLnt{A} $ respectively.
  For the graded fragment, most of the work is done by \Cref{lem:ctx-change-vector},
  and \Cref{lem:mixed-ctx-formation} takes care of the mixed fragment.
\end{proof}

\substitutionTheorem*

\begin{proof}
  The proof is by induction over the derivation of the second judgment.
  For each rule, we need to pattern-match the resulting judgment
  with the pattern in the statement of the theorem,
  which we will do at the beginning of each case.
  
  \proofitem{Case \dmGLdruleGXXsubusageName}
  \[
    \inferrule[\dmGLdruleGXXsubusageName{}]
    {
      \delta_{{\mathrm{1}}} ,  \dmGLnt{r_{{\mathrm{1}}}}  ,  \delta'_{{\mathrm{1}}}   \leq    \delta_{{\mathrm{2}}} ,   \dmGLnt{r_{{\mathrm{2}}}}   ,   \delta'_{{\mathrm{2}}}
      \\
      \delta_{{\mathrm{1}}} ,  \dmGLnt{r_{{\mathrm{1}}}}  ,  \delta'_{{\mathrm{1}}}   \odot  \Delta  \dmGLsym{,}  \dmGLmv{x}  \dmGLsym{:}  \dmGLnt{X}  \dmGLsym{,}  \Delta'  \vdash_\mathsf{G}  \dmGLnt{t}  \dmGLsym{:}  \dmGLnt{Y}
    }
    {
      \delta_{{\mathrm{2}}} ,   \dmGLnt{r_{{\mathrm{2}}}}   ,   \delta'_{{\mathrm{2}}}    \odot  \Delta  \dmGLsym{,}  \dmGLmv{x}  \dmGLsym{:}  \dmGLnt{X}  \dmGLsym{,}  \Delta'  \vdash_\mathsf{G}  \dmGLnt{t}  \dmGLsym{:}  \dmGLnt{Y}
    }
  \]
  By induction we know that
  $ \delta_{{\mathrm{1}}}  +  \dmGLnt{r_{{\mathrm{1}}}}  \cdot  \delta_{{\mathrm{0}}} ,  \delta'_{{\mathrm{1}}}   \odot  \Delta  \dmGLsym{,}  \dmGLsym{[}  \dmGLnt{t_{{\mathrm{0}}}}  \dmGLsym{/}  \dmGLmv{x}  \dmGLsym{]}  \Delta'  \vdash_\mathsf{G}  \dmGLsym{[}  \dmGLnt{t_{{\mathrm{0}}}}  \dmGLsym{/}  \dmGLmv{x}  \dmGLsym{]}  \dmGLnt{t}  \dmGLsym{:}  \dmGLsym{[}  \dmGLnt{t_{{\mathrm{0}}}}  \dmGLsym{/}  \dmGLmv{x}  \dmGLsym{]}  \dmGLnt{Y} $
  and since $ + , \cdot $ are monotonic,
  we have $ \delta_{{\mathrm{1}}}  +  \dmGLnt{r_{{\mathrm{1}}}}  \cdot  \delta_{{\mathrm{0}}}  \leq  \delta_{{\mathrm{2}}}  +  \dmGLnt{r_{{\mathrm{2}}}}  \cdot  \delta_{{\mathrm{0}}} $
  and hence also $ \delta_{{\mathrm{1}}}  +   \dmGLnt{r_{{\mathrm{1}}}}   \cdot   \delta_{{\mathrm{0}}}  ,  \delta'_{{\mathrm{1}}}   \leq   \delta_{{\mathrm{2}}}  +   \dmGLnt{r_{{\mathrm{2}}}}   \cdot   \delta_{{\mathrm{0}}}  ,  \delta'_{{\mathrm{2}}} $.
  Applying \dmGLdruleGXXsubusageName{} concludes this case.

  \proofitem{Case \dmGLdruleGXXweakName}
  We need to match the judgment
  $ \delta ,  0   \odot  \Delta  \dmGLsym{,}  \dmGLmv{z}  \dmGLsym{:}  \dmGLnt{Z}  \vdash_\mathsf{G}  \dmGLnt{t}  \dmGLsym{:}  \dmGLnt{Y} $
  with the pattern
  $ \delta ,  \dmGLnt{r}  ,  \delta'   \odot  \Delta  \dmGLsym{,}  \dmGLmv{x}  \dmGLsym{:}  \dmGLnt{X}  \dmGLsym{,}  \Delta'  \vdash_\mathsf{G}  \dmGLnt{t}  \dmGLsym{:}  \dmGLnt{Y} $.
  There are two cases.
  In the first case we have $ \Delta' = \emptyset $ and $ X = Z $.
  Then our judgment was obtained as follows
  \[
	\inferrule
    {
      \dmGLmv{x} \, \notin \, \sfoperator{dom} \, \Delta\\
      \delta  \odot  \Delta  \vdash_\mathsf{G}  \dmGLnt{t}  \dmGLsym{:}  \dmGLnt{Y} \\\\
      \delta'_{{\mathrm{0}}}  \odot  \Delta  \vdash_\mathsf{G}  \dmGLnt{X}  \dmGLsym{:}   \mathsf{Type}
    }
    {
      \delta ,  0   \odot  \Delta  \dmGLsym{,}  \dmGLmv{x}  \dmGLsym{:}  \dmGLnt{X}  \vdash_\mathsf{G}  \dmGLnt{t}  \dmGLsym{:}  \dmGLnt{Y}
    }
    \text{\drulename[G]{weak}}
  \]
  and we must prove
  $ \delta  +   0   \cdot   \delta_{{\mathrm{0}}}   \odot  \Delta  \vdash_\mathsf{G}  \dmGLsym{[}  \dmGLnt{t_{{\mathrm{0}}}}  \dmGLsym{/}  \dmGLmv{x}  \dmGLsym{]}  \dmGLmv{x}  \dmGLsym{:}  \dmGLsym{[}  \dmGLnt{t_{{\mathrm{0}}}}  \dmGLsym{/}  \dmGLmv{x}  \dmGLsym{]}  \dmGLnt{Y} $.
  But we have $ \delta  \odot  \Delta  \vdash_\mathsf{G}  \dmGLnt{t}  \dmGLsym{:}  \dmGLnt{Y} $, and therefore know that the variable
  $ x $ does not occur freely in $ t $ or $ Y $.
  It follows that we need to prove
  $ \delta  \odot  \Delta  \vdash_\mathsf{G}  \dmGLnt{t}  \dmGLsym{:}  \dmGLnt{Y} $, which we know by assumption.
  
  In the second case we have $ \Delta' \neq \emptyset $.
  In this case the judgment was obtained by
  \[
	\inferrule
    {
      \dmGLmv{z} \, \notin \, \sfoperator{dom} \, \dmGLsym{(}  \Delta  \dmGLsym{,}  \dmGLmv{x}  \dmGLsym{:}  \dmGLnt{X}  \dmGLsym{,}  \Delta'  \dmGLsym{)}\\
      \delta ,  \dmGLnt{r}  ,  \delta'   \odot  \Delta  \dmGLsym{,}  \dmGLmv{x}  \dmGLsym{:}  \dmGLnt{X}  \dmGLsym{,}  \Delta'  \vdash_\mathsf{G}  \dmGLnt{t}  \dmGLsym{:}  \dmGLnt{Y}\\
      \delta_{{\mathrm{1}}} ,  \dmGLnt{r_{{\mathrm{1}}}}  ,  \delta'_{{\mathrm{1}}}   \odot  \Delta  \dmGLsym{,}  \dmGLmv{x}  \dmGLsym{:}  \dmGLnt{X}  \dmGLsym{,}  \Delta'  \vdash_\mathsf{G}  \dmGLnt{Z}  \dmGLsym{:}   \mathsf{Type}
    }
    {
      \delta ,  \dmGLnt{r}  ,  \delta'  ,  0   \odot  \Delta  \dmGLsym{,}  \dmGLmv{x}  \dmGLsym{:}  \dmGLnt{X}  \dmGLsym{,}  \Delta'  \dmGLsym{,}  \dmGLmv{z}  \dmGLsym{:}  \dmGLnt{Z}  \vdash_\mathsf{G}  \dmGLnt{t}  \dmGLsym{:}  \dmGLnt{Y}
    }
  \]
  We need to prove
  \[
	\delta  +   \dmGLnt{r}   \cdot   \delta_{{\mathrm{0}}}  ,  \delta'  ,  0   \odot  \Delta  \dmGLsym{,}  \dmGLsym{[}  \dmGLnt{t_{{\mathrm{0}}}}  \dmGLsym{/}  \dmGLmv{x}  \dmGLsym{]}  \Delta'  \dmGLsym{,}  \dmGLmv{z}  \dmGLsym{:}  \dmGLsym{[}  \dmGLnt{t_{{\mathrm{0}}}}  \dmGLsym{/}  \dmGLmv{x}  \dmGLsym{]}  \dmGLnt{Z}  \vdash_\mathsf{G}  \dmGLsym{[}  \dmGLnt{t_{{\mathrm{0}}}}  \dmGLsym{/}  \dmGLmv{x}  \dmGLsym{]}  \dmGLnt{t}  \dmGLsym{:}  \dmGLsym{[}  \dmGLnt{t_{{\mathrm{0}}}}  \dmGLsym{/}  \dmGLmv{x}  \dmGLsym{]}  \dmGLnt{Y}.
  \]
  From the inductive hypothesis we have
  \begin{align*}
    &
      \delta  +   \dmGLnt{r}   \cdot   \delta_{{\mathrm{0}}}  ,  \delta'   \odot  \Delta  \dmGLsym{,}  \dmGLsym{[}  \dmGLnt{t_{{\mathrm{0}}}}  \dmGLsym{/}  \dmGLmv{x}  \dmGLsym{]}  \Delta'  \vdash_\mathsf{G}  \dmGLsym{[}  \dmGLnt{t_{{\mathrm{0}}}}  \dmGLsym{/}  \dmGLmv{x}  \dmGLsym{]}  \dmGLnt{t}  \dmGLsym{:}  \dmGLsym{[}  \dmGLnt{t_{{\mathrm{0}}}}  \dmGLsym{/}  \dmGLmv{x}  \dmGLsym{]}  \dmGLnt{Y}
    \\
    \text{and}
    \quad
    &
      \delta_{{\mathrm{1}}}  +   \dmGLnt{r_{{\mathrm{1}}}}   \cdot   \delta_{{\mathrm{0}}}  ,  \delta'_{{\mathrm{1}}}   \odot  \Delta  \dmGLsym{,}  \dmGLsym{[}  \dmGLnt{t_{{\mathrm{0}}}}  \dmGLsym{/}  \dmGLmv{x}  \dmGLsym{]}  \Delta'  \vdash_\mathsf{G}  \dmGLsym{[}  \dmGLnt{t_{{\mathrm{0}}}}  \dmGLsym{/}  \dmGLmv{x}  \dmGLsym{]}  \dmGLnt{Z}  \dmGLsym{:}   \mathsf{Type}
  \end{align*}
  and applying \drulename[G]{weak} to these judgments produces
  the desired judgment.

  \proofitem{Case \dmGLdruleGXXconvertName}
  Immediate by induction
  and the fact that ``$ \equiv $'' is a congruence relation,
  i.e. $ \dmGLsym{[}  \dmGLnt{t_{{\mathrm{0}}}}  \dmGLsym{/}  \dmGLmv{x}  \dmGLsym{]}  \dmGLnt{Y} \equiv \dmGLsym{[}  \dmGLnt{t_{{\mathrm{0}}}}  \dmGLsym{/}  \dmGLmv{x}  \dmGLsym{]}  \dmGLnt{Y'} $ follows from
  $ \dmGLnt{Y} \equiv \dmGLnt{Y'} $.

  \proofitem{Case \drulename[G]{var}}
  We need to pattern-match the judgment
  $ \vec{0} ,   1    \odot  \Delta  \dmGLsym{,}  \dmGLmv{z}  \dmGLsym{:}  \dmGLnt{Z}  \vdash_\mathsf{G}  \dmGLmv{z}  \dmGLsym{:}  \dmGLnt{Z} $
  with the pattern
  $ \delta ,  \dmGLnt{r}  ,  \delta'   \odot  \Delta  \dmGLsym{,}  \dmGLmv{x}  \dmGLsym{:}  \dmGLnt{X}  \dmGLsym{,}  \Delta'  \vdash_\mathsf{G}  \dmGLnt{t}  \dmGLsym{:}  \dmGLnt{Y} $.
  There are two possibilities to do this.
  In the first, $ \Delta' = \emptyset $ and $ X = Y = Z $.
  In this case the judgment we are performing induction on was obtained thus
  \[
	\inferrule{
      \dmGLmv{x} \, \notin \, \sfoperator{dom} \, \Delta\\
      \delta_{{\mathrm{1}}}  \odot  \Delta  \vdash_\mathsf{G}  \dmGLnt{X}  \dmGLsym{:}   \mathsf{Type}
    }
    {
      \vec{0} ,   1    \odot  \Delta  \dmGLsym{,}  \dmGLmv{x}  \dmGLsym{:}  \dmGLnt{X}  \vdash_\mathsf{G}  \dmGLmv{x}  \dmGLsym{:}  \dmGLnt{X}
    }
  \]
  and we must prove
  $ \vec{0}  +    1    \cdot   \delta_{{\mathrm{0}}}   \odot  \Delta  \vdash_\mathsf{G}  \dmGLsym{[}  \dmGLnt{t}  \dmGLsym{/}  \dmGLmv{x}  \dmGLsym{]}  \dmGLmv{x}  \dmGLsym{:}  \dmGLsym{[}  \dmGLnt{t}  \dmGLsym{/}  \dmGLmv{x}  \dmGLsym{]}  \dmGLnt{X} $.
  From $ \delta_{{\mathrm{1}}}  \odot  \Delta  \vdash_\mathsf{G}  \dmGLnt{X}  \dmGLsym{:}   \mathsf{Type} $ and $ \dmGLmv{x} \, \notin \, \sfoperator{dom} \, \Delta $,
  we know that $ x $ is not free in $ X $.
  Therefore, we need to prove
  $ \delta_{{\mathrm{0}}}  \odot  \Delta  \vdash_\mathsf{G}  \dmGLnt{t}  \dmGLsym{:}  \dmGLnt{X} $, which is an assumption.

  In the second case,
  the judgment we are performing induction on was obtained as
  \[
	\inferrule{
      \dmGLmv{z} \, \notin \, \sfoperator{dom} \, \dmGLsym{(}  \Delta  \dmGLsym{,}  \dmGLmv{x}  \dmGLsym{:}  \dmGLnt{X}  \dmGLsym{,}  \Delta'  \dmGLsym{)}\\
      \delta ,  \dmGLnt{r}  ,  \delta'   \odot  \Delta  \dmGLsym{,}  \dmGLmv{x}  \dmGLsym{:}  \dmGLnt{X}  \dmGLsym{,}  \Delta'  \vdash_\mathsf{G}  \dmGLnt{Z}  \dmGLsym{:}   \mathsf{Type}
    }
    {
      \vec{0} ,  0  ,  \vec{0}  ,   1    \odot  \Delta  \dmGLsym{,}  \dmGLmv{x}  \dmGLsym{:}  \dmGLnt{X}  \dmGLsym{,}  \Delta'  \dmGLsym{,}  \dmGLmv{z}  \dmGLsym{:}  \dmGLnt{Z}  \vdash_\mathsf{G}  \dmGLmv{z}  \dmGLsym{:}  \dmGLnt{Z}
    }
  \]
  and we need to prove
  $ \vec{0} ,  \vec{0}  ,   1    \odot  \Delta  \dmGLsym{,}  \dmGLsym{[}  \dmGLnt{t_{{\mathrm{0}}}}  \dmGLsym{/}  \dmGLmv{z}  \dmGLsym{]}  \Delta'  \dmGLsym{,}  \dmGLmv{z}  \dmGLsym{:}  \dmGLsym{[}  \dmGLnt{t_{{\mathrm{0}}}}  \dmGLsym{/}  \dmGLmv{x}  \dmGLsym{]}  \dmGLnt{Z}  \vdash_\mathsf{G}  \dmGLsym{[}  \dmGLnt{t_{{\mathrm{0}}}}  \dmGLsym{/}  \dmGLmv{x}  \dmGLsym{]}  \dmGLmv{z}  \dmGLsym{:}  \dmGLsym{[}  \dmGLnt{t_{{\mathrm{0}}}}  \dmGLsym{/}  \dmGLmv{x}  \dmGLsym{]}  \dmGLnt{Z} $.
  But this follows from the inductive hypothesis
  $ \delta  +   \dmGLnt{r}   \cdot   \delta_{{\mathrm{0}}}  ,  \delta'   \odot  \Delta  \dmGLsym{,}  \Delta'  \vdash_\mathsf{G}  \dmGLsym{[}  \dmGLnt{t_{{\mathrm{0}}}}  \dmGLsym{/}  \dmGLmv{x}  \dmGLsym{]}  \dmGLnt{Z}  \dmGLsym{:}   \mathsf{Type} $
  the variable rule and the fact that $ z = \dmGLsym{[}  \dmGLnt{t_{{\mathrm{0}}}}  \dmGLsym{/}  \dmGLmv{x}  \dmGLsym{]}  \dmGLmv{z} $.

  \proofitem{Cases
    \drulename[G]{type},
    \drulename[G]{linear},
    \drulename[G]{unit},
    \drulename[G]{unitIntro}}
  Trivial.

  \proofitem{Case \dmGLdruleGXXunitElimName}
  \[
    \inferrule
    {
      \delta_{{\mathrm{1}}} ,  \dmGLnt{r_{{\mathrm{1}}}}  ,  \delta'_{{\mathrm{1}}}   \odot  \Delta  \dmGLsym{,}  \dmGLmv{x}  \dmGLsym{:}  \dmGLnt{X}  \dmGLsym{,}  \Delta'  \vdash_\mathsf{G}  \dmGLnt{t_{{\mathrm{1}}}}  \dmGLsym{:}  \mathbf{J}
      \\
      \delta_{{\mathrm{2}}} ,  \dmGLnt{r_{{\mathrm{2}}}}  ,  \delta'_{{\mathrm{2}}}  ,  \dmGLnt{q_{{\mathrm{2}}}}   \odot  \Delta  \dmGLsym{,}  \dmGLmv{x}  \dmGLsym{:}  \dmGLnt{X}  \dmGLsym{,}  \Delta'  \dmGLsym{,}  \dmGLmv{x'}  \dmGLsym{:}  \mathbf{J}  \vdash_\mathsf{G}  \dmGLnt{X'}  \dmGLsym{:}   \mathsf{Type}
      \\
      \delta_{{\mathrm{3}}} ,  \dmGLnt{r_{{\mathrm{3}}}}  ,  \delta'_{{\mathrm{3}}}   \odot  \Delta  \dmGLsym{,}  \dmGLmv{x}  \dmGLsym{:}  \dmGLnt{X}  \dmGLsym{,}  \Delta'  \vdash_\mathsf{G}  \dmGLnt{t_{{\mathrm{3}}}}  \dmGLsym{:}  \dmGLsym{[}  \mathbf{j}  \dmGLsym{/}  \dmGLmv{x}  \dmGLsym{]}  \dmGLnt{X'}
    }
    {
      \delta_{{\mathrm{1}}}  +  \delta_{{\mathrm{3}}} ,  \dmGLnt{r_{{\mathrm{1}}}}  +  \dmGLnt{r_{{\mathrm{3}}}}  ,  \delta'_{{\mathrm{1}}}   +  \delta'_{{\mathrm{3}}} \odot \Delta  \dmGLsym{,}  \dmGLmv{x}  \dmGLsym{:}  \dmGLnt{X}  \dmGLsym{,}  \Delta' \\
      \vdash_\mathsf{G} \sfoperator{let} \, \mathbf{j} \, \dmGLsym{=}  \dmGLnt{t_{{\mathrm{1}}}} \, \sfoperator{in} \, \dmGLnt{t_{{\mathrm{3}}}} : \dmGLsym{[}  \dmGLnt{t_{{\mathrm{1}}}}  \dmGLsym{/}  \dmGLmv{x'}  \dmGLsym{]}  \dmGLnt{X'}
    }
  \]
  By induction we have
  \begin{mathpar}
    \delta_{{\mathrm{1}}}  +   \dmGLnt{r_{{\mathrm{1}}}}   \cdot   \delta_{{\mathrm{0}}}  ,  \delta'_{{\mathrm{1}}}   \odot  \Delta  \dmGLsym{,}  \dmGLsym{[}  \dmGLnt{t_{{\mathrm{0}}}}  \dmGLsym{/}  \dmGLmv{x}  \dmGLsym{]}  \Delta'  \vdash_\mathsf{G}  \dmGLsym{[}  \dmGLnt{t_{{\mathrm{0}}}}  \dmGLsym{/}  \dmGLmv{x}  \dmGLsym{]}  \dmGLnt{t_{{\mathrm{1}}}}  \dmGLsym{:}  \mathbf{J}
    \and
    \delta_{{\mathrm{2}}}  +   \dmGLnt{r_{{\mathrm{2}}}}   \cdot   \delta_{{\mathrm{0}}}  ,  \delta'_{{\mathrm{2}}}  ,  \dmGLnt{q_{{\mathrm{2}}}}   \odot  \Delta  \dmGLsym{,}  \dmGLsym{[}  \dmGLnt{t_{{\mathrm{0}}}}  \dmGLsym{/}  \dmGLmv{x}  \dmGLsym{]}  \Delta'  \dmGLsym{,}  \dmGLmv{x'}  \dmGLsym{:}  \mathbf{J}  \vdash_\mathsf{G}  \dmGLsym{[}  \dmGLnt{t_{{\mathrm{0}}}}  \dmGLsym{/}  \dmGLmv{x}  \dmGLsym{]}  \dmGLnt{X'}  \dmGLsym{:}   \mathsf{Type}
    \and
    \delta_{{\mathrm{3}}}  +   \dmGLnt{r_{{\mathrm{3}}}}   \cdot   \delta_{{\mathrm{0}}}  ,  \delta'_{{\mathrm{3}}}   \odot  \Delta  \dmGLsym{,}  \dmGLsym{[}  \dmGLnt{t_{{\mathrm{0}}}}  \dmGLsym{/}  \dmGLmv{x}  \dmGLsym{]}  \Delta'  \vdash_\mathsf{G}  \dmGLsym{[}  \dmGLnt{t_{{\mathrm{0}}}}  \dmGLsym{/}  \dmGLmv{x}  \dmGLsym{]}  \dmGLnt{t_{{\mathrm{3}}}}  \dmGLsym{:}  \dmGLsym{[}  \dmGLnt{t_{{\mathrm{0}}}}  \dmGLsym{/}  \dmGLmv{x}  \dmGLsym{]}  \dmGLsym{[}  \mathbf{j}  \dmGLsym{/}  \dmGLmv{x'}  \dmGLsym{]}  \dmGLnt{X'}
  \end{mathpar}
  Applying \dmGLdruleGXXunitElimName to these judgments yields
  \begin{gather*}
    \delta_{{\mathrm{1}}}  +   \dmGLnt{r_{{\mathrm{1}}}}   \cdot   \delta_{{\mathrm{0}}}   +  \delta_{{\mathrm{3}}}  +   \dmGLnt{r_{{\mathrm{3}}}}   \cdot   \delta_{{\mathrm{0}}}  ,  \delta'_{{\mathrm{1}}}   +  \delta'_{{\mathrm{3}}}
    \odot \Delta  \dmGLsym{,}  \dmGLsym{[}  \dmGLnt{t_{{\mathrm{0}}}}  \dmGLsym{/}  \dmGLmv{x}  \dmGLsym{]}  \Delta' \\[-1ex]
    \vdash_\mathsf{G} \sfoperator{let} \, \mathbf{j} \, \dmGLsym{=}  \dmGLsym{[}  \dmGLnt{t_{{\mathrm{0}}}}  \dmGLsym{/}  \dmGLmv{x}  \dmGLsym{]}  \dmGLnt{t_{{\mathrm{1}}}} \, \sfoperator{in} \, \dmGLsym{[}  \dmGLnt{t_{{\mathrm{0}}}}  \dmGLsym{/}  \dmGLmv{x}  \dmGLsym{]}  \dmGLnt{t_{{\mathrm{3}}}}
    : \dmGLsym{[}  \dmGLnt{t_{{\mathrm{0}}}}  \dmGLsym{/}  \dmGLmv{x}  \dmGLsym{]}  \dmGLsym{[}  \mathbf{j}  \dmGLsym{/}  \dmGLmv{x'}  \dmGLsym{]}  \dmGLnt{X'}
  \end{gather*}
  Since $ \mathbf{j} $ has no free variables, the terms
  $ \dmGLsym{[}  \dmGLnt{t_{{\mathrm{0}}}}  \dmGLsym{/}  \dmGLmv{x}  \dmGLsym{]}  \dmGLsym{(}  \sfoperator{let} \, \mathbf{j} \, \dmGLsym{=}  \dmGLnt{t_{{\mathrm{1}}}} \, \sfoperator{in} \, \dmGLnt{t_{{\mathrm{3}}}}  \dmGLsym{)} $ and
  $ \sfoperator{let} \, \mathbf{j} \, \dmGLsym{=}  \dmGLsym{[}  \dmGLnt{t_{{\mathrm{0}}}}  \dmGLsym{/}  \dmGLmv{x}  \dmGLsym{]}  \dmGLnt{t_{{\mathrm{1}}}} \, \sfoperator{in} \, \dmGLsym{[}  \dmGLnt{t_{{\mathrm{0}}}}  \dmGLsym{/}  \dmGLmv{x}  \dmGLsym{]}  \dmGLnt{t_{{\mathrm{3}}}} $
  are equal, which concludes this case.

  \proofitem{Case \dmGLdruleGXXradjIntroName}
  \[
    \inferrule{
      \delta ,  \dmGLnt{r}  ,  \delta'   \odot  \Delta  \dmGLsym{,}  \dmGLmv{x}  \dmGLsym{:}  \dmGLnt{X}  \dmGLsym{,}  \Delta'  \dmGLsym{;}  \emptyset  \vdash_\mathsf{M}  \dmGLnt{l}  \dmGLsym{:}  \dmGLnt{A}
    }
    {
      \delta ,  \dmGLnt{r}  ,  \delta'   \odot  \Delta  \dmGLsym{,}  \dmGLmv{x}  \dmGLsym{:}  \dmGLnt{X}  \dmGLsym{,}  \Delta'  \vdash_\mathsf{G}  \mathcal{G} \, \dmGLnt{l}  \dmGLsym{:}  \mathcal{G} \, \dmGLnt{A}
    }
  \]
  Applying the induction hypothesis to the premise of the rule,
  we get
  \[
    \delta ,   \dmGLnt{r}   \cdot   \delta_{{\mathrm{0}}}   ,  \delta'   \odot  \Delta  \dmGLsym{,}  \dmGLsym{[}  \dmGLnt{t_{{\mathrm{0}}}}  \dmGLsym{/}  \dmGLmv{x}  \dmGLsym{]}  \Delta'  \dmGLsym{;}  \emptyset  \vdash_\mathsf{M}  \dmGLsym{[}  \dmGLnt{t_{{\mathrm{0}}}}  \dmGLsym{/}  \dmGLmv{x}  \dmGLsym{]}  \dmGLnt{l}  \dmGLsym{:}  \dmGLsym{[}  \dmGLnt{t_{{\mathrm{0}}}}  \dmGLsym{/}  \dmGLmv{x}  \dmGLsym{]}  \dmGLnt{A}
  \]
  Applying \dmGLdruleGXXradjIntroName to this judgment,
  concludes this case, once we observe that
  $ \dmGLsym{[}  \dmGLnt{t_{{\mathrm{0}}}}  \dmGLsym{/}  \dmGLmv{x}  \dmGLsym{]}  \dmGLsym{(}  \mathcal{G} \, \dmGLnt{l}  \dmGLsym{)} $ is equal to $ \mathcal{G} \, \dmGLsym{[}  \dmGLnt{t_{{\mathrm{0}}}}  \dmGLsym{/}  \dmGLmv{x}  \dmGLsym{]}  \dmGLnt{l} $
  and similarly for $ \dmGLnt{A} $ instead of $ \dmGLnt{l} $.

  The remaining cases in the graded fragment follow analogously and
  are therefore omitted.

  \proofitem{Case \dmGLdruleMXXidName}
  In the mixed graded case we have
  \[
	\inferrule
    {
      \dmGLmv{y} \, \notin \, \sfoperator{dom} \, \dmGLsym{(}  \Delta  \dmGLsym{,}  \dmGLmv{x}  \dmGLsym{:}  \dmGLnt{X}  \dmGLsym{,}  \Delta'  \dmGLsym{)}\\
      \delta ,  \dmGLnt{r}  ,  \delta'   \odot  \Delta  \dmGLsym{,}  \dmGLmv{x}  \dmGLsym{:}  \dmGLnt{X}  \dmGLsym{,}  \Delta'  \vdash_\mathsf{G}  \dmGLnt{A}  \dmGLsym{:}   \mathsf{Linear}
    }
    {
      \vec{0} ,  0  ,  \vec{0}   \odot  \Delta  \dmGLsym{,}  \dmGLmv{x}  \dmGLsym{:}  \dmGLnt{X}  \dmGLsym{,}  \Delta'  \dmGLsym{;}  \dmGLmv{y}  \dmGLsym{:}  \dmGLnt{A}  \vdash_\mathsf{M}  \dmGLmv{y}  \dmGLsym{:}  \dmGLnt{A}
    }
  \]
  and we need to prove
  $ \vec{0} ,  \vec{0}   \odot  \Delta  \dmGLsym{,}  \dmGLsym{[}  \dmGLnt{t_{{\mathrm{0}}}}  \dmGLsym{/}  \dmGLmv{x}  \dmGLsym{]}  \Delta'  \dmGLsym{;}  \dmGLmv{y}  \dmGLsym{:}  \dmGLsym{[}  \dmGLnt{t_{{\mathrm{0}}}}  \dmGLsym{/}  \dmGLmv{x}  \dmGLsym{]}  \dmGLnt{A}  \vdash_\mathsf{M}  \dmGLmv{y}  \dmGLsym{:}  \dmGLsym{[}  \dmGLnt{t_{{\mathrm{0}}}}  \dmGLsym{/}  \dmGLmv{x}  \dmGLsym{]}  \dmGLnt{A} $
  which follows immediately from the inductive hypothesis and the
  \drulename[M]{id} rule.
  There is nothing to do for the mixed linear case.

  \proofitem{Case \dmGLdruleMXXsubusageName}
  Both the mixed linear and mixed graded cases are analogous
  to the \dmGLdruleGXXsubusageName case.

  \proofitem{Case \dmGLdruleMXXweakName}
  The mixed graded case is analogous to the purely graded case.
  For the mixed linear case,
  the rule \dmGLdruleMXXweakName pattern matches as
  \[
    \inferrule
    {
      \delta  \odot  \Delta  \dmGLsym{;}   \Gamma  \dmGLsym{,}  \dmGLmv{y}  \dmGLsym{:}  \dmGLnt{A}  ,  \Gamma'   \vdash_\mathsf{M}  \dmGLnt{l}  \dmGLsym{:}  \dmGLnt{B}
      \\
      \delta'  \odot  \Delta  \vdash_\mathsf{G}  \dmGLnt{X}  \dmGLsym{:}   \mathsf{Type}
    }
    {
      \delta ,  0   \odot  \Delta  \dmGLsym{,}  \dmGLmv{x}  \dmGLsym{:}  \dmGLnt{X}  \dmGLsym{;}   \Gamma  \dmGLsym{,}  \dmGLmv{y}  \dmGLsym{:}  \dmGLnt{A}  ,  \Gamma'   \vdash_\mathsf{M}  \dmGLnt{l}  \dmGLsym{:}  \dmGLnt{B}
    }
  \]
  By induction on the first hypothesis we have
  \[
    \delta  \odot  \Delta  \dmGLsym{;}    \Gamma  ,  \Gamma_{{\mathrm{0}}}   ,  \Gamma'   \vdash_\mathsf{M}  \dmGLsym{[}  \dmGLnt{l_{{\mathrm{0}}}}  \dmGLsym{/}  \dmGLmv{y}  \dmGLsym{]}  \dmGLnt{l}  \dmGLsym{:}  \dmGLnt{B}
  \]
  Weakening this by $ X $ produces the desired result.

  \proofitem{Case \dmGLdruleMXXexchangeName}
  The mixed graded case is straightforward and we omit the details.
  For the mixed linear case, consider the rule
  \[
    \inferrule
    {
      \delta  \odot  \Delta  \dmGLsym{;}   \Gamma  \dmGLsym{,}  \dmGLmv{x}  \dmGLsym{:}  \dmGLnt{B_{{\mathrm{1}}}}  \dmGLsym{,}  \dmGLmv{y}  \dmGLsym{:}  \dmGLnt{B_{{\mathrm{2}}}}  ,  \Gamma'   \vdash_\mathsf{M}  \dmGLnt{l}  \dmGLsym{:}  \dmGLnt{C}
    }
    {
      \delta  \odot  \Delta  \dmGLsym{;}   \Gamma  \dmGLsym{,}  \dmGLmv{y}  \dmGLsym{:}  \dmGLnt{B_{{\mathrm{2}}}}  \dmGLsym{,}  \dmGLmv{x}  \dmGLsym{:}  \dmGLnt{B_{{\mathrm{1}}}}  ,  \Gamma'   \vdash_\mathsf{M}  \dmGLnt{l}  \dmGLsym{:}  \dmGLnt{C}
    }\text{\dmGLdruleMXXexchangeName}
  \]
  where we assume that the type $ \dmGLnt{A} $ appears somewhere in the context
  $ \Gamma  \dmGLsym{,}  \dmGLmv{x}  \dmGLsym{:}  \dmGLnt{B_{{\mathrm{1}}}}  \dmGLsym{,}  \dmGLmv{y}  \dmGLsym{:}  \dmGLnt{B_{{\mathrm{2}}}}  ,  \Gamma' $.
  The cases where $ \dmGLnt{A} $ appears in $ \Gamma $ or $ \Gamma' $
  are straightforward.
  The cases that remain are that $ \dmGLnt{A} $ is $ \dmGLnt{B_{{\mathrm{1}}}} $ or $ \dmGLnt{B_{{\mathrm{2}}}} $.
  These cases are parallel,
  so we only treat the case that $ \dmGLnt{A} $ is $ \dmGLnt{B_{{\mathrm{2}}}} $.
  The inductive hypothesis yields
  \[
    \delta  \odot  \Delta  \dmGLsym{;}    \Gamma  \dmGLsym{,}  \dmGLmv{x}  \dmGLsym{:}  \dmGLnt{B_{{\mathrm{1}}}}  ,  \Gamma_{{\mathrm{0}}}   ,  \Gamma'   \vdash_\mathsf{M}  \dmGLsym{[}  \dmGLnt{l_{{\mathrm{0}}}}  \dmGLsym{/}  \dmGLmv{y}  \dmGLsym{]}  \dmGLnt{l}  \dmGLsym{:}  \dmGLnt{C}
  \]
  and we want to show that
  \[
    \delta  \odot  \Delta  \dmGLsym{;}    \Gamma  ,  \Gamma_{{\mathrm{0}}}   \dmGLsym{,}  \dmGLmv{x}  \dmGLsym{:}  \dmGLnt{B_{{\mathrm{1}}}}  ,  \Gamma'   \vdash_\mathsf{M}  \dmGLsym{[}  \dmGLnt{l_{{\mathrm{0}}}}  \dmGLsym{/}  \dmGLmv{y}  \dmGLsym{]}  \dmGLnt{l}  \dmGLsym{:}  \dmGLnt{C}
  \]
  This is done by repeatedly applying the exchange rule
  to move $ \dmGLnt{B_{{\mathrm{1}}}} $ past all types in $ \Gamma_{{\mathrm{0}}} $.

  \proofitem{\drulename[M]{tensorIntro}}
  For the mixed graded case we have
  \[
	\inferrule{
      \delta_{{\mathrm{1}}} ,  \dmGLnt{r_{{\mathrm{1}}}}  ,  \delta'_{{\mathrm{1}}}   \odot  \Delta  \dmGLsym{,}  \dmGLmv{x}  \dmGLsym{:}  \dmGLnt{X}  \dmGLsym{,}  \Delta'  \dmGLsym{;}  \Gamma_{{\mathrm{1}}}  \vdash_\mathsf{M}  \dmGLnt{l_{{\mathrm{1}}}}  \dmGLsym{:}  \dmGLnt{A}\\
      \delta_{{\mathrm{2}}} ,  \dmGLnt{r_{{\mathrm{2}}}}  ,  \delta'_{{\mathrm{2}}}   \odot  \Delta  \dmGLsym{,}  \dmGLmv{x}  \dmGLsym{:}  \dmGLnt{X}  \dmGLsym{,}  \Delta'  \dmGLsym{;}  \Gamma_{{\mathrm{2}}}  \vdash_\mathsf{M}  \dmGLnt{l_{{\mathrm{2}}}}  \dmGLsym{:}  \dmGLnt{B}\\
    }
    {
      \dmGLsym{(}    \delta_{{\mathrm{1}}} ,  \dmGLnt{r_{{\mathrm{1}}}}  ,  \delta'_{{\mathrm{1}}}   \dmGLsym{)}  +  \dmGLsym{(}    \delta_{{\mathrm{2}}} ,  \dmGLnt{r_{{\mathrm{2}}}}  ,  \delta'_{{\mathrm{2}}}   \dmGLsym{)}  \odot  \Delta  \dmGLsym{,}  \dmGLmv{x}  \dmGLsym{:}  \dmGLnt{X}  \dmGLsym{,}  \Delta'  \dmGLsym{;}   \Gamma_{{\mathrm{1}}}  ,  \Gamma_{{\mathrm{2}}}   \vdash_\mathsf{M}  \dmGLsym{(}  \dmGLnt{l_{{\mathrm{1}}}}  \dmGLsym{,}  \dmGLnt{l_{{\mathrm{2}}}}  \dmGLsym{)}  \dmGLsym{:}  \dmGLnt{A}  \otimes  \dmGLnt{B}
    }
  \]
  By the inductive hypothesis we have the judgments
  \begin{gather*}
    \delta_{{\mathrm{1}}}  +   \dmGLnt{r_{{\mathrm{1}}}}   \cdot   \delta_{{\mathrm{0}}}  ,  \delta'_{{\mathrm{1}}}   \odot  \Delta  \dmGLsym{,}  \dmGLsym{[}  \dmGLnt{t_{{\mathrm{0}}}}  \dmGLsym{/}  \dmGLmv{x}  \dmGLsym{]}  \Delta'  \dmGLsym{;}  \dmGLsym{[}  \dmGLnt{t_{{\mathrm{0}}}}  \dmGLsym{/}  \dmGLmv{x}  \dmGLsym{]}  \Gamma_{{\mathrm{1}}}  \vdash_\mathsf{M}  \dmGLsym{[}  \dmGLnt{t_{{\mathrm{0}}}}  \dmGLsym{/}  \dmGLmv{x}  \dmGLsym{]}  \dmGLnt{l_{{\mathrm{1}}}}  \dmGLsym{:}  \dmGLsym{[}  \dmGLnt{t_{{\mathrm{0}}}}  \dmGLsym{/}  \dmGLmv{x}  \dmGLsym{]}  \dmGLnt{A}\\
    \delta_{{\mathrm{2}}}  +   \dmGLnt{r_{{\mathrm{2}}}}   \cdot   \delta_{{\mathrm{0}}}  ,  \delta'_{{\mathrm{2}}}   \odot  \Delta  \dmGLsym{,}  \dmGLsym{[}  \dmGLnt{t_{{\mathrm{0}}}}  \dmGLsym{/}  \dmGLmv{x}  \dmGLsym{]}  \Delta'  \dmGLsym{;}  \dmGLsym{[}  \dmGLnt{t_{{\mathrm{0}}}}  \dmGLsym{/}  \dmGLmv{x}  \dmGLsym{]}  \Gamma_{{\mathrm{2}}}  \vdash_\mathsf{M}  \dmGLsym{[}  \dmGLnt{t_{{\mathrm{0}}}}  \dmGLsym{/}  \dmGLmv{x}  \dmGLsym{]}  \dmGLnt{l_{{\mathrm{2}}}}  \dmGLsym{:}  \dmGLsym{[}  \dmGLnt{t_{{\mathrm{0}}}}  \dmGLsym{/}  \dmGLmv{x}  \dmGLsym{]}  \dmGLnt{B}
  \end{gather*}
  Applying \drulename[M]{tensorIntro} yields
  \begin{align*}
    & \dmGLsym{(}  \delta_{{\mathrm{1}}}  +  \delta_{{\mathrm{2}}}  \dmGLsym{)}  +    \dmGLsym{(}  \dmGLnt{r_{{\mathrm{1}}}}  +  \dmGLnt{r_{{\mathrm{2}}}}  \dmGLsym{)}   \cdot   \delta_{{\mathrm{0}}}   ,  \delta'_{{\mathrm{1}}}   +  \delta'_{{\mathrm{2}}} \\
    & \quad \at \Delta  \dmGLsym{,}  \dmGLsym{[}  \dmGLnt{t_{{\mathrm{0}}}}  \dmGLsym{/}  \dmGLmv{x}  \dmGLsym{]}  \Delta' ; \dmGLsym{[}  \dmGLnt{t_{{\mathrm{0}}}}  \dmGLsym{/}  \dmGLmv{x}  \dmGLsym{]}  \Gamma_{{\mathrm{1}}}  ,  \dmGLsym{[}  \dmGLnt{t_{{\mathrm{0}}}}  \dmGLsym{/}  \dmGLmv{x}  \dmGLsym{]}  \Gamma_{{\mathrm{2}}} \\
    & \quad \proves_{\ms M} \dmGLsym{(}  \dmGLsym{[}  \dmGLnt{t_{{\mathrm{0}}}}  \dmGLsym{/}  \dmGLmv{x}  \dmGLsym{]}  \dmGLnt{l_{{\mathrm{1}}}}  \dmGLsym{,}  \dmGLsym{[}  \dmGLnt{t_{{\mathrm{0}}}}  \dmGLsym{/}  \dmGLmv{x}  \dmGLsym{]}  \dmGLnt{l_{{\mathrm{2}}}}  \dmGLsym{)} : \dmGLsym{(}  \dmGLsym{[}  \dmGLnt{t_{{\mathrm{0}}}}  \dmGLsym{/}  \dmGLmv{x}  \dmGLsym{]}  \dmGLnt{A}  \dmGLsym{)}  \otimes  \dmGLsym{(}  \dmGLsym{[}  \dmGLnt{t_{{\mathrm{0}}}}  \dmGLsym{/}  \dmGLmv{x}  \dmGLsym{]}  \dmGLnt{B}  \dmGLsym{)}
  \end{align*}
  We conclude by noting that the terms
  $ \dmGLsym{(}  \dmGLsym{[}  \dmGLnt{t_{{\mathrm{0}}}}  \dmGLsym{/}  \dmGLmv{x}  \dmGLsym{]}  \dmGLnt{l_{{\mathrm{1}}}}  \dmGLsym{,}  \dmGLsym{[}  \dmGLnt{t_{{\mathrm{0}}}}  \dmGLsym{/}  \dmGLmv{x}  \dmGLsym{]}  \dmGLnt{l_{{\mathrm{2}}}}  \dmGLsym{)} $ and $ \dmGLsym{[}  \dmGLnt{t_{{\mathrm{0}}}}  \dmGLsym{/}  \dmGLmv{x}  \dmGLsym{]}  \dmGLsym{(}  \dmGLnt{l_{{\mathrm{1}}}}  \dmGLsym{,}  \dmGLnt{l_{{\mathrm{2}}}}  \dmGLsym{)} $ are equal,
  and similarly for the types
  $ \dmGLsym{(}  \dmGLsym{[}  \dmGLnt{t_{{\mathrm{0}}}}  \dmGLsym{/}  \dmGLmv{x}  \dmGLsym{]}  \dmGLnt{A}  \dmGLsym{)}  \otimes  \dmGLsym{(}  \dmGLsym{[}  \dmGLnt{t_{{\mathrm{0}}}}  \dmGLsym{/}  \dmGLmv{x}  \dmGLsym{]}  \dmGLnt{B}  \dmGLsym{)} $ and $ \dmGLsym{[}  \dmGLnt{t_{{\mathrm{0}}}}  \dmGLsym{/}  \dmGLmv{x}  \dmGLsym{]}  \dmGLsym{(}  \dmGLnt{A}  \otimes  \dmGLnt{B}  \dmGLsym{)} $.

  For the mixed linear case of this rule we have
  \[
	\inferrule
    {
      \delta_{{\mathrm{1}}}  \odot  \Delta  \dmGLsym{;}  \Gamma_{{\mathrm{1}}}  \vdash_\mathsf{M}  \dmGLnt{l_{{\mathrm{1}}}}  \dmGLsym{:}  \dmGLnt{B_{{\mathrm{1}}}} \\
      \delta_{{\mathrm{2}}}  \odot  \Delta  \dmGLsym{;}  \Gamma_{{\mathrm{2}}}  \vdash_\mathsf{M}  \dmGLnt{l_{{\mathrm{2}}}}  \dmGLsym{:}  \dmGLnt{B_{{\mathrm{2}}}} \\
    }
    {
      \delta_{{\mathrm{1}}}  +  \delta_{{\mathrm{2}}}  \odot  \Delta  \dmGLsym{;}   \Gamma_{{\mathrm{1}}}  ,  \Gamma_{{\mathrm{2}}}   \vdash_\mathsf{M}  \dmGLsym{(}  \dmGLnt{l_{{\mathrm{1}}}}  \dmGLsym{,}  \dmGLnt{l_{{\mathrm{2}}}}  \dmGLsym{)}  \dmGLsym{:}  \dmGLnt{B_{{\mathrm{1}}}}  \otimes  \dmGLnt{B_{{\mathrm{2}}}}
    }
  \]
  and we assume that $ x : A $ occurs somewhere in the linear context $ \Gamma_{{\mathrm{1}}}  ,  \Gamma_{{\mathrm{2}}} $.
  There are two symmetric cases, depending on whether $ x: A $ occurs in $ \Gamma_{{\mathrm{1}}} $ or $ \Gamma_{{\mathrm{2}}} $.
  We consider the case where $ x : A $ occurs in $ \Gamma_{{\mathrm{1}}} $,
  i.e. $ \Gamma_{{\mathrm{1}}} $ has the form $ \Gamma'_{{\mathrm{1}}}  \dmGLsym{,}  \dmGLmv{x}  \dmGLsym{:}  \dmGLnt{A}  ,  \Gamma''_{{\mathrm{1}}} $.
  From the inductive hypothesis we no get
  $ \delta_{{\mathrm{1}}}  \odot  \Delta  \dmGLsym{;}   \Gamma'_{{\mathrm{1}}}  ,  \Gamma''_{{\mathrm{1}}}   \vdash_\mathsf{M}  \dmGLsym{[}  \dmGLnt{l_{{\mathrm{0}}}}  \dmGLsym{/}  \dmGLmv{x}  \dmGLsym{]}  \dmGLnt{l_{{\mathrm{1}}}}  \dmGLsym{:}  \dmGLnt{B_{{\mathrm{1}}}} $.
  Applying \drulename[M]{tensorIntro} now yields
  $ \delta_{{\mathrm{1}}}  +  \delta_{{\mathrm{2}}}  \odot  \Delta  \dmGLsym{;}    \Gamma'_{{\mathrm{1}}}  ,  \Gamma''_{{\mathrm{1}}}   ,  \Gamma_{{\mathrm{2}}}   \vdash_\mathsf{M}  \dmGLsym{(}  \dmGLsym{[}  \dmGLnt{l_{{\mathrm{0}}}}  \dmGLsym{/}  \dmGLmv{x}  \dmGLsym{]}  \dmGLnt{l_{{\mathrm{1}}}}  \dmGLsym{,}  \dmGLnt{l_{{\mathrm{2}}}}  \dmGLsym{)}  \dmGLsym{:}  \dmGLnt{B_{{\mathrm{1}}}}  \otimes  \dmGLnt{B_{{\mathrm{2}}}} $.
  Since we rename variables when concatenating linear contexts,
  the terms $ \dmGLnt{l_{{\mathrm{1}}}} $ and $ \dmGLnt{l_{{\mathrm{2}}}} $ have disjoint sets of free variables,
  and it follows that the terms
  $ \dmGLsym{(}  \dmGLsym{[}  \dmGLnt{l_{{\mathrm{0}}}}  \dmGLsym{/}  \dmGLmv{x}  \dmGLsym{]}  \dmGLnt{l_{{\mathrm{1}}}}  \dmGLsym{,}  \dmGLnt{l_{{\mathrm{2}}}}  \dmGLsym{)} $ and $ \dmGLsym{[}  \dmGLnt{l_{{\mathrm{0}}}}  \dmGLsym{/}  \dmGLmv{x}  \dmGLsym{]}  \dmGLsym{(}  \dmGLnt{l_{{\mathrm{1}}}}  \dmGLsym{,}  \dmGLnt{l_{{\mathrm{2}}}}  \dmGLsym{)} $ are equal.
  This concludes this case.

  \proofitem{\drulename[M]{tensorElim}}
  For the mixed graded case we have
  \[
    \inferrule{
      \delta_{{\mathrm{2}}} ,  \dmGLnt{r_{{\mathrm{2}}}}  ,  \delta'_{{\mathrm{2}}}   \odot  \Delta  \dmGLsym{,}  \dmGLmv{x}  \dmGLsym{:}  \dmGLnt{X}  \dmGLsym{,}  \Delta'  \dmGLsym{;}  \Gamma_{{\mathrm{2}}}  \vdash_\mathsf{M}  \dmGLnt{l_{{\mathrm{1}}}}  \dmGLsym{:}  \dmGLnt{B_{{\mathrm{1}}}}  \otimes  \dmGLnt{B_{{\mathrm{2}}}} \\
      \delta_{{\mathrm{1}}} ,  \dmGLnt{r_{{\mathrm{1}}}}  ,  \delta'_{{\mathrm{1}}}   \odot  \Delta  \dmGLsym{,}  \dmGLmv{x}  \dmGLsym{:}  \dmGLnt{X}  \dmGLsym{,}  \Delta'  \dmGLsym{;}   \Gamma_{{\mathrm{1}}}  \dmGLsym{,}  \dmGLmv{y_{{\mathrm{1}}}}  \dmGLsym{:}  \dmGLnt{B_{{\mathrm{1}}}}  \dmGLsym{,}  \dmGLmv{y_{{\mathrm{2}}}}  \dmGLsym{:}  \dmGLnt{B_{{\mathrm{2}}}}  ,  \Gamma_{{\mathrm{3}}}   \vdash_\mathsf{M}  \dmGLnt{l_{{\mathrm{2}}}}  \dmGLsym{:}  \dmGLnt{C}
    }
    {
      \dmGLsym{(}    \delta_{{\mathrm{1}}} ,  \dmGLnt{r_{{\mathrm{1}}}}  ,  \delta'_{{\mathrm{1}}}   \dmGLsym{)}  +  \dmGLsym{(}    \delta_{{\mathrm{2}}} ,  \dmGLnt{r_{{\mathrm{2}}}}  ,  \delta'_{{\mathrm{2}}}   \dmGLsym{)} \at
      \Delta  \dmGLsym{,}  \dmGLmv{x}  \dmGLsym{:}  \dmGLnt{X}  \dmGLsym{,}  \Delta' ; \Gamma_{{\mathrm{1}}}  ,  \Gamma_{{\mathrm{2}}}   ,  \Gamma_{{\mathrm{3}}}
      \\
      \vdash_\mathsf{M} \sfoperator{let} \, \dmGLsym{(}  \dmGLmv{y_{{\mathrm{1}}}}  \dmGLsym{,}  \dmGLmv{y_{{\mathrm{2}}}}  \dmGLsym{)}  \dmGLsym{=}  \dmGLnt{l_{{\mathrm{1}}}} \, \sfoperator{in} \, \dmGLnt{l_{{\mathrm{2}}}}  : \dmGLnt{C}
    }
  \]
  and we need to prove
  \begin{align*}
    & \dmGLsym{(}  \delta_{{\mathrm{1}}}  +  \delta_{{\mathrm{2}}}  \dmGLsym{)}  +   \dmGLsym{(}  \dmGLnt{r_{{\mathrm{1}}}}  +  \dmGLnt{r_{{\mathrm{2}}}}  \dmGLsym{)}   \cdot   \delta_{{\mathrm{0}}}  ,  \delta'_{{\mathrm{1}}}   +  \delta'_{{\mathrm{2}}} \\
    & \quad \at \Delta  \dmGLsym{,}  \dmGLsym{[}  \dmGLnt{t_{{\mathrm{0}}}}  \dmGLsym{/}  \dmGLmv{x}  \dmGLsym{]}  \Delta' ; \dmGLsym{[}  \dmGLnt{t_{{\mathrm{0}}}}  \dmGLsym{/}  \dmGLmv{x}  \dmGLsym{]}  \dmGLsym{(}   \Gamma_{{\mathrm{1}}}  ,   \Gamma_{{\mathrm{2}}}  ,  \Gamma_{{\mathrm{3}}}    \dmGLsym{)} \\
    & \quad \proves_{\ms M} \dmGLsym{[}  \dmGLnt{t_{{\mathrm{0}}}}  \dmGLsym{/}  \dmGLmv{x}  \dmGLsym{]}  \dmGLsym{(}  \sfoperator{let} \, \dmGLsym{(}  \dmGLmv{y_{{\mathrm{1}}}}  \dmGLsym{,}  \dmGLmv{y_{{\mathrm{2}}}}  \dmGLsym{)}  \dmGLsym{=}  \dmGLnt{l_{{\mathrm{1}}}} \, \sfoperator{in} \, \dmGLnt{l_{{\mathrm{2}}}}  \dmGLsym{)} : \dmGLsym{[}  \dmGLnt{t_{{\mathrm{0}}}}  \dmGLsym{/}  \dmGLmv{x}  \dmGLsym{]}  \dmGLnt{C}
  \end{align*}
  By the inductive hypothesis we have the judgments
  \begin{align*}
    & \delta_{{\mathrm{2}}}  +   \dmGLnt{r_{{\mathrm{2}}}}   \cdot   \delta_{{\mathrm{0}}}  ,  \delta'_{{\mathrm{2}}} \at \Delta  \dmGLsym{,}  \dmGLsym{[}  \dmGLnt{t_{{\mathrm{0}}}}  \dmGLsym{/}  \dmGLmv{x}  \dmGLsym{]}  \Delta' ; \dmGLsym{[}  \dmGLnt{t_{{\mathrm{0}}}}  \dmGLsym{/}  \dmGLmv{x}  \dmGLsym{]}  \Gamma_{{\mathrm{2}}}\\
    & \qquad \vdash_\mathsf{M} \dmGLsym{[}  \dmGLnt{t_{{\mathrm{0}}}}  \dmGLsym{/}  \dmGLmv{x}  \dmGLsym{]}  \dmGLnt{l_{{\mathrm{1}}}}  : \dmGLsym{[}  \dmGLnt{t_{{\mathrm{0}}}}  \dmGLsym{/}  \dmGLmv{x}  \dmGLsym{]}  \dmGLnt{B_{{\mathrm{1}}}}  \otimes  \dmGLsym{[}  \dmGLnt{t_{{\mathrm{0}}}}  \dmGLsym{/}  \dmGLmv{x}  \dmGLsym{]}  \dmGLnt{B_{{\mathrm{2}}}}\\
    & \delta_{{\mathrm{1}}}  +   \dmGLnt{r_{{\mathrm{1}}}}   \cdot   \delta_{{\mathrm{0}}}  ,  \delta'_{{\mathrm{1}}} \at \Delta  \dmGLsym{,}  \dmGLsym{[}  \dmGLnt{t_{{\mathrm{0}}}}  \dmGLsym{/}  \dmGLmv{x}  \dmGLsym{]}  \Delta' \\
    & \qquad ; \dmGLsym{[}  \dmGLnt{t_{{\mathrm{0}}}}  \dmGLsym{/}  \dmGLmv{x}  \dmGLsym{]}  \dmGLsym{(}   \Gamma_{{\mathrm{1}}}  \dmGLsym{,}  \dmGLmv{y_{{\mathrm{1}}}}  \dmGLsym{:}  \dmGLnt{B_{{\mathrm{1}}}}  \dmGLsym{,}  \dmGLmv{y_{{\mathrm{2}}}}  \dmGLsym{:}  \dmGLnt{B_{{\mathrm{2}}}}  ,  \Gamma_{{\mathrm{3}}}   \dmGLsym{)} \\
    & \qquad \vdash_\mathsf{M} \dmGLsym{[}  \dmGLnt{t_{{\mathrm{0}}}}  \dmGLsym{/}  \dmGLmv{x}  \dmGLsym{]}  \dmGLnt{l_{{\mathrm{2}}}} : \dmGLsym{[}  \dmGLnt{t_{{\mathrm{0}}}}  \dmGLsym{/}  \dmGLmv{x}  \dmGLsym{]}  \dmGLnt{C}
  \end{align*}
  from which the desired judgment follows.

  For the mixed linear case, we have
  \[
	\inferrule
    {
      \delta_{{\mathrm{2}}}  \odot  \Delta  \dmGLsym{;}  \Gamma_{{\mathrm{2}}}  \vdash_\mathsf{M}  \dmGLnt{l_{{\mathrm{1}}}}  \dmGLsym{:}  \dmGLnt{B_{{\mathrm{1}}}}  \otimes  \dmGLnt{B_{{\mathrm{2}}}}\\
      \delta_{{\mathrm{1}}}  \odot  \Delta  \dmGLsym{;}   \Gamma_{{\mathrm{1}}}  \dmGLsym{,}  \dmGLmv{y_{{\mathrm{1}}}}  \dmGLsym{:}  \dmGLnt{B_{{\mathrm{1}}}}  \dmGLsym{,}  \dmGLmv{y_{{\mathrm{2}}}}  \dmGLsym{:}  \dmGLnt{B_{{\mathrm{2}}}}  ,  \Gamma_{{\mathrm{3}}}   \vdash_\mathsf{M}  \dmGLnt{l_{{\mathrm{2}}}}  \dmGLsym{:}  \dmGLnt{C}
    }
    {
      \delta_{{\mathrm{1}}}  +  \delta_{{\mathrm{2}}}  \odot  \Delta  \dmGLsym{;}    \Gamma_{{\mathrm{1}}}  ,  \Gamma_{{\mathrm{2}}}   ,  \Gamma_{{\mathrm{3}}}   \vdash_\mathsf{M}  \sfoperator{let} \, \dmGLsym{(}  \dmGLmv{y_{{\mathrm{1}}}}  \dmGLsym{,}  \dmGLmv{y_{{\mathrm{2}}}}  \dmGLsym{)}  \dmGLsym{=}  \dmGLnt{l_{{\mathrm{1}}}} \, \sfoperator{in} \, \dmGLnt{l_{{\mathrm{2}}}}  \dmGLsym{:}  \dmGLnt{C}
    }
  \]
  We assume $ x : A $ occurs somewhere in the context
  $ \Gamma_{{\mathrm{1}}}  ,  \Gamma_{{\mathrm{2}}}   ,  \Gamma_{{\mathrm{3}}} $.
  We have three cases, depending on whether $ x : A $ occurs in
  $ \Gamma_{{\mathrm{1}}} $, $ \Gamma_{{\mathrm{2}}} $ or $ \Gamma_{{\mathrm{3}}} $.
  In the first case, where $ x : A $ occurs in $ \Gamma_{{\mathrm{1}}} $,
  the second assumption of the rule above has the form
  \[
    \delta_{{\mathrm{1}}}  \odot  \Delta  \dmGLsym{;}    \Gamma'_{{\mathrm{1}}}  \dmGLsym{,}  \dmGLmv{x}  \dmGLsym{:}  \dmGLnt{A}  ,  \Gamma''_{{\mathrm{1}}}   \dmGLsym{,}  \dmGLmv{y_{{\mathrm{1}}}}  \dmGLsym{:}  \dmGLnt{B}  \dmGLsym{,}  \dmGLmv{y_{{\mathrm{2}}}}  \dmGLsym{:}  \dmGLnt{B}  ,  \Gamma_{{\mathrm{2}}}   \vdash_\mathsf{M}  \dmGLnt{l_{{\mathrm{2}}}}  \dmGLsym{:}  \dmGLnt{C}
  \]
  By induction we have
  \[
	\delta_{{\mathrm{1}}}  \odot  \Delta  \dmGLsym{;}    \Gamma'_{{\mathrm{1}}}  ,  \Gamma''_{{\mathrm{1}}}   \dmGLsym{,}  \dmGLmv{y_{{\mathrm{1}}}}  \dmGLsym{:}  \dmGLnt{B}  \dmGLsym{,}  \dmGLmv{y_{{\mathrm{2}}}}  \dmGLsym{:}  \dmGLnt{B}  ,  \Gamma_{{\mathrm{2}}}   \vdash_\mathsf{M}  \dmGLsym{[}  \dmGLnt{l_{{\mathrm{0}}}}  \dmGLsym{/}  \dmGLmv{x}  \dmGLsym{]}  \dmGLnt{l_{{\mathrm{2}}}}  \dmGLsym{:}  \dmGLnt{C}
  \]
  and applying the rule \drulename[M]{tensorElim} to this judgment
  together with the other assumption of the rule above
  produces the judgment we need to prove.
  The other two cases proceed analogously.

  The other cases for the mixed fragment proceed analogously.
  We highlight the cases involving the adjoints $ \mathcal{G} $ and $ \mathcal{F} $,
  as they are nonstandard.

  \proofitem{Case \dmGLdruleMXXradjElimName}
  The mixed linear case is trivial.
  For the mixed graded case:
  \[
    \inferrule
    {
      \delta ,  \dmGLnt{r}  ,  \delta'   \odot  \Delta  \dmGLsym{,}  \dmGLmv{x}  \dmGLsym{:}  \dmGLnt{X}  \dmGLsym{,}  \Delta'  \vdash_\mathsf{G}  \dmGLnt{t}  \dmGLsym{:}  \mathcal{G} \, \dmGLnt{A}
    }
    {
      \delta ,  \dmGLnt{r}  ,  \delta'   \odot  \Delta  \dmGLsym{,}  \dmGLmv{x}  \dmGLsym{:}  \dmGLnt{X}  \dmGLsym{,}  \Delta'  \dmGLsym{;}  \emptyset  \vdash_\mathsf{M}  \mathcal G^{-1} \, \dmGLnt{t}  \dmGLsym{:}  \dmGLnt{A}
    }
  \]
  As in the previous case, we have that the terms
  $ \dmGLsym{[}  \dmGLnt{t_{{\mathrm{0}}}}  \dmGLsym{/}  \dmGLmv{x}  \dmGLsym{]}  \mathcal G^{-1} \, \dmGLnt{t} $ and $ \mathcal G^{-1} \, \dmGLsym{[}  \dmGLnt{t_{{\mathrm{0}}}}  \dmGLsym{/}  \dmGLmv{x}  \dmGLsym{]}  \dmGLnt{t} $ are equal,
  and $ \dmGLsym{[}  \dmGLnt{t_{{\mathrm{0}}}}  \dmGLsym{/}  \dmGLmv{x}  \dmGLsym{]}  \dmGLsym{(}  \mathcal{G} \, \dmGLnt{A}  \dmGLsym{)} $ and $ \mathcal{G} \, \dmGLsym{[}  \dmGLnt{t_{{\mathrm{0}}}}  \dmGLsym{/}  \dmGLmv{x}  \dmGLsym{]}  \dmGLnt{A} $ are also equal.
  So, by applying the induction hypothesis to the the premise of the rule above
  and then applying \dmGLdruleMXXradjElimName we get
  \[
    \inferrule
    {
      \delta  +   \dmGLnt{r}   \cdot   \delta_{{\mathrm{0}}}   +  \delta'  \odot  \Delta  \dmGLsym{,}  \dmGLsym{[}  \dmGLnt{t_{{\mathrm{0}}}}  \dmGLsym{/}  \dmGLmv{x}  \dmGLsym{]}  \Delta'  \vdash_\mathsf{G}  \dmGLsym{[}  \dmGLnt{t_{{\mathrm{0}}}}  \dmGLsym{/}  \dmGLmv{x}  \dmGLsym{]}  \dmGLnt{t}  \dmGLsym{:}  \mathcal{G} \, \dmGLsym{[}  \dmGLnt{t_{{\mathrm{0}}}}  \dmGLsym{/}  \dmGLmv{x}  \dmGLsym{]}  \dmGLnt{A}
    }
    {
      \delta  +   \dmGLnt{r}   \cdot   \delta_{{\mathrm{0}}}   +  \delta'  \odot  \Delta  \dmGLsym{,}  \dmGLsym{[}  \dmGLnt{t_{{\mathrm{0}}}}  \dmGLsym{/}  \dmGLmv{x}  \dmGLsym{]}  \Delta'  \dmGLsym{;}  \emptyset  \vdash_\mathsf{M}  \mathcal G^{-1} \, \dmGLsym{[}  \dmGLnt{t_{{\mathrm{0}}}}  \dmGLsym{/}  \dmGLmv{x}  \dmGLsym{]}  \dmGLnt{t}  \dmGLsym{:}  \dmGLsym{[}  \dmGLnt{t_{{\mathrm{0}}}}  \dmGLsym{/}  \dmGLmv{x}  \dmGLsym{]}  \dmGLnt{A}
    }
  \]
  Observe that the terms
  $ \mathcal G^{-1} \, \dmGLsym{[}  \dmGLnt{t_{{\mathrm{0}}}}  \dmGLsym{/}  \dmGLmv{x}  \dmGLsym{]}  \dmGLnt{t} $
  and
  $ \dmGLsym{[}  \dmGLnt{t_{{\mathrm{0}}}}  \dmGLsym{/}  \dmGLmv{x}  \dmGLsym{]}  \mathcal G^{-1} \, \dmGLnt{t} $
  are equal.
  This produces the desired result.

  \proofitem{Case \dmGLdruleMXXladjIntroName}
  In the mixed linear case we have
  \[
    \inferrule
    {
      \delta_{{\mathrm{3}}} ,  \dmGLnt{r_{{\mathrm{3}}}}   \odot  \Delta  \dmGLsym{,}  \dmGLmv{x}  \dmGLsym{:}  \dmGLnt{X}  \vdash_\mathsf{G}  \dmGLnt{A}  \dmGLsym{:}   \mathsf{Linear} \\
      \delta_{{\mathrm{1}}}  \odot  \Delta  \vdash_\mathsf{G}  \dmGLnt{t}  \dmGLsym{:}  \dmGLnt{X} \\
      \delta_{{\mathrm{2}}}  \odot  \Delta  \dmGLsym{;}   \Gamma  \dmGLsym{,}  \dmGLmv{y}  \dmGLsym{:}  \dmGLnt{B}  ,  \Gamma'   \vdash_\mathsf{M}  \dmGLnt{l}  \dmGLsym{:}  \dmGLsym{[}  \dmGLnt{t}  \dmGLsym{/}  \dmGLmv{x}  \dmGLsym{]}  \dmGLnt{A}
    }
    {
      \dmGLnt{r}   \cdot   \delta_{{\mathrm{1}}}   +  \delta_{{\mathrm{2}}}  \odot  \Delta  \dmGLsym{;}   \Gamma  \dmGLsym{,}  \dmGLmv{y}  \dmGLsym{:}  \dmGLnt{B}  ,  \Gamma'   \vdash_\mathsf{M}  \mathcal{F} \, \dmGLsym{(}  \dmGLnt{t}  \dmGLsym{,}  \dmGLnt{l}  \dmGLsym{)}  \dmGLsym{:}   \mathcal{F}  ( \dmGLmv{x}  :^{ \dmGLnt{r} }  \dmGLnt{X} ). \dmGLnt{A}
    }
  \]
  Keeping the first two premises unchanged,
  and applying the induction hypothesis to the third premise,
  we can use the rule \dmGLdruleMXXladjIntroName as follows
  \[
    \inferrule
    {
      \delta_{{\mathrm{3}}} ,  \dmGLnt{r_{{\mathrm{3}}}}   \odot  \Delta  \dmGLsym{,}  \dmGLmv{x}  \dmGLsym{:}  \dmGLnt{X}  \vdash_\mathsf{G}  \dmGLnt{A}  \dmGLsym{:}   \mathsf{Linear} \\
      \delta_{{\mathrm{1}}}  \odot  \Delta  \vdash_\mathsf{G}  \dmGLnt{t}  \dmGLsym{:}  \dmGLnt{X} \\
      \delta_{{\mathrm{0}}}  +  \delta_{{\mathrm{2}}}  \odot  \Delta  \dmGLsym{;}    \Gamma  ,  \Gamma_{{\mathrm{0}}}   ,  \Gamma'   \vdash_\mathsf{M}  \dmGLsym{[}  \dmGLnt{l_{{\mathrm{0}}}}  \dmGLsym{/}  \dmGLmv{y}  \dmGLsym{]}  \dmGLnt{l}  \dmGLsym{:}  \dmGLsym{[}  \dmGLnt{t}  \dmGLsym{/}  \dmGLmv{x}  \dmGLsym{]}  \dmGLnt{A}
    }
    {
      \dmGLnt{r}   \cdot   \delta_{{\mathrm{1}}}   +  \delta_{{\mathrm{2}}}  \odot  \Delta  \dmGLsym{;}    \Gamma  ,  \Gamma_{{\mathrm{0}}}   ,  \Gamma'   \vdash_\mathsf{M}  \mathcal{F} \, \dmGLsym{(}  \dmGLnt{t}  \dmGLsym{,}  \dmGLsym{[}  \dmGLnt{l_{{\mathrm{0}}}}  \dmGLsym{/}  \dmGLmv{y}  \dmGLsym{]}  \dmGLnt{l}  \dmGLsym{)}  \dmGLsym{:}   \mathcal{F}  ( \dmGLmv{x}  :^{ \dmGLnt{r} }  \dmGLnt{X} ). \dmGLnt{A}
    }
  \]
  To conclude,
  we only need to observe that
  $ \mathcal{F} \, \dmGLsym{(}  \dmGLnt{t}  \dmGLsym{,}  \dmGLsym{[}  \dmGLnt{l_{{\mathrm{0}}}}  \dmGLsym{/}  \dmGLmv{y}  \dmGLsym{]}  \dmGLnt{l}  \dmGLsym{)} $ is equal to
  $ \dmGLsym{[}  \dmGLnt{l_{{\mathrm{0}}}}  \dmGLsym{/}  \dmGLmv{y}  \dmGLsym{]}  \mathcal{F} \, \dmGLsym{(}  \dmGLnt{t}  \dmGLsym{,}  \dmGLnt{l}  \dmGLsym{)} $
  which is clear since $ \dmGLmv{y} $ cannot occur freely in $ \dmGLnt{t} $.
  % \todom{
  % Though I feel like actually proving this might be an incredibly
  % annoying induction.
  % }
  
  The mixed graded case is more interesting:
  \[
    \inferrule
    {
      \delta_{{\mathrm{3}}} ,  \dmGLnt{r_{{\mathrm{3}}}}  ,  \delta'_{{\mathrm{3}}}  ,  \dmGLnt{r'_{{\mathrm{3}}}}   \odot  \Delta  \dmGLsym{,}  \dmGLmv{y}  \dmGLsym{:}  \dmGLnt{Y}  \dmGLsym{,}  \Delta'  \dmGLsym{,}  \dmGLmv{x}  \dmGLsym{:}  \dmGLnt{X}  \vdash_\mathsf{G}  \dmGLnt{A}  \dmGLsym{:}   \mathsf{Linear} \\
      \delta_{{\mathrm{1}}} ,  \dmGLnt{r_{{\mathrm{1}}}}  ,  \delta'_{{\mathrm{1}}}   \odot  \Delta  \dmGLsym{,}  \dmGLmv{y}  \dmGLsym{:}  \dmGLnt{Y}  \dmGLsym{,}  \Delta'  \vdash_\mathsf{G}  \dmGLnt{t}  \dmGLsym{:}  \dmGLnt{X} \\
      \delta_{{\mathrm{2}}} ,  \dmGLnt{r_{{\mathrm{2}}}}  ,  \delta'_{{\mathrm{2}}}   \odot  \Delta  \dmGLsym{,}  \dmGLmv{y}  \dmGLsym{:}  \dmGLnt{Y}  \dmGLsym{,}  \Delta'  \dmGLsym{;}  \Gamma  \vdash_\mathsf{M}  \dmGLnt{l}  \dmGLsym{:}  \dmGLsym{[}  \dmGLnt{t}  \dmGLsym{/}  \dmGLmv{x}  \dmGLsym{]}  \dmGLnt{A}
    }
    {
      \dmGLnt{r}   \cdot   \dmGLsym{(}    \delta_{{\mathrm{1}}} ,  \dmGLnt{r_{{\mathrm{1}}}}  ,  \delta'_{{\mathrm{1}}}   \dmGLsym{)}   +  \dmGLsym{(}    \delta_{{\mathrm{2}}} ,  \dmGLnt{r_{{\mathrm{2}}}}  ,  \delta'_{{\mathrm{2}}}   \dmGLsym{)} \odot \Delta  \dmGLsym{,}  \dmGLmv{y}  \dmGLsym{:}  \dmGLnt{Y}  \dmGLsym{,}  \Delta' ; \Gamma \\
      \vdash_\mathsf{M} \mathcal{F} \, \dmGLsym{(}  \dmGLnt{t}  \dmGLsym{,}  \dmGLnt{l}  \dmGLsym{)} : \mathcal{F}  ( \dmGLmv{x}  :^{ \dmGLnt{r} }  \dmGLnt{X} ). \dmGLnt{A}
    }
  \]
  applying the induction hypothesis to all premises,
  we get the following judgments:
  \begin{mathpar}
    \delta_{{\mathrm{3}}} ,  \dmGLnt{r_{{\mathrm{3}}}}  ,  \delta'_{{\mathrm{3}}}   \odot  \Delta  \dmGLsym{,}  \dmGLsym{[}  \dmGLnt{t_{{\mathrm{0}}}}  \dmGLsym{/}  \dmGLmv{y}  \dmGLsym{]}  \Delta'  \dmGLsym{,}  \dmGLmv{x}  \dmGLsym{:}  \dmGLsym{[}  \dmGLnt{t_{{\mathrm{0}}}}  \dmGLsym{/}  \dmGLmv{y}  \dmGLsym{]}  \dmGLnt{X}  \vdash_\mathsf{G}  \dmGLsym{[}  \dmGLnt{t_{{\mathrm{0}}}}  \dmGLsym{/}  \dmGLmv{y}  \dmGLsym{]}  \dmGLnt{A}  \dmGLsym{:}   \mathsf{Linear}
    \and
    \dmGLnt{r_{{\mathrm{1}}}}   \cdot   \delta_{{\mathrm{0}}}   +  \delta_{{\mathrm{1}}} ,  \delta'_{{\mathrm{1}}}   \odot  \Delta  \dmGLsym{,}  \dmGLsym{[}  \dmGLnt{t_{{\mathrm{0}}}}  \dmGLsym{/}  \dmGLmv{y}  \dmGLsym{]}  \Delta'  \vdash_\mathsf{G}  \dmGLsym{[}  \dmGLnt{t_{{\mathrm{0}}}}  \dmGLsym{/}  \dmGLmv{y}  \dmGLsym{]}  \dmGLnt{t}  \dmGLsym{:}  \dmGLsym{[}  \dmGLnt{t_{{\mathrm{0}}}}  \dmGLsym{/}  \dmGLmv{y}  \dmGLsym{]}  \dmGLnt{X}
    \and
    \dmGLnt{r_{{\mathrm{2}}}}   \cdot   \delta_{{\mathrm{0}}}   +  \delta_{{\mathrm{2}}} ,  \delta'_{{\mathrm{2}}}   \odot  \Delta  \dmGLsym{,}  \dmGLsym{[}  \dmGLnt{t_{{\mathrm{0}}}}  \dmGLsym{/}  \dmGLmv{y}  \dmGLsym{]}  \Delta'  \dmGLsym{;}  \dmGLsym{[}  \dmGLnt{t_{{\mathrm{0}}}}  \dmGLsym{/}  \dmGLmv{y}  \dmGLsym{]}  \Gamma  \vdash_\mathsf{M}  \dmGLsym{[}  \dmGLnt{t_{{\mathrm{0}}}}  \dmGLsym{/}  \dmGLmv{y}  \dmGLsym{]}  \dmGLnt{l}  \dmGLsym{:}  \dmGLsym{[}  \dmGLnt{t_{{\mathrm{0}}}}  \dmGLsym{/}  \dmGLmv{y}  \dmGLsym{]}  \dmGLsym{[}  \dmGLnt{t}  \dmGLsym{/}  \dmGLmv{x}  \dmGLsym{]}  \dmGLnt{A}
  \end{mathpar}
  We cannot immediately apply the rule \dmGLdruleMXXladjIntroName
  to these sequents, as the third one does not have the right form.
  We observe that
  $ \dmGLsym{[}  \dmGLnt{t_{{\mathrm{0}}}}  \dmGLsym{/}  \dmGLmv{y}  \dmGLsym{]}  \dmGLsym{[}  \dmGLnt{t}  \dmGLsym{/}  \dmGLmv{x}  \dmGLsym{]}  \dmGLnt{A} $ is equal to
  $ \dmGLsym{[}  \dmGLsym{[}  \dmGLnt{t_{{\mathrm{0}}}}  \dmGLsym{/}  \dmGLmv{y}  \dmGLsym{]}  \dmGLnt{t}  \dmGLsym{/}  \dmGLmv{x}  \dmGLsym{]}  \dmGLsym{[}  \dmGLnt{t_{{\mathrm{0}}}}  \dmGLsym{/}  \dmGLmv{y}  \dmGLsym{]}  \dmGLnt{A} $.
  Substituting this into the third sequent puts it into the right form
  to apply \dmGLdruleMXXladjIntroName and we obtain
  \[
    \inferrule
    {
      \delta_{{\mathrm{3}}} ,  \dmGLnt{r_{{\mathrm{3}}}}  ,  \delta'_{{\mathrm{3}}}   \odot  \Delta  \dmGLsym{,}  \dmGLsym{[}  \dmGLnt{t_{{\mathrm{0}}}}  \dmGLsym{/}  \dmGLmv{y}  \dmGLsym{]}  \Delta'  \dmGLsym{,}  \dmGLmv{x}  \dmGLsym{:}  \dmGLsym{[}  \dmGLnt{t_{{\mathrm{0}}}}  \dmGLsym{/}  \dmGLmv{y}  \dmGLsym{]}  \dmGLnt{X}  \vdash_\mathsf{G}  \dmGLsym{[}  \dmGLnt{t_{{\mathrm{0}}}}  \dmGLsym{/}  \dmGLmv{y}  \dmGLsym{]}  \dmGLnt{A}  \dmGLsym{:}   \mathsf{Linear}
      \\\\
      \dmGLnt{r_{{\mathrm{1}}}}   \cdot   \delta_{{\mathrm{0}}}   +  \delta_{{\mathrm{1}}} ,  \delta'_{{\mathrm{1}}}   \odot  \Delta  \dmGLsym{,}  \dmGLsym{[}  \dmGLnt{t_{{\mathrm{0}}}}  \dmGLsym{/}  \dmGLmv{y}  \dmGLsym{]}  \Delta'  \vdash_\mathsf{G}  \dmGLsym{[}  \dmGLnt{t_{{\mathrm{0}}}}  \dmGLsym{/}  \dmGLmv{y}  \dmGLsym{]}  \dmGLnt{t}  \dmGLsym{:}  \dmGLsym{[}  \dmGLnt{t_{{\mathrm{0}}}}  \dmGLsym{/}  \dmGLmv{y}  \dmGLsym{]}  \dmGLnt{X}
      \\\\
      \dmGLnt{r_{{\mathrm{2}}}}   \cdot   \delta_{{\mathrm{0}}}   +  \delta_{{\mathrm{2}}} ,  \delta'_{{\mathrm{2}}}
      \odot
      \Delta  \dmGLsym{,}  \dmGLsym{[}  \dmGLnt{t_{{\mathrm{0}}}}  \dmGLsym{/}  \dmGLmv{y}  \dmGLsym{]}  \Delta'
      ;
      \dmGLsym{[}  \dmGLnt{t_{{\mathrm{0}}}}  \dmGLsym{/}  \dmGLmv{y}  \dmGLsym{]}  \Gamma \\\\ \quad
      \vdash_\mathsf{M}
      \dmGLsym{[}  \dmGLnt{t_{{\mathrm{0}}}}  \dmGLsym{/}  \dmGLmv{y}  \dmGLsym{]}  \dmGLnt{l}
      :
      \dmGLsym{[}  \dmGLsym{[}  \dmGLnt{t_{{\mathrm{0}}}}  \dmGLsym{/}  \dmGLmv{y}  \dmGLsym{]}  \dmGLnt{t}  \dmGLsym{/}  \dmGLmv{x}  \dmGLsym{]}  \dmGLsym{[}  \dmGLnt{t_{{\mathrm{0}}}}  \dmGLsym{/}  \dmGLmv{y}  \dmGLsym{]}  \dmGLnt{A}
    }
    {
      \dmGLnt{r}   \cdot   \dmGLsym{(}   \dmGLnt{r_{{\mathrm{1}}}}   \cdot    \delta_{{\mathrm{0}}}  +  \delta_{{\mathrm{1}}} ,  \delta'_{{\mathrm{1}}}    \dmGLsym{)}   +  \dmGLsym{(}   \dmGLnt{r_{{\mathrm{2}}}}   \cdot    \delta_{{\mathrm{0}}}  +  \delta_{{\mathrm{2}}} ,  \delta'_{{\mathrm{2}}}    \dmGLsym{)} 
      \odot
      \Delta  \dmGLsym{,}  \dmGLsym{[}  \dmGLnt{t_{{\mathrm{0}}}}  \dmGLsym{/}  \dmGLmv{y}  \dmGLsym{]}  \Delta'
      ;
      \dmGLsym{[}  \dmGLnt{t_{{\mathrm{0}}}}  \dmGLsym{/}  \dmGLmv{y}  \dmGLsym{]}  \Gamma \\
      \vdash_\mathsf{M}
      \dmGLsym{[}  \dmGLnt{t_{{\mathrm{0}}}}  \dmGLsym{/}  \dmGLmv{y}  \dmGLsym{]}  \dmGLnt{l}
      :
      \dmGLsym{[}  \dmGLsym{[}  \dmGLnt{t_{{\mathrm{0}}}}  \dmGLsym{/}  \dmGLmv{y}  \dmGLsym{]}  \dmGLnt{t}  \dmGLsym{/}  \dmGLmv{x}  \dmGLsym{]}  \dmGLsym{[}  \dmGLnt{t_{{\mathrm{0}}}}  \dmGLsym{/}  \dmGLmv{y}  \dmGLsym{]}  \dmGLnt{A}
    }
  \]
  Applying the above equality in the conclusion,
  we obtain the desired result.

  \proofitem{Case \dmGLdruleMXXladjElimName}
  The mixed graded case is straightforward.
  In the mixed linear case, we assume that we have 
  \[
    \inferrule
    {
      \delta_{{\mathrm{3}}}  \odot  \Delta  \vdash_\mathsf{G}  \dmGLnt{C}  \dmGLsym{:}   \mathsf{Linear}
      \\
      \delta_{{\mathrm{1}}}  \odot  \Delta  \dmGLsym{;}  \Gamma_{{\mathrm{1}}}  \vdash_\mathsf{M}  \dmGLnt{l_{{\mathrm{1}}}}  \dmGLsym{:}   \mathcal{F}  ( \dmGLmv{x}  :^{ \dmGLnt{q} }  \dmGLnt{X} ). \dmGLnt{B}
      \\
      \delta_{{\mathrm{2}}} ,  \dmGLnt{q}   \odot  \Delta  \dmGLsym{,}  \dmGLmv{x}  \dmGLsym{:}  \dmGLnt{X}  \dmGLsym{;}  \Gamma_{{\mathrm{2}}}  \dmGLsym{,}  \dmGLmv{y}  \dmGLsym{:}  \dmGLnt{B}  \vdash_\mathsf{M}  \dmGLnt{l_{{\mathrm{2}}}}  \dmGLsym{:}  \dmGLnt{C}
    }
    {
      \delta_{{\mathrm{1}}}  +  \delta_{{\mathrm{2}}}  \odot  \Delta  \dmGLsym{;}   \Gamma_{{\mathrm{1}}}  ,  \Gamma_{{\mathrm{2}}}   \vdash_\mathsf{M}  \sfoperator{let} \, \mathcal{F} \, \dmGLsym{(}  \dmGLmv{x}  \dmGLsym{,}  \dmGLmv{y}  \dmGLsym{)}  \dmGLsym{=}  \dmGLnt{l_{{\mathrm{1}}}} \, \sfoperator{in} \, \dmGLnt{l_{{\mathrm{2}}}}  \dmGLsym{:}  \dmGLnt{C}
    }
  \]
  We assume that the type assignment $ \dmGLmv{z}  \dmGLsym{:}  \dmGLnt{A} $
  occurs in $ \Gamma_{{\mathrm{1}}}  ,  \Gamma_{{\mathrm{2}}} $, the linear context of the conclusion
  of this rule.
  There are now two subcases:

  \textit{Case}. $ \dmGLmv{z}  \dmGLsym{:}  \dmGLnt{A} $ occurs in $ \Gamma_{{\mathrm{1}}} $.
  Then $ \Gamma_{{\mathrm{1}}} $ is of the form
  $ \Gamma'_{{\mathrm{1}}}  \dmGLsym{,}  \dmGLmv{z}  \dmGLsym{:}  \dmGLnt{A}  ,  \Gamma''_{{\mathrm{1}}} $ and by the inductive hypothesis we obtain
  $ \delta_{{\mathrm{0}}}  +  \delta_{{\mathrm{1}}}  \odot  \Delta  \dmGLsym{;}    \Gamma'_{{\mathrm{1}}}  ,  \Gamma_{{\mathrm{0}}}   ,  \Gamma''_{{\mathrm{1}}}   \vdash_\mathsf{M}  \dmGLsym{[}  \dmGLnt{l_{{\mathrm{0}}}}  \dmGLsym{/}  \dmGLmv{z}  \dmGLsym{]}  \dmGLnt{l_{{\mathrm{1}}}}  \dmGLsym{:}   \mathcal{F}  ( \dmGLmv{x}  :^{ \dmGLnt{q} }  \dmGLnt{X} ). \dmGLnt{A} $.

  \textit{Case}. $ \dmGLmv{z}  \dmGLsym{:}  \dmGLnt{A} $ occurs in $ \Gamma_{{\mathrm{2}}} $.
  Then $ \Gamma_{{\mathrm{2}}} $ is of the form
  $ \Gamma'_{{\mathrm{2}}}  \dmGLsym{,}  \dmGLmv{z}  \dmGLsym{:}  \dmGLnt{A}  ,  \Gamma''_{{\mathrm{2}}} $ and by the inductive hypothesis we obtain
  $ \delta_{{\mathrm{0}}}  +  \delta_{{\mathrm{2}}}  \odot  \Delta  \dmGLsym{;}    \Gamma'_{{\mathrm{2}}}  ,  \Gamma_{{\mathrm{0}}}   ,  \Gamma''_{{\mathrm{2}}}   \vdash_\mathsf{M}  \dmGLsym{[}  \dmGLnt{l_{{\mathrm{0}}}}  \dmGLsym{/}  \dmGLmv{z}  \dmGLsym{]}  \dmGLnt{l_{{\mathrm{2}}}}  \dmGLsym{:}  \dmGLnt{C} $.

  In either case we conclude immediately,
  by applying the rule \dmGLdruleMXXladjElimName
  to the newly obtained judgment and the two judgments to which we
  did not apply the inductive hypothesis.
  % \todom{
  % The thing that makes this ``actually'' work,
  % is the fact that we somehow assume that $ \dmGLnt{l_{{\mathrm{1}}}} $ and $ \dmGLnt{l_{{\mathrm{2}}}} $
  % have disjoint sets of linear variables,
  % which happens because we rename any variables that occur in both
  % $ \Gamma_{{\mathrm{1}}} $ and $ \Gamma_{{\mathrm{2}}} $ so that they only occur in one of the two
  % when we build up the linear context $ \Gamma_{{\mathrm{1}}}  ,  \Gamma_{{\mathrm{2}}} $ of the conclusion.
  % I feel like this is a remark that should at least go somewhere?
  % }

  The cases concerned with judgments about well-fromedness of graded and mixed contexts
  are straightforward.
\end{proof}

\RadjLeftRule*

\begin{proof}
  We have a derivation
  \[
    \inferrule*[Right = \text{\drulename[G]{var}}]{ \vdots }
    {
      \inferrule* [Right = \text{\dmGLdruleMXXradjElimName{}}]{
        \vec{0} ,   1    \odot  \Delta  \dmGLsym{,}  \dmGLmv{y}  \dmGLsym{:}  \mathcal{G} \, \dmGLnt{A}  \vdash_\mathsf{G}  \dmGLmv{y}  \dmGLsym{:}  \mathcal{G} \, \dmGLnt{A}
      }
      {
        \vec{0} ,   1    \odot  \Delta  \dmGLsym{,}  \dmGLmv{y}  \dmGLsym{:}  \mathcal{G} \, \dmGLnt{A}  \dmGLsym{;}  \emptyset  \vdash_\mathsf{M}  \mathcal G^{-1} \, \dmGLmv{y}  \dmGLsym{:}  \dmGLnt{A}
      }
    }
  \]
  Weakening the assumed sequent by $ \mathcal{G} \, \dmGLnt{A} $ we get
  \[
    \delta ,  0   \odot  \Delta  \dmGLsym{,}  \dmGLmv{y}  \dmGLsym{:}  \mathcal{G} \, \dmGLnt{A}  \dmGLsym{;}  \Gamma  \dmGLsym{,}  \dmGLmv{x}  \dmGLsym{:}  \dmGLnt{A}  \vdash_\mathsf{M}  \dmGLnt{l}  \dmGLsym{:}  \dmGLnt{B}
  \]
  Now by substituion, we get
  \[
    \dmGLsym{(}   \vec{0} ,   1    \dmGLsym{)}  +  \dmGLsym{(}   \delta ,  0   \dmGLsym{)}  \odot  \Delta  \dmGLsym{,}  \dmGLmv{y}  \dmGLsym{:}  \mathcal{G} \, \dmGLnt{A}  \dmGLsym{;}  \Gamma  \vdash_\mathsf{M}  \dmGLsym{[}  \mathcal G^{-1} \, \dmGLmv{y}  \dmGLsym{/}  \dmGLmv{x}  \dmGLsym{]}  \dmGLnt{l}  \dmGLsym{:}  \dmGLnt{B}
  \]
  as desired.
\end{proof}

In the following we will write

\[
  \delta  \odot  \Delta  \vdash_\mathsf{G}  \dmGLnt{X}  \equiv  \dmGLnt{Y}  \dmGLsym{:}   \mathsf{Type}
\]
to mean that the following hold simultaneously:
\begin{enumerate}
\item
  $ \delta  \odot  \Delta  \vdash_\mathsf{G}  \dmGLnt{X}  \dmGLsym{:}   \mathsf{Type} $
\item
  $ \delta  \odot  \Delta  \vdash_\mathsf{G}  \dmGLnt{Y}  \dmGLsym{:}   \mathsf{Type} $
\item
  $ \dmGLnt{X}  \equiv  \dmGLnt{Y} $
\end{enumerate}

\begin{lem}[Lambda Inversion]
  \label{lem:lambdaInversion}
  If $ \delta  \odot  \Delta  \vdash_\mathsf{G}  \dmGLsym{(}  \lambda  \dmGLmv{x}  \dmGLsym{.}  \dmGLnt{t}  \dmGLsym{)}  \dmGLsym{:}  \dmGLnt{T} $,
  then there exist $ \dmGLnt{X} $, $ \dmGLnt{Y} $, $ \dmGLnt{r} $ and $ \delta_{{\mathrm{0}}} $
  such that
  $ \delta_{{\mathrm{0}}}  \odot  \Delta  \vdash_\mathsf{G}  \dmGLnt{T}  \equiv   (  \dmGLmv{x}  :^{ \dmGLnt{r} }  \dmGLnt{X}  )  \to   \dmGLnt{Y}   \dmGLsym{:}   \mathsf{Type} $ and
  $ \delta ,  \dmGLnt{r}   \odot  \Delta  \dmGLsym{,}  \dmGLmv{x}  \dmGLsym{:}  \dmGLnt{X}  \vdash_\mathsf{G}  \dmGLnt{t}  \dmGLsym{:}  \dmGLnt{Y} $.
\end{lem}

\begin{proof}
  By induction on the derivation of $ \delta  \odot  \Delta  \vdash_\mathsf{G}  \dmGLsym{(}  \lambda  \dmGLmv{x}  \dmGLsym{.}  \dmGLnt{t}  \dmGLsym{)}  \dmGLsym{:}  \dmGLnt{T} $.
  \proofitem{Case \drulename[G]{lambda}}
  There is nothing to do in this case.

  \item[ \emph{ Case \drulename[G]{subusage}}]
    Inverting the subusage rule,
    we get $ \delta'  \odot  \Delta  \vdash_\mathsf{G}  \dmGLsym{(}  \lambda  \dmGLmv{x}  \dmGLsym{.}  \dmGLnt{t}  \dmGLsym{)}  \dmGLsym{:}  \dmGLnt{T} $
    for some $ \delta' $ with $ \delta'  \leq  \delta $.
    By the inductive hypothesis we get $ \delta' ,  \dmGLnt{r}   \odot  \Delta  \dmGLsym{,}  \dmGLmv{x}  \dmGLsym{:}  \dmGLnt{X}  \vdash_\mathsf{G}  \dmGLnt{t}  \dmGLsym{:}  \dmGLnt{Y} $
    and by applying subusage we get
    $ \delta' ,  \dmGLnt{r}   \odot  \Delta  \dmGLsym{,}  \dmGLmv{x}  \dmGLsym{:}  \dmGLnt{X}  \vdash_\mathsf{G}  \dmGLnt{t}  \dmGLsym{:}  \dmGLnt{Y} $ as desired.
    $ \delta_{{\mathrm{0}}}  \odot  \Delta  \vdash_\mathsf{G}  \dmGLnt{T}  \equiv   (  \dmGLmv{x}  :^{ \dmGLnt{r} }  \dmGLnt{X}  )  \to   \dmGLnt{Y}   \dmGLsym{:}   \mathsf{Type} $ also follows immediately from the inductive
    hypothesis.

  \item[ \emph{ Case \drulename[G]{weak}}]
    We have that $ \Delta = \Delta'  \dmGLsym{,}  \dmGLmv{z}  \dmGLsym{:}  \dmGLnt{Z} $ and $ \delta = \delta' ,  0 $
    for some $ \Delta'$, $ \delta' $  and $ Z $.
    Furthermore we have $ \delta'  \odot  \Delta'  \vdash_\mathsf{G}  \dmGLsym{(}  \lambda  \dmGLmv{x}  \dmGLsym{.}  \dmGLnt{t}  \dmGLsym{)}  \dmGLsym{:}  \dmGLnt{T} $.
    By inductive hypothesis we have
    $ \delta_{{\mathrm{0}}}  \odot  \Delta'  \vdash_\mathsf{G}  \dmGLnt{T}  \equiv   (  \dmGLmv{x}  :^{ \dmGLnt{r} }  \dmGLnt{X}  )  \to   \dmGLnt{Y}   \dmGLsym{:}   \mathsf{Type} $
    and
    $ \delta' ,  \dmGLnt{r}   \odot  \Delta'  \dmGLsym{,}  \dmGLmv{x}  \dmGLsym{:}  \dmGLnt{X}  \vdash_\mathsf{G}  \dmGLnt{t}  \dmGLsym{:}  \dmGLnt{Y} $.
    Since $ Z $ is a well-formed type in context $ \Delta' $
    it is also well-formed in context $ \Delta'  \dmGLsym{,}  \dmGLmv{x}  \dmGLsym{:}  \dmGLnt{X} $
    and hence we may weaken by $ \dmGLnt{Z} $, obtaining
    $ \delta' ,  \dmGLnt{r}  ,  0   \odot  \Delta'  \dmGLsym{,}  \dmGLmv{x}  \dmGLsym{:}  \dmGLnt{X}  \dmGLsym{,}  \dmGLmv{z}  \dmGLsym{:}  \dmGLnt{Z}  \vdash_\mathsf{G}  \dmGLnt{t}  \dmGLsym{:}  \dmGLnt{Y} $.
    Similarly $ \dmGLnt{X} $ is well formed in context $ \Delta' $
    and hence by exchange we get
    $ \delta' ,  0  ,  \dmGLnt{r}   \odot  \Delta'  \dmGLsym{,}  \dmGLmv{z}  \dmGLsym{:}  \dmGLnt{Z}  \dmGLsym{,}  \dmGLmv{x}  \dmGLsym{:}  \dmGLnt{X}  \vdash_\mathsf{G}  \dmGLnt{t}  \dmGLsym{:}  \dmGLnt{Y} $ as desired.

  \item[ \emph{ Case \drulename[G]{convert}}]
    Immediate since the relation ``$ \equiv $'' is transitive. \qedhere
\end{proof}

\begin{lem}[Unit Inversion]
  \label{lem:unitInversion}
  If $ \delta  \odot  \Delta  \vdash_\mathsf{G}  \mathbf{j}  \dmGLsym{:}  \dmGLnt{T} $ then we have $ \vec{0}  \leq  \delta $
  and $ \delta_{{\mathrm{0}}}  \odot  \Delta  \vdash_\mathsf{G}  \dmGLnt{T}  \equiv  \mathbf{J}  \dmGLsym{:}   \mathsf{Type} $ for some $ \delta_{{\mathrm{0}}} $.
\end{lem}

\begin{proof}
   By induction on the derivation of $ \delta  \odot  \Delta  \vdash_\mathsf{G}  \mathbf{j}  \dmGLsym{:}  \dmGLnt{T} $.

 \item[ \emph{ Case \drulename[G]{unit}}]
   There is nothing to prove in this case.

 \item[ \emph{ Case \drulename[G]{subusage}}]
   Inverting the subusage rule,
   we get $ \delta'  \odot  \Delta  \vdash_\mathsf{G}  \mathbf{j}  \dmGLsym{:}  \dmGLnt{T} $ with $ \delta'  \leq  \delta $.
   By inductive hypothesis we have
   $ \delta_{{\mathrm{0}}}  \odot  \Delta  \vdash_\mathsf{G}  \dmGLnt{T}  \equiv  \mathbf{J}  \dmGLsym{:}   \mathsf{Type} $ and $ \vec{0}  \leq  \delta' $.
   Since $ \leq $ is transitive, we obtain $ \vec{0}  \leq  \delta $.

 \item [ \emph{ Case \drulename[G]{weak}}]
   We have $ \Delta = \Delta'  \dmGLsym{,}  \dmGLmv{x}  \dmGLsym{:}  \dmGLnt{X} $ and $ \delta = \delta' ,  0 $.
   By inductive hypothesis we have
   $ \vec{0}  \leq  \delta' $ and 
   $ \delta_{{\mathrm{0}}}  \odot  \Delta'  \vdash_\mathsf{G}  \dmGLnt{T}  \equiv  \mathbf{J}  \dmGLsym{:}   \mathsf{Type} $
   weakening the latter by $ \dmGLnt{X} $ yields the second claim.
   The first claim follows since $ \vec{0} ,  0   \leq   \delta' ,  0 $.

 \item[ \emph{ Case \drulename[G]{convert}}]
   By inductive hypothesis and the fact that the relation ``$ \equiv $'' is transitive.
\end{proof}

\begin{lem}[Pair Inversion]
  \label{lem:pairInversion}
  If $ \delta  \odot  \Delta  \vdash_\mathsf{G}  \dmGLsym{(}  \dmGLnt{t_{{\mathrm{1}}}}  \dmGLsym{,}  \dmGLnt{t_{{\mathrm{2}}}}  \dmGLsym{)}  \dmGLsym{:}  \dmGLnt{T} $, then
  \begin{enumerate}[label=\roman*)]
    \item
      $ \delta'_{{\mathrm{0}}}  \odot  \Delta  \vdash_\mathsf{G}  \dmGLnt{X}  \dmGLsym{:}   \mathsf{Type} $ and
      $ \delta''_{{\mathrm{0}}} ,  \dmGLnt{r''_{{\mathrm{0}}}}   \odot  \Delta  \dmGLsym{,}  \dmGLmv{x}  \dmGLsym{:}  \dmGLnt{X}  \vdash_\mathsf{G}  \dmGLnt{Y}  \dmGLsym{:}   \mathsf{Type} $

    \item 
      $ \delta_{{\mathrm{0}}}  \odot  \Delta  \vdash_\mathsf{G}  \dmGLnt{T}  \equiv   ( \dmGLmv{x}  :^{ \dmGLnt{r} }  \dmGLnt{X} )  \boxtimes   \dmGLnt{Y}   \dmGLsym{:}   \mathsf{Type} $

    \item 
      $ \delta_{{\mathrm{1}}}  \odot  \Delta  \vdash_\mathsf{G}  \dmGLnt{t_{{\mathrm{1}}}}  \dmGLsym{:}  \dmGLnt{X} $ 

    \item 
      $ \delta_{{\mathrm{2}}}  \odot  \Delta  \vdash_\mathsf{G}  \dmGLnt{t_{{\mathrm{2}}}}  \dmGLsym{:}  \dmGLsym{[}  \dmGLnt{t_{{\mathrm{1}}}}  \dmGLsym{/}  \dmGLmv{x}  \dmGLsym{]}  \dmGLnt{Y} $

    \item
      $ \dmGLnt{r}   \cdot   \delta_{{\mathrm{1}}}   +  \delta_{{\mathrm{2}}}  \leq  \delta $

  \end{enumerate}
  for some
  $ \dmGLnt{X} $, $ \dmGLnt{Y} $, $ \delta_{{\mathrm{0}}} $, $ \delta'_{{\mathrm{0}}} $, $ \delta''_{{\mathrm{0}}} $,
  $ \dmGLnt{r''_{{\mathrm{0}}}} $, $ \delta_{{\mathrm{1}}} $, $ \delta_{{\mathrm{2}}} $.
\end{lem}

\begin{proof}
  By induction on the derivation of $ \delta  \odot  \Delta  \vdash_\mathsf{G}  \dmGLsym{(}  \dmGLnt{t_{{\mathrm{1}}}}  \dmGLsym{,}  \dmGLnt{t_{{\mathrm{2}}}}  \dmGLsym{)}  \dmGLsym{:}  \dmGLnt{T} $.
  
  \proofitem{Case \drulename[G]{gradedPairIntro}}
  In this case, there is nothing to prove.

  \proofitem{Case \drulename[G]{weak}}
  This is similar to the proofs of the previous inversion lemmas:
  Everything follows by applying weakening to the respective points in the inductive hypothesis.
  
  \proofitem{Case \drulename[G]{subusage}}
  Again, similar to the previous inversion lemmas
  and using the fact that $ \leq $ is transitive in point v).
  Note that we cannot state point v) as $ \dmGLnt{r}   \cdot   \delta_{{\mathrm{1}}}   +  \delta_{{\mathrm{2}}} = \delta $.

  \proofitem{Case \drulename[G]{convert}}
  By inductive hypothesis the fact that ``$ \equiv $'' is transitive.
\end{proof}

\begin{lem}[Coproduct Inversion]
  \label{lem:coproductInversion}
  If $ \delta  \odot  \Delta  \vdash_\mathsf{G}  \sfoperator{inl} \, \dmGLnt{t}  \dmGLsym{:}  \dmGLnt{T} $,
  then $ \delta_{{\mathrm{0}}}  \odot  \Delta  \vdash_\mathsf{G}  \dmGLnt{T}  \equiv  \dmGLnt{X_{{\mathrm{1}}}}  \boxplus  \dmGLnt{X_{{\mathrm{2}}}}  \dmGLsym{:}   \mathsf{Type} $
  and $ \delta  \odot  \Delta  \vdash_\mathsf{G}  \dmGLnt{t}  \dmGLsym{:}  \dmGLnt{X_{{\mathrm{1}}}} $
  for appropriate $ \dmGLnt{X_{{\mathrm{1}}}} $, $ \dmGLnt{X_{{\mathrm{2}}}} $ and $ \delta_{{\mathrm{0}}} $.
  A symmetrical statement holds for $ \sfoperator{inr} \, \dmGLnt{t} $.
\end{lem}

\begin{proof}
  By induction on the derivation of $ \delta  \odot  \Delta  \vdash_\mathsf{G}  \sfoperator{inl} \, \dmGLnt{t}  \dmGLsym{:}  \dmGLnt{T} $.

  \proofitem{Case \drulename[G]{coproductInl}}
    There is nothing to prove in this case.

  \proofitem{Case \drulename[G]{weak}}
    In this case we have $ \Delta = \Delta'  \dmGLsym{,}  \dmGLmv{y}  \dmGLsym{:}  \dmGLnt{Y} $ and $ \delta = \delta' ,  0 $
    and $ \delta'  \odot  \Delta'  \vdash_\mathsf{G}  \sfoperator{inl} \, \dmGLnt{t}  \dmGLsym{:}  \dmGLnt{T} $.
    By the inductive hypothesis we have
    $ \delta'  \odot  \Delta'  \vdash_\mathsf{G}  \dmGLnt{t}  \dmGLsym{:}  \dmGLnt{X_{{\mathrm{1}}}} $ and $ \delta'_{{\mathrm{0}}}  \odot  \Delta'  \vdash_\mathsf{G}  \dmGLnt{T}  \equiv  \dmGLnt{X_{{\mathrm{1}}}}  \boxplus  \dmGLnt{X_{{\mathrm{2}}}}  \dmGLsym{:}   \mathsf{Type} $.
    Weakening these by $ \dmGLnt{Y} $ yields the desired result.

  \proofitem{Case \drulename[G]{subusage}}
    In this case we have $ \delta'  \odot  \Delta  \vdash_\mathsf{G}  \sfoperator{inl} \, \dmGLnt{t}  \dmGLsym{:}  \dmGLnt{T} $ for some $ \delta' $ with $ \delta'  \leq  \delta $.
    By the inductive hypothesis we get
    $ \delta_{{\mathrm{0}}}  \odot  \Delta  \vdash_\mathsf{G}  \dmGLnt{T}  \equiv  \dmGLnt{X_{{\mathrm{1}}}}  \boxplus  \dmGLnt{X_{{\mathrm{2}}}}  \dmGLsym{:}   \mathsf{Type} $ and $ \delta'  \odot  \Delta  \vdash_\mathsf{G}  \dmGLnt{t}  \dmGLsym{:}  \dmGLnt{X_{{\mathrm{1}}}} $.
    Applying subusage to the latter judgment yields
    $ \delta  \odot  \Delta  \vdash_\mathsf{G}  \dmGLnt{t}  \dmGLsym{:}  \dmGLnt{X_{{\mathrm{1}}}} $.

  \proofitem{Case \drulename[G]{convert}}
    Follows from the inductive hypothesis and transitivity of ``$ \equiv $''.
\end{proof}

\SubjectReduction*

\begin{proof}
  By induction on the derivation of $ \delta  \odot  \Delta  \vdash_\mathsf{G}  \dmGLnt{t}  \dmGLsym{:}  \dmGLnt{X} $.
  We will omit those rules that cannot produce a term which can be reduced.
  
  \proofitem{Case \drulename[G]{subusage}}
    \[
      \inferrule
      {
        \delta  \odot  \Delta  \vdash_\mathsf{G}  \dmGLnt{t}  \dmGLsym{:}  \dmGLnt{X} \\
        \delta  \leq  \delta'
      }
      {
        \delta'  \odot  \Delta  \vdash_\mathsf{G}  \dmGLnt{t}  \dmGLsym{:}  \dmGLnt{X}
      }
    \]
    By the inductive hypothesis we have $ \delta  \odot  \Delta  \vdash_\mathsf{G}  \dmGLnt{t'}  \dmGLsym{:}  \dmGLnt{X} $
    and by subusage we conclude $ \delta'  \odot  \Delta  \vdash_\mathsf{G}  \dmGLnt{t'}  \dmGLsym{:}  \dmGLnt{X} $.

  \proofitem{Case \drulename[G]{weak}}
    \[
      \inferrule
      {
        \delta  \odot  \Delta  \vdash_\mathsf{G}  \dmGLnt{t}  \dmGLsym{:}  \dmGLnt{X} \\
        \delta_{{\mathrm{0}}}  \odot  \Delta  \vdash_\mathsf{G}  \dmGLnt{Y}  \dmGLsym{:}   \mathsf{Type} \\
        \dmGLmv{x} \, \notin \, \sfoperator{dom} \, \Delta
      }
      {
        \delta ,  0   \odot  \Delta  \dmGLsym{,}  \dmGLmv{y}  \dmGLsym{:}  \dmGLnt{Y}  \vdash_\mathsf{G}  \dmGLnt{t}  \dmGLsym{:}  \dmGLnt{X}
      }
    \]

    By induction we have $ \delta  \odot  \Delta  \vdash_\mathsf{G}  \dmGLnt{t'}  \dmGLsym{:}  \dmGLnt{X} $.
    Weakening this judgment by $ Y $ produces the desired result.

  \proofitem{Case \drulename[G]{convert}}
    \[
      \inferrule
      {
        \delta  \odot  \Delta  \vdash_\mathsf{G}  \dmGLnt{t}  \dmGLsym{:}  \dmGLnt{X} \\
        \delta_{{\mathrm{0}}}  \odot  \Delta  \vdash_\mathsf{G}  \dmGLnt{Y}  \dmGLsym{:}   \mathsf{Type} \\
        \delta'_{{\mathrm{0}}}  \odot  \Delta  \vdash_\mathsf{G}  \dmGLnt{X}  \equiv  \dmGLnt{Y}  \dmGLsym{:}   \mathsf{Type}
      }
      {
        \delta  \odot  \Delta  \vdash_\mathsf{G}  \dmGLnt{t}  \dmGLsym{:}  \dmGLnt{Y}
      }
    \]
    We get $ \delta  \odot  \Delta  \vdash_\mathsf{G}  \dmGLnt{t'}  \dmGLsym{:}  \dmGLnt{X} $ by induction.
    Applying conversion yields $ \delta  \odot  \Delta  \vdash_\mathsf{G}  \dmGLnt{t'}  \dmGLsym{:}  \dmGLnt{Y} $, as desired.

  \proofitem{Case \drulename[G]{unitElim}}
  The term constructed by this rule can only be reduced if it is of the form
    \[
      \inferrule
      {
        \delta_{{\mathrm{0}}} ,  \dmGLnt{r_{{\mathrm{0}}}}   \odot  \Delta  \dmGLsym{,}  \dmGLmv{x}  \dmGLsym{:}  \mathbf{J}  \vdash_\mathsf{G}  \dmGLnt{X}  \dmGLsym{:}   \mathsf{Type} \\\\
        \delta'  \odot  \Delta  \vdash_\mathsf{G}  \mathbf{j}  \dmGLsym{:}  \mathbf{J} \\
        \delta  \odot  \Delta  \vdash_\mathsf{G}  \dmGLnt{t}  \dmGLsym{:}  \dmGLsym{[}  \mathbf{j}  \dmGLsym{/}  \dmGLmv{x}  \dmGLsym{]}  \dmGLnt{X} \\
      }
      {
        \delta  +  \delta'  \odot  \Delta  \vdash_\mathsf{G}  \sfoperator{let} \, \mathbf{j} \, \dmGLsym{=}  \mathbf{j} \, \sfoperator{in} \, \dmGLnt{t}  \dmGLsym{:}  \dmGLsym{[}  \mathbf{j}  \dmGLsym{/}  \dmGLmv{x}  \dmGLsym{]}  \dmGLnt{X}
      }
    \]
    There is only one reduction that can be applied here, namely
    $ \sfoperator{let} \, \mathbf{j} \, \dmGLsym{=}  \mathbf{j} \, \sfoperator{in} \, \dmGLnt{t}  \leadsto  \dmGLnt{t} $,
    so we want to prove
    $ \delta  +  \delta'  \odot  \Delta  \vdash_\mathsf{G}  \dmGLnt{t}  \dmGLsym{:}  \dmGLsym{[}  \mathbf{j}  \dmGLsym{/}  \dmGLmv{x}  \dmGLsym{]}  \dmGLnt{X} $.
    By \Cref{lem:unitInversion} we have $ \vec{0}  \leq  \delta' $.
    By assumption we have $ \delta  \odot  \Delta  \vdash_\mathsf{G}  \dmGLnt{t}  \dmGLsym{:}  \dmGLsym{[}  \mathbf{j}  \dmGLsym{/}  \dmGLmv{x}  \dmGLsym{]}  \dmGLnt{X} $.
    By subusaging we now obtain
    $ \delta  +  \delta'  \odot  \Delta  \vdash_\mathsf{G}  \dmGLnt{t}  \dmGLsym{:}  \dmGLsym{[}  \mathbf{j}  \dmGLsym{/}  \dmGLmv{x}  \dmGLsym{]}  \dmGLnt{X} $.

  \proofitem{Case \drulename[G]{gradedPairElim}}
  Again, this can only be reduced if the constructed term looks as follows:
    \[
      \inferrule
      {
        \delta_{{\mathrm{0}}} ,  \dmGLnt{r_{{\mathrm{0}}}}   \odot  \Delta  \dmGLsym{,}  \dmGLmv{z}  \dmGLsym{:}   ( \dmGLmv{x}  :^{ \dmGLnt{r} }  \dmGLnt{X} )  \boxtimes   \dmGLnt{Y}   \vdash_\mathsf{G}  \dmGLnt{W}  \dmGLsym{:}   \mathsf{Type}\\\\
        \delta_{{\mathrm{1}}}  \odot  \Delta  \vdash_\mathsf{G}  \dmGLsym{(}  \dmGLnt{t_{{\mathrm{1}}}}  \dmGLsym{,}  \dmGLnt{t_{{\mathrm{2}}}}  \dmGLsym{)}  \dmGLsym{:}   ( \dmGLmv{x}  :^{ \dmGLnt{r} }  \dmGLnt{X} )  \boxtimes   \dmGLnt{Y}\\\\
        \delta_{{\mathrm{2}}} ,  \dmGLnt{r}  \cdot  \dmGLnt{q}  ,  \dmGLnt{q}   \odot  \Delta  \dmGLsym{,}  \dmGLmv{x}  \dmGLsym{:}  \dmGLnt{X}  \dmGLsym{,}  \dmGLmv{y}  \dmGLsym{:}  \dmGLnt{Y}  \vdash_\mathsf{G}  \dmGLnt{t}  \dmGLsym{:}  \dmGLsym{[}  \dmGLsym{(}  \dmGLmv{x}  \dmGLsym{,}  \dmGLmv{y}  \dmGLsym{)}  \dmGLsym{/}  \dmGLmv{z}  \dmGLsym{]}  \dmGLnt{W}
      }
      {
        \dmGLnt{q}   \cdot   \delta_{{\mathrm{1}}}   +  \delta_{{\mathrm{2}}}  \odot  \Delta  \vdash_\mathsf{G}  \sfoperator{let} \, \dmGLsym{(}  \dmGLmv{x}  \dmGLsym{,}  \dmGLmv{y}  \dmGLsym{)}  \dmGLsym{=}  \dmGLsym{(}  \dmGLnt{t_{{\mathrm{1}}}}  \dmGLsym{,}  \dmGLnt{t_{{\mathrm{2}}}}  \dmGLsym{)} \, \sfoperator{in} \, \dmGLnt{t}  \dmGLsym{:}  \dmGLsym{[}  \dmGLsym{(}  \dmGLnt{t_{{\mathrm{1}}}}  \dmGLsym{,}  \dmGLnt{t_{{\mathrm{2}}}}  \dmGLsym{)}  \dmGLsym{/}  \dmGLmv{z}  \dmGLsym{]}  \dmGLnt{W}
      }
    \]
    We want to prove
    \[
      \dmGLnt{q}   \cdot   \delta_{{\mathrm{1}}}   +  \delta_{{\mathrm{2}}}  \odot  \Delta  \vdash_\mathsf{G}  \dmGLsym{[}  \dmGLnt{t_{{\mathrm{1}}}}  \dmGLsym{/}  \dmGLmv{x}  \dmGLsym{]}  \dmGLsym{[}  \dmGLnt{t_{{\mathrm{2}}}}  \dmGLsym{/}  \dmGLmv{y}  \dmGLsym{]}  \dmGLnt{t}  \dmGLsym{:}  \dmGLsym{[}  \dmGLsym{(}  \dmGLnt{t_{{\mathrm{1}}}}  \dmGLsym{,}  \dmGLnt{t_{{\mathrm{2}}}}  \dmGLsym{)}  \dmGLsym{/}  \dmGLmv{z}  \dmGLsym{]}  \dmGLnt{W}.
    \]
    Consider the second of the assumptions in the above rule.
    By \Cref{lem:pairInversion}
    we see that there are $ \delta'_{{\mathrm{1}}}, \delta''_{{\mathrm{1}}} $ such that
    $ \delta'_{{\mathrm{1}}}  \odot  \Delta  \vdash_\mathsf{G}  \dmGLnt{t_{{\mathrm{1}}}}  \dmGLsym{:}  \dmGLnt{X} $,
    $ \delta''_{{\mathrm{1}}}  \odot  \Delta  \vdash_\mathsf{G}  \dmGLnt{t_{{\mathrm{2}}}}  \dmGLsym{:}  \dmGLsym{[}  \dmGLnt{t_{{\mathrm{1}}}}  \dmGLsym{/}  \dmGLmv{x}  \dmGLsym{]}  \dmGLnt{Y} $
    and $ \dmGLnt{r}   \cdot   \delta'_{{\mathrm{1}}}   +  \delta''_{{\mathrm{1}}} = \delta_{{\mathrm{1}}} $.
    Weakening the latter of the two judgments produces
    $ \delta''_{{\mathrm{1}}} ,  0   \odot  \Delta  \dmGLsym{,}  \dmGLmv{x'}  \dmGLsym{:}  \dmGLnt{X}  \vdash_\mathsf{G}  \dmGLnt{t_{{\mathrm{2}}}}  \dmGLsym{:}  \dmGLsym{[}  \dmGLnt{t_{{\mathrm{1}}}}  \dmGLsym{/}  \dmGLmv{x}  \dmGLsym{]}  \dmGLnt{X} $.
    Now applying substitution to the third assumption twice yields
    \begin{align*}
      & \delta_{{\mathrm{2}}} ,  \dmGLnt{r}  \cdot  \dmGLnt{q}  ,  \dmGLnt{q}   \odot  \Delta  \dmGLsym{,}  \dmGLmv{x}  \dmGLsym{:}  \dmGLnt{X}  \dmGLsym{,}  \dmGLmv{y}  \dmGLsym{:}  \dmGLnt{Y}  \vdash_\mathsf{G}  \dmGLnt{t}  \dmGLsym{:}  \dmGLsym{[}  \dmGLsym{(}  \dmGLmv{x}  \dmGLsym{,}  \dmGLmv{y}  \dmGLsym{)}  \dmGLsym{/}  \dmGLmv{z}  \dmGLsym{]}  \dmGLnt{W} \\
      & \text{substitute $ \dmGLnt{t_{{\mathrm{2}}}} $ for $ \dmGLmv{y} $} \\
      & \delta_{{\mathrm{2}}}  +   \dmGLnt{q}   \cdot   \delta''_{{\mathrm{1}}}  ,  \dmGLnt{r}  \cdot  \dmGLnt{q}   \odot  \Delta  \dmGLsym{,}  \dmGLmv{x}  \dmGLsym{:}  \dmGLnt{X}  \vdash_\mathsf{G}  \dmGLsym{[}  \dmGLnt{t_{{\mathrm{2}}}}  \dmGLsym{/}  \dmGLmv{y}  \dmGLsym{]}  \dmGLnt{t}  \dmGLsym{:}  \dmGLsym{[}  \dmGLsym{(}  \dmGLmv{x}  \dmGLsym{,}  \dmGLnt{t_{{\mathrm{2}}}}  \dmGLsym{)}  \dmGLsym{/}  \dmGLmv{z}  \dmGLsym{]}  \dmGLnt{W} \\
      & \text{substitute $ \dmGLnt{t_{{\mathrm{1}}}} $ for $ \dmGLmv{x} $} \\
      & \delta_{{\mathrm{2}}}  +   \dmGLnt{q}   \cdot   \delta''_{{\mathrm{1}}}   +   \dmGLnt{r}  \cdot  \dmGLnt{q}   \cdot   \delta'_{{\mathrm{1}}}   \odot  \Delta  \vdash_\mathsf{G}  \dmGLsym{[}  \dmGLnt{t_{{\mathrm{1}}}}  \dmGLsym{/}  \dmGLmv{x}  \dmGLsym{]}  \dmGLsym{[}  \dmGLnt{t_{{\mathrm{2}}}}  \dmGLsym{/}  \dmGLmv{y}  \dmGLsym{]}  \dmGLnt{t}  \dmGLsym{:}  \dmGLsym{[}  \dmGLsym{(}  \dmGLnt{t_{{\mathrm{1}}}}  \dmGLsym{,}  \dmGLnt{t_{{\mathrm{2}}}}  \dmGLsym{)}  \dmGLsym{/}  \dmGLmv{z}  \dmGLsym{]}  \dmGLnt{W}
    \end{align*}
    So, we are done.

  \proofitem{Case \drulename[G]{coproductElim}}
    In this case there are two reductions possible,
    one where the constructed term looks thus
    \[
      \inferrule
      {
        1   \leq  \dmGLnt{q} \\\\
        \delta_{{\mathrm{0}}} ,  \dmGLnt{r_{{\mathrm{0}}}}   \odot  \Delta  \dmGLsym{,}  \dmGLmv{y}  \dmGLsym{:}  \dmGLnt{X_{{\mathrm{1}}}}  \boxplus  \dmGLnt{X_{{\mathrm{2}}}}  \vdash_\mathsf{G}  \dmGLnt{Y}  \dmGLsym{:}   \mathsf{Type} \\\\
        \delta_{{\mathrm{1}}}  \odot  \Delta  \vdash_\mathsf{G}  \sfoperator{inl} \, \dmGLnt{t}  \dmGLsym{:}  \dmGLnt{X_{{\mathrm{1}}}}  \boxplus  \dmGLnt{X_{{\mathrm{2}}}}\\\\
        \delta_{{\mathrm{2}}} ,  \dmGLnt{q}   \odot  \Delta  \vdash_\mathsf{G}  \dmGLnt{s_{{\mathrm{1}}}}  \dmGLsym{:}   (  \dmGLmv{x}  :^{ \dmGLnt{q} }  \dmGLnt{X_{{\mathrm{1}}}}  )  \to   \dmGLsym{[}  \sfoperator{inl} \, \dmGLmv{x}  \dmGLsym{/}  \dmGLmv{y}  \dmGLsym{]}  \dmGLnt{Y}\\\\
        \delta_{{\mathrm{2}}} ,  \dmGLnt{q}   \odot  \Delta  \vdash_\mathsf{G}  \dmGLnt{s_{{\mathrm{2}}}}  \dmGLsym{:}   (  \dmGLmv{x}  :^{ \dmGLnt{q} }  \dmGLnt{X_{{\mathrm{2}}}}  )  \to   \dmGLsym{[}  \sfoperator{inr} \, \dmGLmv{x}  \dmGLsym{/}  \dmGLmv{y}  \dmGLsym{]}  \dmGLnt{Y}
      }
      {
        \dmGLnt{q}   \cdot   \delta_{{\mathrm{1}}}   +  \delta_{{\mathrm{2}}}  \odot  \Delta  \vdash_\mathsf{G}   \sfoperator{case} _{ \dmGLnt{q} }  \sfoperator{inl} \, \dmGLnt{t}   \sfoperator{of}   \dmGLnt{s_{{\mathrm{1}}}}  ;  \dmGLnt{s_{{\mathrm{2}}}}   \dmGLsym{:}  \dmGLsym{[}  \sfoperator{inl} \, \dmGLnt{t}  \dmGLsym{/}  \dmGLmv{y}  \dmGLsym{]}  \dmGLnt{Y}
      }
    \]
    and a symmetrical version for $ \sfoperator{inr} $ instead of $ \sfoperator{inl} $.
    We will only treat the $ \sfoperator{inl} $ case.
    Want to show
    \[
      \dmGLnt{q}   \cdot   \delta_{{\mathrm{1}}}   +  \delta_{{\mathrm{2}}}  \odot  \Delta  \vdash_\mathsf{G}  \dmGLnt{s_{{\mathrm{1}}}} \, \dmGLnt{t}  \dmGLsym{:}  \dmGLsym{[}  \sfoperator{inl} \, \dmGLnt{t}  \dmGLsym{/}  \dmGLmv{y}  \dmGLsym{]}  \dmGLnt{Y}
    \]
    By \Cref{lem:coproductInversion} we get $ \delta_{{\mathrm{1}}}  \odot  \Delta  \vdash_\mathsf{G}  \dmGLnt{t}  \dmGLsym{:}  \dmGLnt{X_{{\mathrm{1}}}} $.
    Then, by \drulename[G]{app} we get
    $ \dmGLnt{q}   \cdot   \delta_{{\mathrm{1}}}   +  \delta_{{\mathrm{2}}}  \odot  \Delta  \vdash_\mathsf{G}  \dmGLnt{s_{{\mathrm{1}}}} \, \dmGLnt{t}  \dmGLsym{:}  \dmGLsym{[}  \dmGLnt{t}  \dmGLsym{/}  \dmGLmv{x}  \dmGLsym{]}  \dmGLsym{[}  \sfoperator{inl} \, \dmGLmv{x}  \dmGLsym{/}  \dmGLmv{y}  \dmGLsym{]}  \dmGLnt{Y} $,
    and since $ \dmGLsym{[}  \dmGLnt{t}  \dmGLsym{/}  \dmGLmv{x}  \dmGLsym{]}  \dmGLsym{[}  \sfoperator{inl} \, \dmGLmv{x}  \dmGLsym{/}  \dmGLmv{y}  \dmGLsym{]}  \dmGLnt{Y} $ is equal to $ \dmGLsym{[}  \sfoperator{inl} \, \dmGLnt{t}  \dmGLsym{/}  \dmGLmv{y}  \dmGLsym{]}  \dmGLnt{Y} $ we are done.

    \proofitem{Case \drulename[G]{app}}
    In this case there are two possible reductions.
    The first is by the reduction rule \drulename[gRED]{lambda}.
    In this case the constructed term must look as follows:
    \[
      \inferrule
      {
        \delta_{{\mathrm{0}}} ,  \dmGLnt{r_{{\mathrm{0}}}}   \odot  \Delta  \dmGLsym{,}  \dmGLmv{x}  \dmGLsym{:}  \dmGLnt{X}  \vdash_\mathsf{G}  \dmGLnt{Y}  \dmGLsym{:}   \mathsf{Type} \\\\
        \delta_{{\mathrm{1}}}  \odot  \Delta  \vdash_\mathsf{G}  \dmGLsym{(}  \lambda  \dmGLmv{x}  \dmGLsym{.}  \dmGLnt{t'}  \dmGLsym{)}  \dmGLsym{:}   (  \dmGLmv{x}  :^{ \dmGLnt{r} }  \dmGLnt{X}  )  \to   \dmGLnt{Y} \\\\
        \delta_{{\mathrm{2}}}  \odot  \Delta  \vdash_\mathsf{G}  \dmGLnt{t}  \dmGLsym{:}  \dmGLnt{X}
      }
      {
        \delta_{{\mathrm{1}}}  +   \dmGLnt{r}   \cdot   \delta_{{\mathrm{2}}}   \odot  \Delta  \vdash_\mathsf{G}  \dmGLsym{(}  \lambda  \dmGLmv{x}  \dmGLsym{.}  \dmGLnt{t'}  \dmGLsym{)} \, \dmGLnt{t}  \dmGLsym{:}  \dmGLsym{[}  \dmGLnt{t}  \dmGLsym{/}  \dmGLmv{x}  \dmGLsym{]}  \dmGLnt{Y}
      }
    \]
    By \Cref{lem:lambdaInversion}
    we get $ \delta_{{\mathrm{2}}} ,  \dmGLnt{r}   \odot  \Delta  \dmGLsym{,}  \dmGLmv{x}  \dmGLsym{:}  \dmGLnt{X}  \vdash_\mathsf{G}  \dmGLnt{t'}  \dmGLsym{:}  \dmGLnt{Y} $.
    By substitution we obtain
    $ \delta_{{\mathrm{1}}}  +   \dmGLnt{r}   \cdot   \delta_{{\mathrm{2}}}   \odot  \Delta  \vdash_\mathsf{G}  \dmGLsym{[}  \dmGLnt{t}  \dmGLsym{/}  \dmGLmv{x}  \dmGLsym{]}  \dmGLnt{t'}  \dmGLsym{:}  \dmGLsym{[}  \dmGLnt{t}  \dmGLsym{/}  \dmGLmv{x}  \dmGLsym{]}  \dmGLnt{Y} $ as desired.

    The second possible reduction is by the rule \drulename[gRED]{appL}.
    The application of the rule \drulename[G]{app} is thus:
    \[
      \inferrule
      {
        \delta_{{\mathrm{0}}} ,  \dmGLnt{r_{{\mathrm{0}}}}   \odot  \Delta  \dmGLsym{,}  \dmGLmv{x}  \dmGLsym{:}  \dmGLnt{X}  \vdash_\mathsf{G}  \dmGLnt{Y}  \dmGLsym{:}   \mathsf{Type} \\\\
        \delta_{{\mathrm{1}}}  \odot  \Delta  \vdash_\mathsf{G}  \dmGLnt{t_{{\mathrm{1}}}}  \dmGLsym{:}   (  \dmGLmv{x}  :^{ \dmGLnt{r} }  \dmGLnt{X}  )  \to   \dmGLnt{Y} \\\\
        \delta_{{\mathrm{2}}}  \odot  \Delta  \vdash_\mathsf{G}  \dmGLnt{t}  \dmGLsym{:}  \dmGLnt{X}
      }
      {
        \delta_{{\mathrm{1}}}  +   \dmGLnt{r}   \cdot   \delta_{{\mathrm{2}}}   \odot  \Delta  \vdash_\mathsf{G}  \dmGLnt{t_{{\mathrm{1}}}} \, \dmGLnt{t}  \dmGLsym{:}  \dmGLsym{[}  \dmGLnt{t}  \dmGLsym{/}  \dmGLmv{x}  \dmGLsym{]}  \dmGLnt{Y}
      }
    \]
    We have by assumption that $ \dmGLnt{t_{{\mathrm{1}}}}  \leadsto  \dmGLnt{t_{{\mathrm{2}}}} $ for some $ \dmGLnt{t_{{\mathrm{2}}}} $.
    By the inductive hypothesis we know that
    $ \delta_{{\mathrm{1}}}  \odot  \Delta  \vdash_\mathsf{G}  \dmGLnt{t_{{\mathrm{2}}}}  \dmGLsym{:}   (  \dmGLmv{x}  :^{ \dmGLnt{r} }  \dmGLnt{X}  )  \to   \dmGLnt{Y} $ and by applying
    \drulename[G]{app}, we get $ \delta_{{\mathrm{1}}}  +   \dmGLnt{r}   \cdot   \delta_{{\mathrm{2}}}   \odot  \Delta  \vdash_\mathsf{G}  \dmGLnt{t_{{\mathrm{2}}}} \, \dmGLnt{t}  \dmGLsym{:}  \dmGLsym{[}  \dmGLnt{t}  \dmGLsym{/}  \dmGLmv{x}  \dmGLsym{]}  \dmGLnt{Y} $ as desired.

    The cases for the mixed reduction are similar
    and require similar inversion lemmas as the ones stated above.
    We omit these cases.
\end{proof}

\section{Metatheory of \glad{}}
\label{app:dal-meta}

\dalsubst*

\begin{proof}
  The proof is by induction on the derivation of the first sequent in each point.

  \proofitem{Case \dalrulename[glad]{var}}
  We need to match the judgment
  \[
    \vec{0}  \dalsym{,}   1   \mid  \mathcal
{M}  \dalsym{,}  \mathfrak{m}   \odot   \Gamma  \dalsym{,}  \dalmv{y}  \colon  \dalnt{B}   \vdash _{ \mathfrak{m} }  \dalmv{y}  \colon  \dalnt{B}
  \]
  with the pattern
  \[
	\delta  \dalsym{,}  \dalnt{r}  \dalsym{,}  \delta'  \mid  \mathcal
{M}  \dalsym{,}  \mathfrak{m}_{{\mathrm{0}}}  \dalsym{,}  \mathcal
{M}'   \odot   \Gamma  \dalsym{,}  \dalmv{x}  \colon  \dalnt{A}  \dalsym{,}  \Gamma'   \vdash _{ \mathfrak{m} }  \dalnt{b}  \colon  \dalnt{B}
  \]
  There are two cases.
  In the first, we have $ \delta' = \emptyset $,
  and therefore $ A = B $.
  i.e. the judgment we are performing induction on was derived as
  \[
	\inferrule
    {
      \dalmv{x}  \notin \operatorname{\mathsf{dom} }  \Gamma\\
      \delta  \mid  \mathcal
{M}   \odot   \Gamma   \vdash _{ \mathfrak{m}_{{\mathrm{0}}} }  \dalnt{A}  \colon  \mathsf{Type}
    }
    {
      \vec{0}  \dalsym{,}   1   \mid  \mathcal
{M}  \dalsym{,}  \mathfrak{m}_{{\mathrm{0}}}   \odot   \Gamma  \dalsym{,}  \dalmv{x}  \colon  \dalnt{A}   \vdash _{ \mathfrak{m}_{{\mathrm{0}}} }  \dalmv{x}  \colon  \dalnt{A}
    }
  \]
  and we need to prove
  \[
	\vec{0}  \dalsym{+}    1    \cdot   \delta_{{\mathrm{0}}}   \mid  \mathcal
{M}   \odot   \Gamma   \vdash _{ \mathfrak{m}_{{\mathrm{0}}} }  \dalsym{[}  \dalnt{a_{{\mathrm{0}}}}  \dalsym{/}  \dalmv{x}  \dalsym{]}  \dalmv{x}  \colon  \dalsym{[}  \dalnt{a_{{\mathrm{0}}}}  \dalsym{/}  \dalmv{x}  \dalsym{]}  \dalnt{A}
  \]
  But this simplifies to the assumption
  $ \delta_{{\mathrm{0}}}  \mid  \mathcal
{M}   \odot   \Gamma   \vdash _{ \mathfrak{m}_{{\mathrm{0}}} }  \dalnt{a_{{\mathrm{0}}}}  \colon  \dalnt{A} $
  once we observe that the variable $ x $ is not free in $ A $,
  and hence $ \dalsym{[}  \dalnt{a_{{\mathrm{0}}}}  \dalsym{/}  \dalmv{x}  \dalsym{]}  \dalnt{A} = A $.

  In the second case we have $ \delta' \neq \emptyset $.
  Then our derivation has the form
  \[
	\inferrule
    {
      \dalmv{y}  \notin \operatorname{\mathsf{dom} }  \dalsym{(}  \Gamma  \dalsym{,}  \dalmv{x}  \colon  \dalnt{A}  \dalsym{,}  \Gamma'  \dalsym{)} \\
      \delta  \dalsym{,}  \dalnt{r}  \dalsym{,}  \delta'  \mid  \mathcal
{M}  \dalsym{,}  \mathfrak{m}  \dalsym{,}  \mathcal
{M}'   \odot   \Gamma  \dalsym{,}  \dalmv{x}  \colon  \dalnt{A}  \dalsym{,}  \Gamma'   \vdash _{ \mathfrak{n} }  \dalnt{B}  \colon  \mathsf{Type}
    }
    {
      \vec{0}  \dalsym{,}   0   \dalsym{,}  \vec{0}  \dalsym{,}   1   \mid  \mathcal
{M}  \dalsym{,}  \mathfrak{m}  \dalsym{,}  \mathcal
{M}'  \dalsym{,}  \mathfrak{n}   \odot   \Gamma  \dalsym{,}  \dalmv{x}  \colon  \dalnt{A}  \dalsym{,}  \Gamma'  \dalsym{,}  \dalmv{y}  \colon  \dalnt{B}   \vdash _{ \mathfrak{n} }  \dalmv{y}  \colon  \dalnt{B}
    }
  \]
  and we need to prove that
  \[
	\vec{0}  \dalsym{,}  \vec{0}  \dalsym{,}   1   \mid  \mathcal
{M}  \dalsym{,}  \mathcal
{M}'  \dalsym{,}  \mathfrak{n}   \odot   \Gamma  \dalsym{,}  \dalsym{[}  \dalnt{a_{{\mathrm{0}}}}  \dalsym{/}  \dalmv{x}  \dalsym{]}  \Gamma'  \dalsym{,}  \dalmv{y}  \colon  \dalsym{[}  \dalnt{a_{{\mathrm{0}}}}  \dalsym{/}  \dalmv{x}  \dalsym{]}  \dalnt{B}   \vdash _{ \mathfrak{n} }  \dalmv{y}  \colon  \dalsym{[}  \dalnt{a_{{\mathrm{0}}}}  \dalsym{/}  \dalmv{x}  \dalsym{]}  \dalnt{B}
  \]
  which follows from the inductive hypothesis
  \[
	\delta  \dalsym{+}   \dalnt{r}   \cdot   \delta_{{\mathrm{0}}}   \dalsym{,}  \delta'  \mid  \mathcal
{M}  \dalsym{,}  \mathcal
{M}'   \odot   \Gamma  \dalsym{,}  \dalsym{[}  \dalnt{a_{{\mathrm{0}}}}  \dalsym{/}  \dalmv{x}  \dalsym{]}  \Gamma'   \vdash _{ \mathfrak{n} }  \dalsym{[}  \dalnt{a_{{\mathrm{0}}}}  \dalsym{/}  \dalmv{x}  \dalsym{]}  \dalnt{B}  \colon  \mathsf{Type}
  \]
  and an application of the rule \dalrulename[glad]{var}.

  \proofitem{Case \dalrulename[glad]{weak}}
  We need to match the judgment 
  \[
	\delta  \dalsym{,}   0   \mid  \mathcal
{M}  \dalsym{,}  \mathfrak{l}   \odot   \Gamma  \dalsym{,}  \dalmv{y}  \colon  \dalnt{C}   \vdash _{ \mathfrak{m} }  \dalnt{b}  \colon  \dalnt{B}
  \]
  with the pattern
  \[
	\delta  \dalsym{,}  \dalnt{r}  \dalsym{,}  \delta'  \mid  \mathcal
{M}  \dalsym{,}  \mathfrak{m}_{{\mathrm{0}}}  \dalsym{,}  \mathcal
{M}   \odot   \Gamma  \dalsym{,}  \dalmv{x}  \colon  \dalnt{A}  \dalsym{,}  \Gamma'   \vdash _{ \mathfrak{m} }  \dalnt{b}  \colon  \dalnt{B}
  \]
  Two cases are possible.
  In the first, we have $ \delta' = \emptyset $ and therefore
  and therefore $ A = C $
  our derivation has the form
  \[
	\inferrule
    {
      \mathsf{Weak} (  \mathfrak{l}  ) \\
      \mathfrak{m}  \leq  \mathfrak{m}_{{\mathrm{0}}} \\
      \dalmv{x}  \notin \operatorname{\mathsf{dom} }  \Gamma \\
      \delta  \mid  \mathcal
{M}   \odot   \Gamma   \vdash _{ \mathfrak{m} }  \dalnt{b}  \colon  \dalnt{B} \\
      \delta_{{\mathrm{1}}}  \mid  \mathcal
{M}   \odot   \Gamma   \vdash _{ \mathfrak{m}_{{\mathrm{0}}} }  \dalnt{A}  \colon  \mathsf{Type}
    }
    {
      \delta  \dalsym{,}   0   \mid  \mathcal
{M}  \dalsym{,}  \mathfrak{m}_{{\mathrm{0}}}   \odot   \Gamma  \dalsym{,}  \dalmv{x}  \colon  \dalnt{A}   \vdash _{ \mathfrak{m} }  \dalnt{b}  \colon  \dalnt{B}
    }
  \]
  and we need to prove
  \[
	\delta  \dalsym{+}    0    \cdot   \delta_{{\mathrm{0}}}   \mid  \mathcal
{M}   \odot   \Gamma   \vdash _{ \mathfrak{m} }  \dalsym{[}  \dalnt{a_{{\mathrm{0}}}}  \dalsym{/}  \dalmv{x}  \dalsym{]}  \dalnt{b}  \colon  \dalsym{[}  \dalnt{a_{{\mathrm{0}}}}  \dalsym{/}  \dalmv{x}  \dalsym{]}  \dalnt{B}
  \]
  which is an assumption since $ x $ is not a free variable in $ b $ or $ B $.

  In the second case we have $ \delta' \neq \emptyset $.
  In this case, our derivation has the form
  \[
	\inferrule
    {
      \mathsf{Weak} (  \mathfrak{l}  ) \\
      \mathfrak{m}  \leq  \mathfrak{l} \\
      \dalmv{y}  \notin \operatorname{\mathsf{dom} }  \dalsym{(}  \Gamma  \dalsym{,}  \dalmv{x}  \colon  \dalnt{A}  \dalsym{,}  \Gamma'  \dalsym{)} \\
      \delta  \dalsym{,}  \dalnt{r}  \dalsym{,}  \delta'  \mid  \mathcal
{M}  \dalsym{,}  \mathfrak{m}_{{\mathrm{0}}}  \dalsym{,}  \mathcal
{M}'   \odot   \Gamma  \dalsym{,}  \dalmv{x}  \colon  \dalnt{A}  \dalsym{,}  \Gamma'   \vdash _{ \mathfrak{m} }  \dalnt{b}  \colon  \dalnt{B} \\
      \delta_{{\mathrm{1}}}  \dalsym{,}  \dalnt{r_{{\mathrm{1}}}}  \dalsym{,}  \delta'_{{\mathrm{1}}}  \mid  \mathcal
{M}  \dalsym{,}  \mathfrak{m}_{{\mathrm{0}}}  \dalsym{,}  \mathcal
{M}   \odot   \Gamma  \dalsym{,}  \dalmv{x}  \colon  \dalnt{A}  \dalsym{,}  \Gamma'   \vdash _{ \mathfrak{l} }  \dalnt{C}  \colon  \mathsf{Type} \\
    }
    {
      \delta  \dalsym{,}  \dalnt{r}  \dalsym{,}  \delta'  \dalsym{,}   0   \mid  \mathcal
{M}  \dalsym{,}  \mathfrak{m}_{{\mathrm{0}}}  \dalsym{,}  \mathcal
{M}'  \dalsym{,}  \mathfrak{l}   \odot   \Gamma  \dalsym{,}  \dalmv{x}  \colon  \dalnt{A}  \dalsym{,}  \Gamma'  \dalsym{,}  \dalmv{y}  \colon  \dalnt{C}   \vdash _{ \mathfrak{m} }  \dalnt{b}  \colon  \dalnt{B}
    }
  \]
  and we need to prove
  \begin{align*}
	& \delta  \dalsym{+}   \dalnt{r}   \cdot   \delta_{{\mathrm{0}}}   \dalsym{,}  \delta'  \dalsym{,}   0
    \mid \mathcal
{M}  \dalsym{,}  \mathcal
{M}'  \dalsym{,}  \mathfrak{l} \\
    & \at \Gamma  \dalsym{,}  \dalsym{[}  \dalnt{a_{{\mathrm{0}}}}  \dalsym{/}  \dalmv{x}  \dalsym{]}  \Gamma'  \dalsym{,}  \dalmv{y}  \colon  \dalsym{[}  \dalnt{a_{{\mathrm{0}}}}  \dalsym{/}  \dalmv{x}  \dalsym{]}  \dalnt{C}
    \proves_{\mathfrak{m}} \dalsym{[}  \dalnt{a_{{\mathrm{0}}}}  \dalsym{/}  \dalmv{x}  \dalsym{]}  \dalnt{b} : \dalsym{[}  \dalnt{a_{{\mathrm{0}}}}  \dalsym{/}  \dalmv{x}  \dalsym{]}  \dalnt{B}
  \end{align*}
  This follows from the inductive hypotheses
  \begin{align*}
	& \delta  \dalsym{+}   \dalnt{r}   \cdot   \delta_{{\mathrm{0}}}   \dalsym{,}  \delta'
    \mid \mathcal
{M}  \dalsym{,}  \mathcal
{M}' \\
    & \at \Gamma  \dalsym{,}  \dalsym{[}  \dalnt{a_{{\mathrm{0}}}}  \dalsym{/}  \dalmv{x}  \dalsym{]}  \Gamma'
      \proves_{\mathfrak{m}} \dalsym{[}  \dalnt{a_{{\mathrm{0}}}}  \dalsym{/}  \dalmv{x}  \dalsym{]}  \dalnt{b} : \dalsym{[}  \dalnt{a_{{\mathrm{0}}}}  \dalsym{/}  \dalmv{x}  \dalsym{]}  \dalnt{B} \\
    & \text{and} \\
	& \delta_{{\mathrm{1}}}  \dalsym{+}   \dalnt{r}   \cdot   \delta_{{\mathrm{0}}}   \dalsym{,}  \delta'_{{\mathrm{1}}}
    \mid \mathcal
{M}  \dalsym{,}  \mathcal
{M}' \\
    & \at \Gamma  \dalsym{,}  \dalsym{[}  \dalnt{a_{{\mathrm{0}}}}  \dalsym{/}  \dalmv{x}  \dalsym{]}  \Gamma'
      \proves_{\mathfrak{l}} \dalsym{[}  \dalnt{a_{{\mathrm{0}}}}  \dalsym{/}  \dalmv{x}  \dalsym{]}  \dalnt{C} : \mathsf{Type}
  \end{align*}
  by applying the rule \dalrulename[glad]{weak}.

  \proofitem{Cases \dalrulename[glad]{unit},
    \dalrulename[glad]{unitIntro}}
  Trivial.

  \proofitem{Case \dalrulename[glad]{unitElim}}
  Our derivation has the form
  \[
	\inferrule
    {
      \delta_{{\mathrm{1}}}  \dalsym{,}  \dalnt{r_{{\mathrm{1}}}}  \dalsym{,}  \delta'_{{\mathrm{1}}}  \mid  \mathcal
{M}  \dalsym{,}  \mathfrak{m}_{{\mathrm{0}}}  \dalsym{,}  \mathcal
{M}'   \odot   \Gamma  \dalsym{,}  \dalmv{x}  \colon  \dalnt{A}  \dalsym{,}  \Gamma'   \vdash _{ \mathfrak{m}_{{\mathrm{1}}} }  \dalnt{c}  \colon   \mathbf{I}_{ \mathfrak{m}_{{\mathrm{1}}} } \\
      \delta_{{\mathrm{3}}}  \dalsym{,}  \dalnt{r_{{\mathrm{3}}}}  \dalsym{,}  \delta'_{{\mathrm{3}}}  \dalsym{,}  \dalnt{q}  \mid  \mathcal
{M}  \dalsym{,}  \mathfrak{m}_{{\mathrm{0}}}  \dalsym{,}  \mathcal
{M}'  \dalsym{,}  \mathfrak{m}_{{\mathrm{1}}}   \odot   \Gamma  \dalsym{,}  \dalmv{x}  \colon  \dalnt{A}  \dalsym{,}  \Gamma'  \dalsym{,}  \dalmv{y}  \colon   \mathbf{I}_{ \mathfrak{m}_{{\mathrm{1}}} }    \vdash _{ \mathfrak{m}_{{\mathrm{2}}} }  \dalnt{B}  \colon  \mathsf{Type} \\
      \delta_{{\mathrm{2}}}  \dalsym{,}  \dalnt{r_{{\mathrm{2}}}}  \dalsym{,}  \delta'_{{\mathrm{2}}}  \mid  \mathcal
{M}  \dalsym{,}  \mathfrak{m}_{{\mathrm{0}}}  \dalsym{,}  \mathcal
{M}'   \odot   \Gamma  \dalsym{,}  \dalmv{x}  \colon  \dalnt{A}  \dalsym{,}  \Gamma'   \vdash _{ \mathfrak{m}_{{\mathrm{2}}} }  \dalnt{b}  \colon  \dalsym{[}   \star_{ \mathfrak{m}_{{\mathrm{1}}} }   \dalsym{/}  \dalmv{y}  \dalsym{]}  \dalnt{B}
    }
    {
      \delta_{{\mathrm{1}}}  \dalsym{+}  \delta_{{\mathrm{2}}}  \dalsym{,}  \dalnt{r_{{\mathrm{1}}}}  \dalsym{+}  \dalnt{r_{{\mathrm{2}}}}  \dalsym{,}  \delta'_{{\mathrm{1}}}  \dalsym{+}  \delta'_{{\mathrm{2}}}
      \mid \mathcal
{M}  \dalsym{,}  \mathfrak{m}_{{\mathrm{0}}}  \dalsym{,}  \mathcal
{M}' \at \Gamma  \dalsym{,}  \dalmv{x}  \colon  \dalnt{A}  \dalsym{,}  \Gamma' \\
      \proves_{ \mathfrak{m}_{{\mathrm{1}}} } \operatorname{\mathsf{let} }  \star_{ \mathfrak{m}_{{\mathrm{1}}} } =  \dalnt{c}   \operatorname{\mathsf{in} }   \dalnt{b} : \dalsym{[}  \dalnt{c}  \dalsym{/}  \dalmv{y}  \dalsym{]}  \dalnt{B}
    }
  \]
  and we need to prove
  \begin{align*}
    &
      \delta_{{\mathrm{1}}}  \dalsym{+}  \delta_{{\mathrm{2}}}  \dalsym{+}   \dalsym{(}  \dalnt{r_{{\mathrm{1}}}}  \dalsym{+}  \dalnt{r_{{\mathrm{2}}}}  \dalsym{)}   \cdot   \delta_{{\mathrm{0}}}   \dalsym{,}  \delta'_{{\mathrm{1}}}  \dalsym{+}  \delta'_{{\mathrm{2}}}
      \mid \mathcal
{M}  \dalsym{,}  \mathcal
{M}'
      \at \Gamma  \dalsym{,}  \dalsym{[}  \dalnt{a_{{\mathrm{0}}}}  \dalsym{/}  \dalmv{x}  \dalsym{]}  \Gamma'
    \\
    & \proves_{ \mathfrak{m}_{{\mathrm{1}}} }
      \dalsym{[}  \dalnt{a_{{\mathrm{0}}}}  \dalsym{/}  \dalmv{x}  \dalsym{]}  \dalsym{(}   \operatorname{\mathsf{let} }  \star_{ \mathfrak{m}_{{\mathrm{1}}} } =  \dalnt{c}   \operatorname{\mathsf{in} }   \dalnt{b}   \dalsym{)} : \dalsym{[}  \dalnt{a_{{\mathrm{0}}}}  \dalsym{/}  \dalmv{x}  \dalsym{]}  \dalsym{[}  \dalnt{c}  \dalsym{/}  \dalmv{y}  \dalsym{]}  \dalnt{B}
  \end{align*}
  We have the inductive hypotheses
  \begin{align*}
    &
      \delta_{{\mathrm{1}}}  \dalsym{+}   \dalnt{r_{{\mathrm{1}}}}   \cdot   \delta_{{\mathrm{0}}}   \dalsym{,}  \delta'_{{\mathrm{1}}}  \mid  \mathcal
{M}  \dalsym{,}  \mathcal
{M}'   \odot   \Gamma  \dalsym{,}  \dalsym{[}  \dalnt{a_{{\mathrm{0}}}}  \dalsym{/}  \dalmv{x}  \dalsym{]}  \Gamma'   \vdash _{ \mathfrak{m}_{{\mathrm{1}}} }  \dalsym{[}  \dalnt{a_{{\mathrm{0}}}}  \dalsym{/}  \dalmv{x}  \dalsym{]}  \dalnt{c}  \colon  \dalsym{[}  \dalnt{a_{{\mathrm{0}}}}  \dalsym{/}  \dalmv{x}  \dalsym{]}  \dalnt{C}
    \\
    & \delta_{{\mathrm{2}}}  \dalsym{+}   \dalnt{r_{{\mathrm{2}}}}   \cdot   \delta_{{\mathrm{0}}}   \dalsym{,}  \delta'_{{\mathrm{2}}} \mid \mathcal
{M}  \dalsym{,}  \mathcal
{M}' \at \Gamma  \dalsym{,}  \dalsym{[}  \dalnt{a_{{\mathrm{0}}}}  \dalsym{/}  \dalmv{x}  \dalsym{]}  \Gamma' \\
    & \quad \proves_{ \mathfrak{m}_{{\mathrm{2}}} } \dalsym{[}  \dalnt{a_{{\mathrm{0}}}}  \dalsym{/}  \dalmv{x}  \dalsym{]}  \dalnt{b} : \dalsym{[}  \dalnt{a_{{\mathrm{0}}}}  \dalsym{/}  \dalmv{x}  \dalsym{]}  \dalsym{[}   \star_{ \mathfrak{m}_{{\mathrm{1}}} }   \dalsym{/}  \dalmv{y}  \dalsym{]}  \dalnt{B}
    \\
    & \delta_{{\mathrm{3}}}  \dalsym{+}   \dalnt{r_{{\mathrm{3}}}}   \cdot   \delta_{{\mathrm{0}}}   \dalsym{,}  \delta'_{{\mathrm{3}}} \mid \mathcal
{M}  \dalsym{,}  \mathcal
{M}'  \dalsym{,}  \mathfrak{m}_{{\mathrm{1}}} \at \Gamma  \dalsym{,}  \dalsym{[}  \dalnt{a_{{\mathrm{0}}}}  \dalsym{/}  \dalmv{x}  \dalsym{]}  \Gamma'  \dalsym{,}  \dalmv{y}  \colon   \mathbf{I}_{ \mathfrak{m}_{{\mathrm{1}}} } \\
    & \quad \proves_{ \mathfrak{m}_{{\mathrm{2}}} } \dalsym{[}  \dalnt{a_{{\mathrm{0}}}}  \dalsym{/}  \dalmv{x}  \dalsym{]}  \dalnt{B} : \mathsf{Type}
  \end{align*}
  In the second hypothesis,
  the type $ \dalsym{[}  \dalnt{a_{{\mathrm{0}}}}  \dalsym{/}  \dalmv{x}  \dalsym{]}  \dalsym{[}   \star_{ \mathfrak{m}_{{\mathrm{1}}} }   \dalsym{/}  \dalmv{y}  \dalsym{]}  \dalnt{B} $ is equal to 
  $ \dalsym{[}   \star_{ \mathfrak{m}_{{\mathrm{1}}} }   \dalsym{/}  \dalmv{y}  \dalsym{]}  \dalsym{[}  \dalnt{a_{{\mathrm{0}}}}  \dalsym{/}  \dalmv{x}  \dalsym{]}  \dalnt{B} $ since $ y $ is not free in $ \dalnt{a_{{\mathrm{0}}}} $
  and $ \star_{ \mathfrak{m}_{{\mathrm{1}}} } $ has no free variables.
  Furthermore, the types
  $ \dalsym{[}  \dalnt{a_{{\mathrm{0}}}}  \dalsym{/}  \dalmv{x}  \dalsym{]}  \dalsym{[}  \dalnt{c}  \dalsym{/}  \dalmv{y}  \dalsym{]}  \dalnt{B} $
  and $ \dalsym{[}  \dalsym{[}  \dalnt{a_{{\mathrm{0}}}}  \dalsym{/}  \dalmv{x}  \dalsym{]}  \dalnt{c}  \dalsym{/}  \dalmv{y}  \dalsym{]}  \dalsym{[}  \dalnt{a_{{\mathrm{0}}}}  \dalsym{/}  \dalmv{x}  \dalsym{]}  \dalnt{B} $ are equal.
  Because of this, we may apply the rule \dalrulename[glad]{unitElim}
  to the three inductive hypotheses to obtain the desired judgment.

  \proofitem{Case \dalrulename[glad]{function}}
  Our derivation has the form
  \[
	\inferrule
    {
      \delta  \dalsym{,}  \dalnt{r}  \dalsym{,}  \delta'  \dalsym{,}  \dalnt{q_{{\mathrm{0}}}}  \mid  \mathcal
{M}  \dalsym{,}  \mathfrak{m}_{{\mathrm{0}}}  \dalsym{,}  \mathcal
{M}'  \dalsym{,}  \mathfrak{m}   \odot   \Gamma  \dalsym{,}  \dalmv{x}  \colon  \dalnt{A}  \dalsym{,}  \Gamma'  \dalsym{,}  \dalmv{y}  \colon  \dalnt{B}   \vdash _{ \mathfrak{n} }  \dalnt{C}  \colon  \mathsf{Type} \\
    }
    {
      \delta  \dalsym{,}  \dalnt{r}  \dalsym{,}  \delta'  \dalsym{,}  \dalnt{q_{{\mathrm{0}}}}  \mid  \mathcal
{M}  \dalsym{,}  \mathfrak{m}_{{\mathrm{0}}}  \dalsym{,}  \mathcal
{M}'   \odot   \Gamma  \dalsym{,}  \dalmv{x}  \colon  \dalnt{A}  \dalsym{,}  \Gamma'   \vdash _{ \mathfrak{n} }   (  \dalmv{y}  :^{ \dalnt{q}  :  \mathfrak{m} }  \dalnt{B}  )  \multimap   \dalnt{C}   \colon  \mathsf{Type}
    }
  \]
  and we need to prove
  \begin{align*}
	& \delta  \dalsym{+}   \dalnt{r}   \cdot   \delta_{{\mathrm{0}}}   \dalsym{,}  \delta' \mid \mathcal
{M}  \dalsym{,}  \mathcal
{M}' \at \Gamma  \dalsym{,}  \dalsym{[}  \dalnt{a_{{\mathrm{0}}}}  \dalsym{/}  \dalmv{x}  \dalsym{]}  \Gamma' \\
    & \quad \proves_{\mf n} \dalsym{[}  \dalnt{a_{{\mathrm{0}}}}  \dalsym{/}  \dalmv{x}  \dalsym{]}  \dalsym{(}   (  \dalmv{y}  :^{ \dalnt{q}  :  \mathfrak{m} }  \dalnt{B}  )  \multimap   \dalnt{C}   \dalsym{)} : \mathsf{Type}
  \end{align*}
  Since the types
  $ \dalsym{[}  \dalnt{a_{{\mathrm{0}}}}  \dalsym{/}  \dalmv{x}  \dalsym{]}  \dalsym{(}   (  \dalmv{y}  :^{ \dalnt{q}  :  \mathfrak{m} }  \dalnt{B}  )  \multimap   \dalnt{C}   \dalsym{)}  $
  and
  $ (  \dalmv{y}  :^{ \dalnt{q}  :  \mathfrak{m} }  \dalsym{[}  \dalnt{a_{{\mathrm{0}}}}  \dalsym{/}  \dalmv{x}  \dalsym{]}  \dalnt{B}  )  \multimap   \dalsym{[}  \dalnt{a_{{\mathrm{0}}}}  \dalsym{/}  \dalmv{x}  \dalsym{]}  \dalnt{C}  $
  are equal, this follows from the inductive hypothesis
  \begin{align*}
    & \delta  \dalsym{+}   \dalnt{r}   \cdot   \delta_{{\mathrm{0}}}   \dalsym{,}  \delta'  \dalsym{,}  \dalnt{q_{{\mathrm{0}}}} \mid \mathcal
{M}  \dalsym{,}  \mathcal
{M}'  \dalsym{,}  \mathfrak{m} \\
    & \quad \at \Gamma  \dalsym{,}  \dalsym{[}  \dalnt{a_{{\mathrm{0}}}}  \dalsym{/}  \dalmv{x}  \dalsym{]}  \Gamma'  \dalsym{,}  \dalmv{y}  \colon  \dalsym{[}  \dalnt{a_{{\mathrm{0}}}}  \dalsym{/}  \dalmv{x}  \dalsym{]}  \dalnt{B} 
      \proves_{ \dalmv{n}} \dalsym{[}  \dalnt{a_{{\mathrm{0}}}}  \dalsym{/}  \dalmv{x}  \dalsym{]}  \dalnt{C} : \mathsf{Type}
  \end{align*}

  \proofitem{Case \dalrulename[glad]{lambda}}
  Our derivation has the form
  \[
    \inferrule
    {
      \delta  \dalsym{,}  \dalnt{r}  \dalsym{,}  \delta'  \dalsym{,}  \dalnt{q}  \mid  \mathcal
{M}  \dalsym{,}  \mathfrak{m}_{{\mathrm{0}}}  \dalsym{,}  \mathcal
{M}'  \dalsym{,}  \mathfrak{m}   \odot   \Gamma  \dalsym{,}  \dalmv{x}  \colon  \dalnt{A}  \dalsym{,}  \Gamma'  \dalsym{,}  \dalmv{y}  \colon  \dalnt{B}   \vdash _{ \mathfrak{n} }  \dalnt{c}  \colon  \dalnt{C}
    }
    {
      \delta  \dalsym{,}  \dalnt{r}  \dalsym{,}  \delta'  \mid  \mathcal
{M}  \dalsym{,}  \mathfrak{m}_{{\mathrm{0}}}  \dalsym{,}  \mathcal
{M}'   \odot   \Gamma  \dalsym{,}  \dalmv{x}  \colon  \dalnt{A}  \dalsym{,}  \Gamma'   \vdash _{ \mathfrak{n} }   \lambda  \dalmv{y}  .  \dalnt{c}   \colon   (  \dalmv{y}  :^{ \dalnt{q}  :  \mathfrak{m} }  \dalnt{B}  )  \multimap   \dalnt{C}
    }
  \]
  We need to prove
  \begin{align*}
	& \delta  \dalsym{+}   \dalnt{r}   \cdot   \delta_{{\mathrm{0}}}   \dalsym{,}  \delta'  \mid \mathcal
{M}  \dalsym{,}  \mathcal
{M}' \at \Gamma  \dalsym{,}  \dalsym{[}  \dalnt{a_{{\mathrm{0}}}}  \dalsym{/}  \dalmv{x}  \dalsym{]}  \Gamma' \proves_{ \dalmv{n}}\\
    & \quad \dalsym{[}  \dalnt{a_{{\mathrm{0}}}}  \dalsym{/}  \dalmv{x}  \dalsym{]}  \dalsym{(}   \lambda  \dalmv{y}  .  \dalnt{c}   \dalsym{)} : (  \dalmv{y}  :^{ \dalnt{q}  :  \mathfrak{m} }  \dalsym{[}  \dalnt{a_{{\mathrm{0}}}}  \dalsym{/}  \dalmv{x}  \dalsym{]}  \dalnt{B}  )  \multimap   \dalsym{[}  \dalnt{a_{{\mathrm{0}}}}  \dalsym{/}  \dalmv{x}  \dalsym{]}  \dalnt{C}
  \end{align*}
  But this follows immediately from the inductive hypothesis
  \begin{align*}
	& \delta  \dalsym{+}   \dalnt{r}   \cdot   \delta_{{\mathrm{0}}}   \dalsym{,}  \delta'  \dalsym{,}  \dalnt{q}  \mid \mathcal
{M}  \dalsym{,}  \mathcal
{M}'  \dalsym{,}  \mathfrak{m}
      \at \Gamma  \dalsym{,}  \dalsym{[}  \dalnt{a_{{\mathrm{0}}}}  \dalsym{/}  \dalmv{x}  \dalsym{]}  \Gamma'  \dalsym{,}  \dalmv{y}  \colon  \dalsym{[}  \dalnt{a_{{\mathrm{0}}}}  \dalsym{/}  \dalmv{x}  \dalsym{]}  \dalnt{B} \proves_{ \dalmv{n}}\\
    & \quad \dalsym{[}  \dalnt{a_{{\mathrm{0}}}}  \dalsym{/}  \dalmv{x}  \dalsym{]}  \dalnt{c} : \dalsym{[}  \dalnt{a_{{\mathrm{0}}}}  \dalsym{/}  \dalmv{x}  \dalsym{]}  \dalnt{C}
  \end{align*}
  and the fact that the terms
  $ \dalsym{[}  \dalnt{a_{{\mathrm{0}}}}  \dalsym{/}  \dalmv{x}  \dalsym{]}  \dalsym{(}   \lambda  \dalmv{y}  .  \dalnt{c}   \dalsym{)} $ and $ \lambda  \dalmv{y}  .  \dalsym{[}  \dalnt{a_{{\mathrm{0}}}}  \dalsym{/}  \dalmv{x}  \dalsym{]}  \dalnt{c} $ are equal.

  \proofitem{Case \dalrulename[glad]{app}}
  The derivation has the form
  \[
    \inferrule
    {
      \delta_{{\mathrm{1}}}  \dalsym{,}  \dalnt{r_{{\mathrm{1}}}}  \dalsym{,}  \delta'_{{\mathrm{1}}}  \mid  \mathcal
{M}  \dalsym{,}  \mathfrak{m}_{{\mathrm{0}}}  \dalsym{,}  \mathcal
{M}'   \odot   \Gamma  \dalsym{,}  \dalmv{x}  \colon  \dalnt{A}  \dalsym{,}  \Gamma'   \vdash _{ \mathfrak{m} }  \dalnt{b}  \colon  \dalnt{B}\\
      \delta_{{\mathrm{2}}}  \dalsym{,}  \dalnt{r_{{\mathrm{2}}}}  \dalsym{,}  \delta'_{{\mathrm{2}}}  \mid  \mathcal
{M}  \dalsym{,}  \mathfrak{m}_{{\mathrm{0}}}  \dalsym{,}  \mathcal
{M}'   \odot   \Gamma  \dalsym{,}  \dalmv{x}  \colon  \dalnt{A}  \dalsym{,}  \Gamma'   \vdash _{ \mathfrak{n} }  \dalnt{c}  \colon   (  \dalmv{y}  :^{ \dalnt{q}  :  \mathfrak{m} }  \dalnt{B}  )  \multimap   \dalnt{C}
    }
    {
      \dalsym{(}  \delta_{{\mathrm{1}}}  \dalsym{,}  \dalnt{r_{{\mathrm{1}}}}  \dalsym{,}  \delta'_{{\mathrm{1}}}  \dalsym{)}  \dalsym{+}   \dalnt{q}   \cdot   \dalsym{(}  \delta_{{\mathrm{2}}}  \dalsym{,}  \dalnt{r_{{\mathrm{2}}}}  \dalsym{,}  \delta'_{{\mathrm{2}}}  \dalsym{)}
      \mid \mathcal
{M}  \dalsym{,}  \mathfrak{m}_{{\mathrm{0}}}  \dalsym{,}  \mathcal
{M}' 
      \at \Gamma  \dalsym{,}  \dalmv{x}  \colon  \dalnt{A}  \dalsym{,}  \Gamma'\\
      \proves_{\mf n} \dalnt{c} \, \dalnt{b} : \dalsym{[}  \dalnt{b}  \dalsym{/}  \dalmv{y}  \dalsym{]}  \dalnt{C}.
    }
  \]
  and we need to show that
  \begin{align*}
    & \delta_{{\mathrm{1}}}  \dalsym{+}   \dalnt{q}   \cdot   \delta_{{\mathrm{2}}}   \dalsym{+}   \dalsym{(}   \dalnt{r_{{\mathrm{1}}}}  \dalsym{+}  \dalnt{q}   \cdot   \dalnt{r_{{\mathrm{2}}}}   \dalsym{)}   \cdot   \delta_{{\mathrm{0}}}   \dalsym{,}  \delta'_{{\mathrm{1}}}  \dalsym{+}   \dalnt{q}   \cdot   \delta'_{{\mathrm{2}}} \mid \mathcal
{M}  \dalsym{,}  \mathcal
{M}' \\
    & \quad \at \Gamma  \dalsym{,}  \dalsym{[}  \dalnt{a_{{\mathrm{0}}}}  \dalsym{/}  \dalmv{x}  \dalsym{]}  \Gamma'
      \proves_{\mf n} \dalsym{[}  \dalnt{a_{{\mathrm{0}}}}  \dalsym{/}  \dalmv{x}  \dalsym{]}  \dalsym{(}  \dalnt{c} \, \dalnt{b}  \dalsym{)} : \dalsym{[}  \dalnt{a_{{\mathrm{0}}}}  \dalsym{/}  \dalmv{x}  \dalsym{]}  \dalsym{[}  \dalnt{b}  \dalsym{/}  \dalmv{y}  \dalsym{]}  \dalnt{C}
  \end{align*}
  The inductive hypotheses are
  \begin{align*}
	& \delta_{{\mathrm{1}}}  \dalsym{+}   \dalnt{r_{{\mathrm{1}}}}   \cdot   \delta_{{\mathrm{0}}}   \dalsym{,}  \delta'_{{\mathrm{1}}}  \mid  \mathcal
{M}  \dalsym{,}  \mathcal
{M}'   \odot   \Gamma  \dalsym{,}  \dalsym{[}  \dalnt{a_{{\mathrm{0}}}}  \dalsym{/}  \dalmv{x}  \dalsym{]}  \Gamma'   \vdash _{ \mathfrak{m} }  \dalsym{[}  \dalnt{a_{{\mathrm{0}}}}  \dalsym{/}  \dalmv{x}  \dalsym{]}  \dalnt{b}  \colon  \dalsym{[}  \dalnt{a_{{\mathrm{0}}}}  \dalsym{/}  \dalmv{x}  \dalsym{]}  \dalnt{B} \\
    & \delta_{{\mathrm{2}}}  \dalsym{+}   \dalnt{r_{{\mathrm{2}}}}   \cdot   \delta_{{\mathrm{0}}}   \dalsym{,}  \delta'_{{\mathrm{2}}} \mid \mathcal
{M}  \dalsym{,}  \mathcal
{M}' \at \Gamma  \dalsym{,}  \dalsym{[}  \dalnt{a_{{\mathrm{0}}}}  \dalsym{/}  \dalmv{x}  \dalsym{]}  \Gamma' \\
    & \quad \proves_{\mf n} \dalsym{[}  \dalnt{a_{{\mathrm{0}}}}  \dalsym{/}  \dalmv{x}  \dalsym{]}  \dalnt{c} : (  \dalmv{y}  :^{ \dalnt{q}  :  \mathfrak{m} }  \dalsym{[}  \dalnt{a_{{\mathrm{0}}}}  \dalsym{/}  \dalmv{x}  \dalsym{]}  \dalnt{B}  )  \multimap   \dalsym{[}  \dalnt{a_{{\mathrm{0}}}}  \dalsym{/}  \dalmv{x}  \dalsym{]}  \dalnt{C}
  \end{align*}
  The types
  $ \dalsym{[}  \dalnt{a_{{\mathrm{0}}}}  \dalsym{/}  \dalmv{x}  \dalsym{]}  \dalsym{[}  \dalnt{b}  \dalsym{/}  \dalmv{y}  \dalsym{]}  \dalnt{C} $
  and
  $ \dalsym{[}  \dalsym{[}  \dalnt{a_{{\mathrm{0}}}}  \dalsym{/}  \dalmv{x}  \dalsym{]}  \dalnt{b}  \dalsym{/}  \dalmv{y}  \dalsym{]}  \dalsym{[}  \dalnt{a_{{\mathrm{0}}}}  \dalsym{/}  \dalmv{x}  \dalsym{]}  \dalnt{C} $
  are equal.
  Therefore we may apply the \dalrulename[glad]{app} rule to the inductive hypotheses to obtain
  \begin{align*}
    & \delta_{{\mathrm{1}}}  \dalsym{+}  \delta_{{\mathrm{2}}}  \dalsym{+}   \dalsym{(}   \dalnt{r_{{\mathrm{1}}}}  \dalsym{+}  \dalnt{q}   \cdot   \dalnt{r_{{\mathrm{2}}}}   \dalsym{)}   \cdot   \delta_{{\mathrm{0}}}   \dalsym{,}  \delta'_{{\mathrm{1}}}  \dalsym{+}   \dalnt{q}   \cdot   \delta'_{{\mathrm{2}}} \mid \mathcal
{M}  \dalsym{,}  \mathcal
{M}' \at \Gamma  \dalsym{,}  \dalsym{[}  \dalnt{a_{{\mathrm{0}}}}  \dalsym{/}  \dalmv{x}  \dalsym{]}  \Gamma' \\
    & \quad \proves_{\mf m} \dalsym{[}  \dalnt{a_{{\mathrm{0}}}}  \dalsym{/}  \dalmv{x}  \dalsym{]}  \dalnt{c} \, \dalsym{[}  \dalnt{a_{{\mathrm{0}}}}  \dalsym{/}  \dalmv{x}  \dalsym{]}  \dalnt{b} : \dalsym{[}  \dalsym{[}  \dalnt{a_{{\mathrm{0}}}}  \dalsym{/}  \dalmv{x}  \dalsym{]}  \dalnt{b}  \dalsym{/}  \dalmv{y}  \dalsym{]}  \dalsym{[}  \dalnt{a_{{\mathrm{0}}}}  \dalsym{/}  \dalmv{x}  \dalsym{]}  \dalnt{C}
  \end{align*}
  This is the desired judgment once we replace the type \linebreak
  $ \dalsym{[}  \dalsym{[}  \dalnt{a_{{\mathrm{0}}}}  \dalsym{/}  \dalmv{x}  \dalsym{]}  \dalnt{b}  \dalsym{/}  \dalmv{y}  \dalsym{]}  \dalsym{[}  \dalnt{a_{{\mathrm{0}}}}  \dalsym{/}  \dalmv{x}  \dalsym{]}  \dalnt{C} $
  by
  $ \dalsym{[}  \dalnt{a_{{\mathrm{0}}}}  \dalsym{/}  \dalmv{x}  \dalsym{]}  \dalsym{[}  \dalnt{b}  \dalsym{/}  \dalmv{y}  \dalsym{]}  \dalnt{C} $
  since we observed this equality before.

  \proofitem{Case \dalrulename[glad]{tensorIntro}}
  Our derivation has the form
  \[
    \inferrule
    {
      \delta_{{\mathrm{3}}}  \dalsym{,}  \dalnt{r_{{\mathrm{3}}}}  \dalsym{,}  \delta'_{{\mathrm{3}}}  \mid  \mathcal
{M}  \dalsym{,}  \mathfrak{m}_{{\mathrm{0}}}  \dalsym{,}  \mathcal
{M}'   \odot   \Gamma  \dalsym{,}  \dalmv{x}  \colon  \dalnt{A}  \dalsym{,}  \Gamma'  \dalsym{,}  \dalmv{y}  \colon  \dalnt{B}   \vdash _{ \mathfrak{n} }  \dalnt{C}  \colon  \mathsf{Type} \\
      \delta_{{\mathrm{1}}}  \dalsym{,}  \dalnt{r_{{\mathrm{1}}}}  \dalsym{,}  \delta'_{{\mathrm{1}}}  \mid  \mathcal
{M}  \dalsym{,}  \mathfrak{m}_{{\mathrm{0}}}  \dalsym{,}  \mathcal
{M}'   \odot   \Gamma  \dalsym{,}  \dalmv{x}  \colon  \dalnt{A}  \dalsym{,}  \Gamma'   \vdash _{ \mathfrak{m} }  \dalnt{b}  \colon  \dalnt{B} \\
      \delta_{{\mathrm{2}}}  \dalsym{,}  \dalnt{r_{{\mathrm{2}}}}  \dalsym{,}  \delta'_{{\mathrm{2}}}  \mid  \mathcal
{M}  \dalsym{,}  \mathfrak{m}_{{\mathrm{0}}}  \dalsym{,}  \mathcal
{M}'   \odot   \Gamma  \dalsym{,}  \dalmv{x}  \colon  \dalnt{A}  \dalsym{,}  \Gamma'   \vdash _{ \mathfrak{n} }  \dalnt{c}  \colon  \dalsym{[}  \dalnt{b}  \dalsym{/}  \dalmv{y}  \dalsym{]}  \dalnt{C}
    }
    {
      \dalnt{q}   \cdot   \dalsym{(}  \delta_{{\mathrm{1}}}  \dalsym{,}  \dalnt{r_{{\mathrm{1}}}}  \dalsym{,}  \dalnt{r'_{{\mathrm{1}}}}  \dalsym{)}   \dalsym{+}  \dalsym{(}  \delta_{{\mathrm{2}}}  \dalsym{,}  \dalnt{r_{{\mathrm{2}}}}  \dalsym{,}  \delta'_{{\mathrm{2}}}  \dalsym{)}
      \mid \mathcal
{M}  \dalsym{,}  \mathfrak{m}_{{\mathrm{0}}}  \dalsym{,}  \mathcal
{M}' \\
      \at \Gamma  \dalsym{,}  \dalmv{x}  \colon  \dalnt{A}  \dalsym{,}  \Gamma'
      \proves_{\mf n} \dalsym{(}  \dalnt{b}  \dalsym{,}  \dalnt{c}  \dalsym{)} : (  \dalmv{y}  :^{ \dalnt{q}  :  \mathfrak{m} }  \dalnt{B}  )  \otimes   \dalnt{C}
    }
  \]
  and we need to prove
  \begin{align*}
    & \dalnt{q}   \cdot   \delta_{{\mathrm{1}}}   \dalsym{+}  \delta_{{\mathrm{2}}}  \dalsym{+}   \dalsym{(}   \dalnt{q}   \cdot   \dalnt{r_{{\mathrm{1}}}}   \dalsym{+}  \dalnt{r_{{\mathrm{2}}}}  \dalsym{)}   \cdot   \delta_{{\mathrm{0}}}   \dalsym{,}   \dalnt{q}   \cdot   \delta'_{{\mathrm{1}}}   \dalsym{+}  \delta'_{{\mathrm{2}}}
      \mid \mathcal
{M}  \dalsym{,}  \mathcal
{M}'
    \\ & \quad
      \at \Gamma  \dalsym{,}  \dalsym{[}  \dalnt{a_{{\mathrm{0}}}}  \dalsym{/}  \dalmv{x}  \dalsym{]}  \Gamma' 
      \proves_{\mf n} \dalsym{[}  \dalnt{a_{{\mathrm{0}}}}  \dalsym{/}  \dalmv{x}  \dalsym{]}  \dalsym{(}  \dalnt{b}  \dalsym{,}  \dalnt{c}  \dalsym{)} : \dalsym{[}  \dalnt{a_{{\mathrm{0}}}}  \dalsym{/}  \dalmv{x}  \dalsym{]}  \dalsym{(}   (  \dalmv{y}  :^{ \dalnt{q}  :  \mathfrak{m} }  \dalnt{B}  )  \otimes   \dalnt{C}   \dalsym{)}
  \end{align*}
  We have the inductive hypotheses
  \begin{align*}
    &
      \delta_{{\mathrm{3}}}  \dalsym{+}   \dalnt{r_{{\mathrm{3}}}}   \cdot   \delta_{{\mathrm{0}}}   \dalsym{,}  \delta'_{{\mathrm{3}}}
      \mid \mathcal
{M}  \dalsym{,}  \mathcal
{M}'
      \at \Gamma  \dalsym{,}  \dalsym{[}  \dalnt{a_{{\mathrm{0}}}}  \dalsym{/}  \dalmv{x}  \dalsym{]}  \Gamma'
      \proves_{\mf n} \dalsym{[}  \dalnt{a_{{\mathrm{0}}}}  \dalsym{/}  \dalmv{x}  \dalsym{]}  \dalnt{C} : \mathsf{Type}
    \\ &
      \delta_{{\mathrm{1}}}  \dalsym{+}   \dalnt{r_{{\mathrm{1}}}}   \cdot   \delta_{{\mathrm{0}}}   \dalsym{,}  \delta'_{{\mathrm{1}}}
      \mid \mathcal
{M}  \dalsym{,}  \mathcal
{M}'
      \at \Gamma  \dalsym{,}  \dalsym{[}  \dalnt{a_{{\mathrm{0}}}}  \dalsym{/}  \dalmv{x}  \dalsym{]}  \Gamma'
      \proves_{\mf n} \dalsym{[}  \dalnt{a_{{\mathrm{0}}}}  \dalsym{/}  \dalmv{x}  \dalsym{]}  \dalnt{b} : \dalsym{[}  \dalnt{a_{{\mathrm{0}}}}  \dalsym{/}  \dalmv{x}  \dalsym{]}  \dalnt{B}
    \\ &
      \delta_{{\mathrm{2}}}  \dalsym{+}   \dalnt{r_{{\mathrm{2}}}}   \cdot   \delta_{{\mathrm{0}}}   \dalsym{,}  \delta'_{{\mathrm{2}}}
      \mid \mathcal
{M}  \dalsym{,}  \mathcal
{M}'
      \at \Gamma  \dalsym{,}  \dalsym{[}  \dalnt{a_{{\mathrm{0}}}}  \dalsym{/}  \dalmv{x}  \dalsym{]}  \Gamma'
    \\ & \quad
      \proves_{\mf n} \dalsym{[}  \dalnt{a_{{\mathrm{0}}}}  \dalsym{/}  \dalmv{x}  \dalsym{]}  \dalnt{c} : \dalsym{[}  \dalnt{a_{{\mathrm{0}}}}  \dalsym{/}  \dalmv{x}  \dalsym{]}  \dalsym{[}  \dalnt{b}  \dalsym{/}  \dalmv{y}  \dalsym{]}  \dalnt{C}
  \end{align*}
  The types
  $ \dalsym{[}  \dalnt{a_{{\mathrm{0}}}}  \dalsym{/}  \dalmv{x}  \dalsym{]}  \dalsym{[}  \dalnt{b}  \dalsym{/}  \dalmv{y}  \dalsym{]}  \dalnt{C} $
  and
  $ \dalsym{[}  \dalsym{[}  \dalnt{a_{{\mathrm{0}}}}  \dalsym{/}  \dalmv{x}  \dalsym{]}  \dalnt{b}  \dalsym{/}  \dalmv{y}  \dalsym{]}  \dalsym{[}  \dalnt{a_{{\mathrm{0}}}}  \dalsym{/}  \dalmv{x}  \dalsym{]}  \dalnt{C} $
  are equal.
  We therefore may apply \dalrulename[glad]{tensorIntro}
  to the three inductive hypotheses and obtain
  \begin{align*}
    &
      \dalnt{q}   \cdot   \dalsym{(}  \delta_{{\mathrm{1}}}  \dalsym{+}   \dalnt{r_{{\mathrm{1}}}}   \cdot   \delta_{{\mathrm{0}}}   \dalsym{,}  \delta'_{{\mathrm{1}}}  \dalsym{)}   \dalsym{+}  \dalsym{(}  \delta_{{\mathrm{2}}}  \dalsym{+}   \dalnt{r_{{\mathrm{2}}}}   \cdot   \delta_{{\mathrm{0}}}   \dalsym{,}  \delta'_{{\mathrm{2}}}  \dalsym{)}
      \mid \mathcal
{M}  \dalsym{,}  \mathcal
{M}'
      \at  \Gamma  \dalsym{,}  \dalsym{[}  \dalnt{a_{{\mathrm{0}}}}  \dalsym{/}  \dalmv{x}  \dalsym{]}  \Gamma'
    \\ & \quad
      \proves_{\mf n}
      \dalsym{(}  \dalsym{[}  \dalnt{a_{{\mathrm{0}}}}  \dalsym{/}  \dalmv{x}  \dalsym{]}  \dalnt{b}  \dalsym{,}  \dalsym{[}  \dalnt{a_{{\mathrm{0}}}}  \dalsym{/}  \dalmv{x}  \dalsym{]}  \dalnt{c}  \dalsym{)} :
      (  \dalmv{y}  :^{ \dalnt{q}  :  \mathfrak{m} }  \dalsym{[}  \dalnt{a_{{\mathrm{0}}}}  \dalsym{/}  \dalmv{x}  \dalsym{]}  \dalnt{B}  )  \otimes   \dalsym{[}  \dalnt{a_{{\mathrm{0}}}}  \dalsym{/}  \dalmv{x}  \dalsym{]}  \dalnt{C}
  \end{align*}
  which is the desired judgment.

  \proofitem{Case \dalrulename[glad]{tensorElim}}
  Our derivationhas the form
  \[
    \inferrule
    {
      \delta_{{\mathrm{3}}}  \dalsym{,}  \dalnt{r_{{\mathrm{3}}}}  \dalsym{,}  \delta'_{{\mathrm{3}}}  \dalsym{,}  \dalnt{r'_{{\mathrm{3}}}}
      \mid \mathcal
{M}  \dalsym{,}  \mathfrak{m}_{{\mathrm{0}}}  \dalsym{,}  \mathcal
{M}'  \dalsym{,}  \mathfrak{m}_{{\mathrm{2}}} \\
      \at \Gamma  \dalsym{,}  \dalmv{x}  \colon  \dalnt{A}  \dalsym{,}  \Gamma'  \dalsym{,}  \dalmv{z}  \colon   (  \dalmv{y}  :^{ \dalnt{q}  :  \mathfrak{m}_{{\mathrm{1}}} }  \dalnt{B_{{\mathrm{1}}}}  )  \otimes   \dalnt{B_{{\mathrm{2}}}}
      \proves_{\mf l} \dalnt{C} : \mathsf{Type}
      \\\\
      \\\\
      \delta_{{\mathrm{1}}}  \dalsym{,}  \dalnt{r_{{\mathrm{1}}}}  \dalsym{,}  \delta'_{{\mathrm{1}}}
      \mid \mathcal
{M}  \dalsym{,}  \mathfrak{m}_{{\mathrm{0}}}  \dalsym{,}  \mathcal
{M}'
      \at \Gamma  \dalsym{,}  \dalmv{x}  \colon  \dalnt{A}  \dalsym{,}  \Gamma'
      \proves_{\mf m _ 2} \dalnt{b} : (  \dalmv{y}  :^{ \dalnt{q}  :  \mathfrak{m}_{{\mathrm{1}}} }  \dalnt{B_{{\mathrm{1}}}}  )  \otimes   \dalnt{B_{{\mathrm{2}}}}
      \\\\
      \\\\
      \delta_{{\mathrm{2}}}  \dalsym{,}  \dalnt{r_{{\mathrm{2}}}}  \dalsym{,}  \delta'_{{\mathrm{2}}}  \dalsym{,}   \dalnt{s}   \cdot   \dalnt{q}   \dalsym{,}  \dalnt{s}
      \mid \mathcal
{M}  \dalsym{,}  \mathfrak{m}_{{\mathrm{0}}}  \dalsym{,}  \mathcal
{M}'  \dalsym{,}  \mathfrak{m}_{{\mathrm{1}}}  \dalsym{,}  \mathfrak{m}_{{\mathrm{2}}}
      \\
      \at \Gamma  \dalsym{,}  \dalmv{x}  \colon  \dalnt{A}  \dalsym{,}  \Gamma'  \dalsym{,}  \dalmv{y_{{\mathrm{1}}}}  \colon  \dalnt{B_{{\mathrm{1}}}}  \dalsym{,}  \dalmv{y_{{\mathrm{2}}}}  \colon  \dalnt{B_{{\mathrm{2}}}}
      \proves_{\mf l} \dalnt{c} : \dalsym{[}  \dalsym{(}  \dalmv{y_{{\mathrm{1}}}}  \dalsym{,}  \dalmv{y_{{\mathrm{2}}}}  \dalsym{)}  \dalsym{/}  \dalmv{z}  \dalsym{]}  \dalnt{C}
    }
    {
      \dalnt{s}   \cdot   \dalsym{(}  \delta_{{\mathrm{1}}}  \dalsym{,}  \dalnt{r_{{\mathrm{1}}}}  \dalsym{,}  \delta'_{{\mathrm{1}}}  \dalsym{)}   \dalsym{+}  \dalsym{(}  \delta_{{\mathrm{2}}}  \dalsym{,}  \dalnt{r_{{\mathrm{2}}}}  \dalsym{,}  \delta'_{{\mathrm{2}}}  \dalsym{)}
      \mid \mathcal
{M}  \dalsym{,}  \mathfrak{m}_{{\mathrm{0}}}  \dalsym{,}  \mathcal
{M}'
      \at \Gamma  \dalsym{,}  \dalmv{x}  \colon  \dalnt{A}  \dalsym{,}  \Gamma' \\
      \proves_{\mf l}
      \operatorname{\mathsf{let} } \, \dalsym{(}  \dalmv{y_{{\mathrm{1}}}}  \dalsym{,}  \dalmv{y_{{\mathrm{2}}}}  \dalsym{)}  \dalsym{=}  \dalnt{b} \, \operatorname{\mathsf{in} } \, \dalnt{c} : \dalsym{[}  \dalnt{b}  \dalsym{/}  \dalmv{z}  \dalsym{]}  \dalnt{C}
    }
  \]
  We need to prove
  \begin{align*}
    &
      \dalnt{s}   \cdot   \delta_{{\mathrm{1}}}   \dalsym{+}  \delta_{{\mathrm{2}}}  \dalsym{+}   \dalsym{(}   \dalnt{s}   \cdot   \dalnt{r_{{\mathrm{1}}}}   \dalsym{+}  \dalnt{r_{{\mathrm{2}}}}  \dalsym{)}   \cdot   \delta_{{\mathrm{0}}}   \dalsym{,}   \dalnt{s}   \cdot   \delta'_{{\mathrm{1}}}   \dalsym{+}  \delta'_{{\mathrm{2}}}
      \mid \mathcal
{M}  \dalsym{,}  \mathcal
{M}'
      \at \Gamma  \dalsym{,}  \dalsym{[}  \dalnt{a_{{\mathrm{0}}}}  \dalsym{/}  \dalmv{x}  \dalsym{]}  \Gamma'
    \\ & \quad
      \proves_{\mf l}
      \dalsym{[}  \dalnt{a_{{\mathrm{0}}}}  \dalsym{/}  \dalmv{x}  \dalsym{]}  \dalsym{(}  \operatorname{\mathsf{let} } \, \dalsym{(}  \dalmv{y_{{\mathrm{1}}}}  \dalsym{,}  \dalmv{y_{{\mathrm{2}}}}  \dalsym{)}  \dalsym{=}  \dalnt{b} \, \operatorname{\mathsf{in} } \, \dalnt{c}  \dalsym{)} : \dalsym{[}  \dalnt{a_{{\mathrm{0}}}}  \dalsym{/}  \dalmv{x}  \dalsym{]}  \dalsym{[}  \dalnt{b}  \dalsym{/}  \dalmv{z}  \dalsym{]}  \dalnt{C}
  \end{align*}
  The inductive hypotheses are
  \begin{align*}
    &
      \delta_{{\mathrm{3}}}  \dalsym{+}   \dalnt{r_{{\mathrm{3}}}}   \cdot   \delta_{{\mathrm{0}}}   \dalsym{,}  \delta'_{{\mathrm{3}}}  \dalsym{,}  \dalnt{r'_{{\mathrm{3}}}}
      \mid \mathcal
{M}  \dalsym{,}  \mathcal
{M}'  \dalsym{,}  \mathfrak{m}_{{\mathrm{2}}}
    \\ & \quad
         \at  \Gamma  \dalsym{,}  \dalsym{[}  \dalnt{a_{{\mathrm{0}}}}  \dalsym{/}  \dalmv{x}  \dalsym{]}  \Gamma'  \dalsym{,}  \dalmv{z}  \colon   (  \dalmv{y}  :^{ \dalnt{q}  :  \mathfrak{m}_{{\mathrm{1}}} }  \dalsym{[}  \dalnt{a_{{\mathrm{0}}}}  \dalsym{/}  \dalmv{z}  \dalsym{]}  \dalnt{B_{{\mathrm{1}}}}  )  \otimes   \dalsym{[}  \dalnt{a_{{\mathrm{0}}}}  \dalsym{/}  \dalmv{x}  \dalsym{]}  \dalnt{B_{{\mathrm{2}}}}
    \\ & \quad
         \proves_{\mf l} \dalsym{[}  \dalnt{a_{{\mathrm{0}}}}  \dalsym{/}  \dalmv{x}  \dalsym{]}  \dalnt{C} : \mathsf{Type} 
    \\ &
         \delta_{{\mathrm{1}}}  \dalsym{+}   \dalnt{r_{{\mathrm{1}}}}   \cdot   \delta_{{\mathrm{0}}}   \dalsym{,}  \delta'_{{\mathrm{1}}}
         \mid \mathcal
{M}  \dalsym{,}  \mathcal
{M}'
         \at \Gamma  \dalsym{,}  \dalsym{[}  \dalnt{a_{{\mathrm{0}}}}  \dalsym{/}  \dalmv{x}  \dalsym{]}  \Gamma'
    \\ & \quad
         \proves_{\mf m2}
         \dalsym{[}  \dalnt{a_{{\mathrm{0}}}}  \dalsym{/}  \dalmv{x}  \dalsym{]}  \dalnt{b} : (  \dalmv{y}  :^{ \dalnt{q}  :  \mathfrak{m}_{{\mathrm{1}}} }  \dalsym{[}  \dalnt{a_{{\mathrm{0}}}}  \dalsym{/}  \dalmv{x}  \dalsym{]}  \dalnt{B_{{\mathrm{1}}}}  )  \otimes   \dalsym{[}  \dalnt{a_{{\mathrm{0}}}}  \dalsym{/}  \dalmv{x}  \dalsym{]}  \dalnt{B_{{\mathrm{2}}}}
    \\ &
         \delta_{{\mathrm{2}}}  \dalsym{+}   \dalnt{r_{{\mathrm{2}}}}   \cdot   \delta_{{\mathrm{0}}}   \dalsym{,}  \delta'_{{\mathrm{2}}}  \dalsym{,}   \dalnt{s}   \cdot   \dalnt{q}   \dalsym{,}  \dalnt{s}
         \mid \mathcal
{M}  \dalsym{,}  \mathcal
{M}'  \dalsym{,}  \mathfrak{m}_{{\mathrm{1}}}  \dalsym{,}  \mathfrak{m}_{{\mathrm{2}}}
    \\ & \quad
         \at \Gamma  \dalsym{,}  \dalsym{[}  \dalnt{a_{{\mathrm{0}}}}  \dalsym{/}  \dalmv{x}  \dalsym{]}  \Gamma'  \dalsym{,}  \dalmv{y_{{\mathrm{1}}}}  \colon  \dalsym{[}  \dalnt{a_{{\mathrm{0}}}}  \dalsym{/}  \dalmv{x}  \dalsym{]}  \dalnt{B_{{\mathrm{1}}}}  \dalsym{,}  \dalmv{y_{{\mathrm{2}}}}  \colon  \dalsym{[}  \dalnt{a_{{\mathrm{0}}}}  \dalsym{/}  \dalmv{x}  \dalsym{]}  \dalnt{B_{{\mathrm{2}}}}
    \\ & \quad
         \proves_{\mf l}
         \dalsym{[}  \dalnt{a_{{\mathrm{0}}}}  \dalsym{/}  \dalmv{x}  \dalsym{]}  \dalnt{c} : \dalsym{[}  \dalnt{a_{{\mathrm{0}}}}  \dalsym{/}  \dalmv{x}  \dalsym{]}  \dalsym{[}  \dalsym{(}  \dalmv{y_{{\mathrm{1}}}}  \dalsym{,}  \dalmv{y_{{\mathrm{2}}}}  \dalsym{)}  \dalsym{/}  \dalmv{z}  \dalsym{]}  \dalnt{C}
  \end{align*}
  The types
  $ \dalsym{[}  \dalnt{a_{{\mathrm{0}}}}  \dalsym{/}  \dalmv{x}  \dalsym{]}  \dalsym{[}  \dalsym{(}  \dalmv{y_{{\mathrm{1}}}}  \dalsym{,}  \dalmv{y_{{\mathrm{2}}}}  \dalsym{)}  \dalsym{/}  \dalmv{z}  \dalsym{]}  \dalnt{C} $
  and
  $ \dalsym{[}  \dalsym{(}  \dalmv{y_{{\mathrm{1}}}}  \dalsym{,}  \dalmv{y_{{\mathrm{2}}}}  \dalsym{)}  \dalsym{/}  \dalmv{z}  \dalsym{]}  \dalsym{[}  \dalnt{a_{{\mathrm{0}}}}  \dalsym{/}  \dalmv{x}  \dalsym{]}  \dalnt{C} $
  are equal.
  We may therefore apply
  \dalrulename[glad]{tensorElim}
  to the three inductive hypotheses and obtain
  \begin{align*}
    &
      \dalnt{s}   \cdot   \dalsym{(}  \delta_{{\mathrm{1}}}  \dalsym{+}   \dalnt{r_{{\mathrm{1}}}}   \cdot   \delta_{{\mathrm{0}}}   \dalsym{,}  \delta'_{{\mathrm{1}}}  \dalsym{)}   \dalsym{+}  \dalsym{(}  \delta_{{\mathrm{2}}}  \dalsym{+}   \dalnt{r_{{\mathrm{2}}}}   \cdot   \delta_{{\mathrm{0}}}   \dalsym{,}  \delta'_{{\mathrm{2}}}  \dalsym{)}
      \mid \mathcal
{M}  \dalsym{,}  \mathcal
{M}'
      \at \Gamma  \dalsym{,}  \dalsym{[}  \dalnt{a_{{\mathrm{0}}}}  \dalsym{/}  \dalmv{x}  \dalsym{]}  \Gamma'
    \\ & \quad
         \proves_{\mf l}
         \operatorname{\mathsf{let} } \, \dalsym{(}  \dalmv{y_{{\mathrm{1}}}}  \dalsym{,}  \dalmv{y_{{\mathrm{2}}}}  \dalsym{)}  \dalsym{=}  \dalsym{[}  \dalnt{a_{{\mathrm{0}}}}  \dalsym{/}  \dalmv{x}  \dalsym{]}  \dalnt{b} \, \operatorname{\mathsf{in} } \, \dalsym{[}  \dalnt{a_{{\mathrm{0}}}}  \dalsym{/}  \dalmv{x}  \dalsym{]}  \dalnt{c} :
         \dalsym{[}  \dalsym{[}  \dalnt{a_{{\mathrm{0}}}}  \dalsym{/}  \dalmv{x}  \dalsym{]}  \dalnt{b}  \dalsym{/}  \dalmv{z}  \dalsym{]}  \dalsym{[}  \dalnt{a_{{\mathrm{0}}}}  \dalsym{/}  \dalmv{x}  \dalsym{]}  \dalnt{C}
  \end{align*}
  which is the desired judgment once we observe that the types
  $ \dalsym{[}  \dalsym{[}  \dalnt{a_{{\mathrm{0}}}}  \dalsym{/}  \dalmv{x}  \dalsym{]}  \dalnt{b}  \dalsym{/}  \dalmv{z}  \dalsym{]}  \dalsym{[}  \dalnt{a_{{\mathrm{0}}}}  \dalsym{/}  \dalmv{x}  \dalsym{]}  \dalnt{C} $
  and 
  $ \dalsym{[}  \dalnt{a_{{\mathrm{0}}}}  \dalsym{/}  \dalmv{x}  \dalsym{]}  \dalsym{[}  \dalnt{b}  \dalsym{/}  \dalmv{z}  \dalsym{]}  \dalnt{C} $ are equal.
  
  \proofitem{Case \dalrulename[glad]{raise}}
  Our derivation has the form
  \[
	\inferrule
    {
      \mathfrak{m}_{{\mathrm{1}}}  \leq  \mathfrak{m}_{{\mathrm{2}}} \\
      \mathfrak{m}_{{\mathrm{2}}}  \leq  \dalsym{(}  \mathcal
{M}  \dalsym{,}  \mathfrak{m}_{{\mathrm{0}}}  \dalsym{,}  \mathcal
{M}'  \dalsym{)} \\
      \delta  \dalsym{,}  \dalnt{r}  \dalsym{,}  \delta'  \mid  \mathcal
{M}  \dalsym{,}  \mathfrak{m}_{{\mathrm{0}}}  \dalsym{,}  \mathcal
{M}'   \odot   \Gamma  \dalsym{,}  \dalmv{x}  \colon  \dalnt{A}  \dalsym{,}  \Gamma'   \vdash _{ \mathfrak{m}_{{\mathrm{1}}} }  \dalnt{b}  \colon  \dalnt{B}
    }
    {
      \delta  \dalsym{,}  \dalnt{r}  \dalsym{,}  \delta'  \mid  \mathcal
{M}  \dalsym{,}  \mathfrak{m}_{{\mathrm{0}}}  \dalsym{,}  \mathcal
{M}'   \odot   \Gamma  \dalsym{,}  \dalmv{x}  \colon  \dalnt{A}  \dalsym{,}  \Gamma'   \vdash _{ \mathfrak{m}_{{\mathrm{2}}} }   \uparrow_{ \mathfrak{m}_{{\mathrm{1}}} }^{ \mathfrak{m}_{{\mathrm{2}}} }\!\!  \dalnt{b}   \colon   \uparrow_{ \mathfrak{m}_{{\mathrm{1}}} }^{ \mathfrak{m}_{{\mathrm{2}}} }\!\!  \dalnt{B}
    }
  \]
  We need to show
  \begin{align*}
	& \delta  \dalsym{+}   \dalnt{r}   \cdot   \delta_{{\mathrm{0}}}   \dalsym{,}  \delta'
    \mid \mathcal
{M}  \dalsym{,}  \mathcal
{M}' \at
    \Gamma  \dalsym{,}  \dalsym{[}  \dalnt{a_{{\mathrm{0}}}}  \dalsym{/}  \dalmv{x}  \dalsym{]}  \Gamma' \\
    & \quad \proves_{ \mathfrak{m}_{{\mathrm{2}}} }
    \uparrow_{ \mathfrak{m}_{{\mathrm{1}}} }^{ \mathfrak{m}_{{\mathrm{2}}} }\!\!  \dalsym{[}  \dalnt{a_{{\mathrm{0}}}}  \dalsym{/}  \dalmv{x}  \dalsym{]}  \dalnt{b} : \uparrow_{ \mathfrak{m}_{{\mathrm{1}}} }^{ \mathfrak{m}_{{\mathrm{2}}} }\!\!  \dalsym{[}  \dalnt{a_{{\mathrm{0}}}}  \dalsym{/}  \dalmv{x}  \dalsym{]}  \dalnt{B}
  \end{align*}
  The inductive hypothesis is
  \begin{align*}
	& \delta  \dalsym{+}   \dalnt{r}   \cdot   \delta_{{\mathrm{0}}}   \dalsym{,}  \delta'
    \mid \mathcal
{M}  \dalsym{,}  \mathcal
{M}' \at
    \Gamma  \dalsym{,}  \dalsym{[}  \dalnt{a_{{\mathrm{0}}}}  \dalsym{/}  \dalmv{x}  \dalsym{]}  \Gamma' \\
    & \quad \proves_{ \mathfrak{m}_{{\mathrm{1}}} }
    \dalsym{[}  \dalnt{a_{{\mathrm{0}}}}  \dalsym{/}  \dalmv{x}  \dalsym{]}  \dalnt{b} : \dalsym{[}  \dalnt{a_{{\mathrm{0}}}}  \dalsym{/}  \dalmv{x}  \dalsym{]}  \dalnt{B}
  \end{align*}
  Furthermore, $ \mathfrak{m}_{{\mathrm{2}}}  \leq  \dalsym{(}  \mathcal
{M}  \dalsym{,}  \mathcal
{M}'  \dalsym{)} $ follows from $ \mathfrak{m}_{{\mathrm{2}}}  \leq  \dalsym{(}  \mathcal
{M}  \dalsym{,}  \mathfrak{m}_{{\mathrm{0}}}  \dalsym{,}  \mathcal
{M}'  \dalsym{)} $.
  Therefore we can apply \dalrulename[glad]{raise} to obtain the desired result.

  \proofitem{Case \dalrulename[glad]{unraise}}
  In this case, our derivation has the form
  \[
	\inferrule
    {
      \delta  \dalsym{,}  \dalnt{r}  \dalsym{,}  \delta'  \mid  \mathcal
{M}  \dalsym{,}  \mathfrak{m}_{{\mathrm{0}}}  \dalsym{,}  \mathcal
{M}'   \odot   \Gamma  \dalsym{,}  \dalmv{x}  \colon  \dalnt{A}  \dalsym{,}  \Gamma'   \vdash _{ \mathfrak{m}_{{\mathrm{2}}} }  \dalnt{b}  \colon   \uparrow_{ \mathfrak{m}_{{\mathrm{1}}} }^{ \mathfrak{m}_{{\mathrm{2}}} }\!\!  \dalnt{B}
    }
    {
      \delta  \dalsym{,}  \dalnt{r}  \dalsym{,}  \delta'  \mid  \mathcal
{M}  \dalsym{,}  \mathfrak{m}_{{\mathrm{0}}}  \dalsym{,}  \mathcal
{M}'   \odot   \Gamma  \dalsym{,}  \dalmv{x}  \colon  \dalnt{A}  \dalsym{,}  \Gamma'   \vdash _{ \mathfrak{m}_{{\mathrm{1}}} }   \downarrow_{ \mathfrak{m}_{{\mathrm{1}}} }^{ \mathfrak{m}_{{\mathrm{2}}} }\!\!  \dalnt{b}   \colon  \dalnt{B}
    }
  \]
  We need to show
  \begin{align*}
	& \delta  \dalsym{+}   \dalnt{r}   \cdot   \delta_{{\mathrm{0}}}   \dalsym{,}  \delta'
    \mid \mathcal
{M}  \dalsym{,}  \mathcal
{M}' \at
    \Gamma  \dalsym{,}  \dalsym{[}  \dalnt{a_{{\mathrm{0}}}}  \dalsym{/}  \dalmv{x}  \dalsym{]}  \Gamma' \\
    & \quad \proves_{ \mathfrak{m}_{{\mathrm{1}}} }
    \downarrow_{ \mathfrak{m}_{{\mathrm{1}}} }^{ \mathfrak{m}_{{\mathrm{2}}} }\!\!  \dalsym{[}  \dalnt{a_{{\mathrm{0}}}}  \dalsym{/}  \dalmv{x}  \dalsym{]}  \dalnt{b} : \dalsym{[}  \dalnt{a_{{\mathrm{0}}}}  \dalsym{/}  \dalmv{x}  \dalsym{]}  \dalnt{B}
  \end{align*}
  which follows immediately from the inductive hypothesis
  \begin{align*}
	& \delta  \dalsym{+}   \dalnt{r}   \cdot   \delta_{{\mathrm{0}}}   \dalsym{,}  \delta'
    \mid \mathcal
{M}  \dalsym{,}  \mathcal
{M}' \at
    \Gamma  \dalsym{,}  \dalsym{[}  \dalnt{a_{{\mathrm{0}}}}  \dalsym{/}  \dalmv{x}  \dalsym{]}  \Gamma' \\
    & \quad \proves_{ \mathfrak{m}_{{\mathrm{2}}} }
    \dalsym{[}  \dalnt{a_{{\mathrm{0}}}}  \dalsym{/}  \dalmv{x}  \dalsym{]}  \dalnt{b} : \uparrow_{ \mathfrak{m}_{{\mathrm{1}}} }^{ \mathfrak{m}_{{\mathrm{2}}} }\!\!  \dalsym{[}  \dalnt{a_{{\mathrm{0}}}}  \dalsym{/}  \dalmv{x}  \dalsym{]}  \dalnt{B}
  \end{align*}
  by applying \dalrulename[glad]{unraise}.

  The remaining cases are similar and we omit them.
  \qedhere
\end{proof}

\end{document}